%% file: main.tex
\documentclass[11pt,oneside,reqno]{amsart}
\usepackage{graphicx} 

\usepackage[T1]{fontenc}
\usepackage[latin9]{inputenc}
\usepackage[a4paper]{geometry}
\usepackage{color}
\usepackage{float}
\usepackage{amsbsy}
\usepackage{amstext}
\usepackage{amsmath,amssymb,amsthm}
\usepackage{enumitem}
\usepackage{dsfont}
\usepackage{subcaption}
\usepackage{comment}
\usepackage{threeparttable}

\usepackage{tikz}
\usetikzlibrary{shapes.geometric, arrows.meta,positioning,calc} 

\usepackage{booktabs}
\usepackage{longtable}
\usepackage{setspace}
\usepackage[unicode=true,pdfusetitle,
 bookmarks=false,bookmarksnumbered=false,bookmarksopen=false,
 breaklinks=true,pdfborder={0 0 0},pdfborderstyle={},
 colorlinks=true]
 {hyperref}
\usepackage[noabbrev,nameinlink,capitalise]{cleveref}

\makeatletter

\usepackage{tabularx}
\theoremstyle{plain}
\newtheorem{assumption}{\protect\assumptionname}
\newtheorem{remark}{Remark}
        \let\oldremark\remark
        \let\endoldremark\endremark
        \renewenvironment{remark}[1][]{%
        \pushQED{$\blacksquare$}
        \oldremark[#1]\normalfont
        }{
        \popQED\endoldremark
    }
\theoremstyle{definition}
\newtheorem{definition}{\protect\definitionname}
\theoremstyle{plain}
\newtheorem{theorem}{\protect\theoremname}
\theoremstyle{plain}
\newtheorem{lemma}{\protect\lemmaname}
\theoremstyle{plain}
\newtheorem{corollary}{\protect\corollaryname}
\theoremstyle{plain}
\newtheorem{example}{\protect\examplename}
\theoremstyle{plain}
\newtheorem{condition}{Condition}

\crefname{table}{Table}{Tables}
\crefname{subtable}{Table}{Tables}
\crefname{part}{Part}{Parts}
\crefname{chapter}{Chapter}{Chapters}
\crefname{section}{Section}{Sections}
\crefname{subsection}{Section}{Sections}
\crefname{subsubsection}{Section}{Sections}
\crefname{appendix}{Appendix}{Appendices}
\crefname{condition}{Condition}{Conditions}
\crefname{theorem}{Theorem}{Theorems}
\crefname{lemma}{Lemma}{Lemmas}
\crefname{algorithm}{Algorithm}{Algorithms}
\crefname{listing}{Listing}{Listings}
\crefname{figure}{Figure}{Figures}
\crefname{assumption}{Assumption}{Assumptions}
\crefname{equation}{}{}
\crefformat{equation}{(#2#1#3)}

\usepackage[natbibapa]{apacite}
\usepackage{bibunits}
\defaultbibliographystyle{apecon}

\usepackage{xcolor}
\definecolor{darkblue}{rgb}{0.0, 0.0, 0.55}
\definecolor{teal}{rgb}{0.0, 0.5, 0.5}
\hypersetup{linkcolor=teal, urlcolor=blue, citecolor=darkblue}

\renewcommand{\bar}[1]{\overline{#1}}
\renewcommand{\tilde}[1]{\widetilde{#1}}
\renewcommand{\hat}[1]{\widehat{#1}}

\newcommand{\bR}{\boldsymbol{R}}
\newcommand{\bd}{\boldsymbol{d}}
\newcommand{\bD}{\boldsymbol{D}}

\newcommand{\bA}{\boldsymbol{A}}
\newcommand{\bC}{\boldsymbol{C}}

\newcommand{\indep}{\mathrel{\perp\mspace{-10mu}\perp}}

\newcommand{\causal}{\mathrm{causal}}
\newcommand{\causample}{\mathrm{causal,sample}}

\newcommand{\Var}{\operatorname{Var}}
\newcommand{\Cov}{\operatorname{Cov}}
\newcommand{\Lip}{\operatorname{Lip}}
\renewcommand{\ref}[1]{\cref{#1}}

\tikzset{
  startstop/.style = {rectangle, rounded corners, draw, fill=gray!10, align=center, text width=4cm, minimum height=1cm, font=\small},
  decision/.style  = {diamond, draw, align=center, text width=4cm, inner sep=2pt, font=\small, fill=orange!10,aspect=2},
  process/.style   = {rectangle, draw, fill=blue!5, align=center, text width=4cm, minimum height=1cm, font=\small},
  arrow/.style     = {thick, -{Stealth}}
}

\usepackage{soul}

\usepackage{threeparttable}

\usepackage{array}
\newcolumntype{H}{>{\iffalse}c<{\fi}@{}}

\makeatother

\providecommand{\assumptionname}{Assumption}
\providecommand{\corollaryname}{Corollary}
\providecommand{\definitionname}{Definition}
\providecommand{\lemmaname}{Lemma}

\providecommand{\theoremname}{Theorem}
\providecommand{\examplename}{Example}

\begin{document}

\title{Design-Based and Network Sampling-Based Uncertainties in Network Experiments}
\thanks{We are grateful to Jack Porter for his invaluable advice on this project. We also thank Kirill Borusyak, Yong Cai, Harold Chiang, Max Cytrynbaum, Bruce Hansen, Tadao Hoshino, Laura Schechter, Xiaoxia Shi, Xun Tang, Kohei Yata, Ruonan Xu, and participants in seminars in Wisconsin and the World Congress of the Econometric Society for their valuable comments.
This work is supported by the Summer Research Fellowship from the University of Wisconsin-Madison.}
\author{Kensuke Sakamoto and Yuya Shimizu}
\address{Department of Economics, University of Wisconsin, Madison. 1180 Observatory Drive, Madison, WI 53706-1393, USA.}
\email{\href{mailto:}{ksakamoto2@wisc.edu}}
\address{Department of Economics, University of Wisconsin, Madison. 1180 Observatory Drive, Madison, WI 53706-1393, USA.}
\email{\href{mailto:}{yuya.shimizu@wisc.edu}}
\date{\today}

\begin{abstract}   

    Ordinary least squares (OLS) estimators are widely used in network experiments to estimate spillover effects. We study the causal interpretation of, and inference for the OLS estimator under both design-based uncertainty from random treatment assignment and sampling-based uncertainty in network links. We show that correlations among regressors that capture the exposure to neighbors' treatments can induce contamination bias, preventing OLS from aggregating heterogeneous spillover effects for a clear causal interpretation. We derive the OLS estimator's asymptotic distribution and propose a network-robust variance estimator. Simulations and an empirical application demonstrate that contamination bias can be substantial, leading to inflated spillover estimates.

\end{abstract}

\maketitle 



\begin{bibunit}[apecon]

\section{Introduction}
\label{sec:intro}
Network experiments, or randomized controlled trials (RCTs) on networks, have become increasingly common in applied economics (e.g., \citealp{cai2015social}; \citealp{dizon2020}; \citealp{carter2021subsidies}; \citealp{fernando2021}; \citealp{Beaman2021-hw}). A central objective of these experiments is to estimate the ``spillover effect'' of policy interventions as they propagate through networks. For example, \cite{cai2015social} estimate spillover effects from randomly assigned information sessions on rice farmers' decisions to purchase a weather insurance product in Chinese villages. In this paper, we develop a comprehensive theoretical framework for ordinary least squares (OLS) estimators in network experiments, explicitly accounting for both design-based uncertainty, arising from randomness in treatment assignment, and sampling-based uncertainty, arising from randomness in sampling units and network links. Our theory is motivated by two key gaps between empirical practice in applied work and existing econometric theory.

The first gap lies in the choice of estimator.
In applications, researchers predominantly use OLS estimators to estimate spillover effects, employing exposure mappings that summarize treatment status and network structure. In our survey of 29 papers analyzing network experiments, published in the ``top 5'' economics journals and two leading field journals, all of the studies report using the OLS estimator, while only two papers use propensity score-based estimators.\footnote{Specifically, we considered papers published from April 2010 through April 2025 in the following journals: American Economic Review, Econometrica, Quarterly Journal of Economics, Journal of Political Economy, Review of Economic Studies, American Economic Journal: Applied Economics, and Journal of Development Economics. We searched for articles that listed ``networks'' and either ``field experiments'' or ``randomized trial'' as keywords on the Web of Science platform. This search resulted in 52 papers, of which 29 conducted network experiments and are mentioned in the text. These papers are referenced in \cref{app:survey}.} This pattern stands in contrast to the theoretical literature on inference in network experiments (e.g., \citealp{aronow2017estimating}; \citealp{leung2022causal}; \citealp{gao2023causal}), which provides inference results for inverse probability weighting (IPW) estimators that directly estimate average spillover effects.

The other gap is due to ignoring a source of randomness.
In many applied cases, researchers need to collect network information through surveys. This collection process can introduce an extra layer of uncertainty beyond design-based uncertainty.
Moreover, the collected network may only partially capture the true network governing the propagation mechanism. By contrast, the theoretical literature on causal inference in network experiments typically abstracts away from sampling-based uncertainty, assuming that the data correspond to the entire population and that the observed network is complete.

To address these gaps, we make three contributions. First, we develop a novel framework that jointly incorporates design-based randomness in treatment assignment and sampling-based randomness in network links. Our framework considers a finite population of $n$ units, from which units are randomly sampled and treatments are assigned. We explicitly model the network sampling process, focusing on two common sampling methods: (i) induced subgraph sampling, where each sampled unit reports friends within the sample, and (ii) star sampling, where each sampled unit reports friends from the entire population. In this setup, unlike in non-network experiments, sampling-based uncertainty arises from two sources: (i) which units are sampled, and (ii) which links are observed. 
We consider potential outcomes that depend on the entire treatment vector, thus violating the Stable Unit Treatment Value Assumption (SUTVA). To address the resulting dimensionality problem, we assume that the potential outcomes are linear in an exposure mapping, a set of user-specified sufficient statistics summarizing treatment status and network structure. For example, a common exposure mapping includes the fraction of one's friends who are treated. Importantly, we do not assume that the user-specified exposure mapping is correctly specified; it may differ from the true exposure mapping in both functional form and dimension. This flexibility also allows us to incorporate censored network links in a unified way.

As our second contribution, we investigate whether the estimands associated with the OLS estimator can be interpreted as causal spillover effects.
We distinguish between two causal targets: a population-level estimand and a sample-level estimand. The population-level estimand is defined as the weighted average of the treatment effect vector across the entire population, including those who are not sampled, with complete network information. On the other hand, the sample-level estimand is defined as the sample average of the treatment effect vector across the sampled units, with the sampled network information. 
We show that both types of estimands can be contaminated: each element of the multi-dimensional estimands may reflect causal effects from other elements of the exposure mapping vector.
With heterogeneous treatment effects, correlations among elements in the exposure mapping vector (e.g., the proportion of treated friends and the proportion of friends' treated friends) blur the distinction between the true causal effects in one element and those in another. Although the population-level causal estimand can be free from contamination if the exposure mapping is defined such that there is no correlation among its elements, the sample-level causal estimand can still be subject to contamination, and thus lacks causal interpretability due to network sampling. Missing links can create undesirable correlations between the observed and true exposure mapping across different elements. As a result, the two estimands can remain distinct even in large samples unless the exposure mapping is correctly specified and the network links for the neighborhood are completely sampled.

In our third contribution, we derive asymptotic theory for the OLS estimator and find conditions under which the OLS estimator approximates the estimands.
We show that the OLS estimator is consistent for the sample-level causal estimand, conditionally or unconditionally on the sampling uncertainty. However, because the sample-level causal estimand generally lacks causal interpretability, results from OLS estimation should be interpreted with caution.
If the exposure mapping is correctly specified and there is no potential correlation between the true and observed exposure mappings, the sample-level causal estimand is consistent for the population-level causal estimand; thus, we can guarantee a clear interpretation of the OLS estimator.
We further derive the estimator's asymptotic distribution and provide a conservative network heteroskedasticity and autocorrelation consistent (HAC) variance estimator.

This paper contributes to the literature on design-based inference in network experiments (\citealp{aronow2017estimating}; \citealp{leung2022causal}; \citealp{gao2023causal}; \citealp{hoshino2024}). Previous works have primarily focused on design-based uncertainty, where treatment assignment is the only source of randomness and complete network information is assumed to be available without sampling uncertainty. Additionally, these works have mainly considered IPW estimators, which allow for direct estimation of causal spillover effects, while the OLS estimator has received less attention. To focus on IPW estimators, these works typically assume that the exposure mapping takes discrete values, such as an indicator of whether a unit has at least one treated friend.\footnote{\cite{gao2023causal} discuss the potential application of IPW-based estimators to continuous exposure mappings.}
In contrast, this paper considers both design-based and sampling-based uncertainties with an explicit network collection process, and focuses on the OLS estimator with exposure mappings as regressors, which is widely used in empirical applications and allows for continuous exposure mappings.

This paper also relates to the literature on simultaneous design-based and sampling-based inference (see \citealp{abadie2020sampling}; \citealp{xu2022design}; \citealp{abadie2023should}; \citealp{viviano2024policy}). Our framework extends the approach of \cite{abadie2020sampling} to the network setting by allowing for both design-based and sampling-based uncertainties in network experiments, and by focusing on both population-level and sample-level estimands. We differ from \cite{abadie2020sampling} in several important respects. First, we explicitly model network sampling, where the observed network may be only partially observed. Second, we study the OLS estimator with exposure mappings as regressors, which induces dependence among outcomes and between regressors and sampling indicators, features not present in their analysis. Third, we provide an element-wise causal interpretation of the estimands and the OLS estimator, which is not addressed in their work.
Relatedly, \cite{viviano2024policy} also considers both design-based and sampling-based uncertainties, including uncertainty arising from network sampling. However, while his approach assumes that all relevant network information for computing the true exposure mapping is observed, our framework allows for the possibility that some relevant network information is unobserved due to sampling uncertainty. Additionally, while \cite{viviano2024policy} focuses on a sample-level estimand that maximizes a welfare measure, our study is concerned with inference for both population-level and sample-level causal estimands, emphasizing the potential divergence between the two.

This paper is also related to the literature studying the impact of network data collection on parameters of interest (\citealp{chandrasekhar2016econometrics}; \citealp{griffith2022name}; \citealp{lewbel2023ignoring}; \citealp{hsieh2024}). While these papers share a similar motivation in that the network sampling process can affect the estimation of spillover effects, they primarily focus on the potential bias of estimators with respect to homogeneous parameters due to network sampling. In contrast, this paper focuses on the causal interpretability of the OLS estimator with heterogeneous spillover effects. This distinction is important because attenuation bias, as highlighted for example in \cite{chandrasekhar2016econometrics}, does not necessarily hinder learning about spillover effects if the estimator preserves the sign of the underlying effects. However, we show that the OLS estimator with exposure mappings may not preserve the sign of the true spillover effects due to contamination bias, potentially leading to misleading conclusions.

More broadly, this paper contributes to the literature on the causal interpretability of estimators in linear regressions with heterogeneous treatment effects (\citealp{angrist1998}; \citealp{goldsmith2022contamination}; \citealp{borusyak2024negative}). In particular, \cite{goldsmith2022contamination} show that the OLS estimator with multi-dimensional treatment indicators can be contaminated in the presence of heterogeneous treatment effects, which aligns with our findings in \cref{cor:causal_element}. There are two important differences. First, we consider a finite population model, whereas \cite{goldsmith2022contamination} focus on an infinite population model, making it nontrivial to extend their results to our setting. Second, we allow for general exposure mappings as regressors, while \cite{goldsmith2022contamination} restrict attention to mutually exclusive multi-dimensional treatment indicators. In our context, contamination bias arises from overlaps in the treatment status across elements of the exposure mapping, whereas such overlaps are not possible in the non-network setup of \cite{goldsmith2022contamination}.

The remainder of this paper is organized as follows. 
\cref{sec:model} introduces the framework for network sampling, the model, and assumptions.
\cref{sec:main} presents the main results, including a causal interpretation and asymptotic theory.
\cref{sec:variance} proposes a network heteroskedasticity and autocorrelation consistent (HAC) estimator for the standard errors.
\cref{sec:simulation} provides a simulation study to illustrate the finite sample properties of the proposed estimator.
\cref{sec:empirical} applies the proposed method to a real-world dataset.
\cref{sec:conclusion} concludes the paper and provides a flowchart (\cref{fig:flowchart}) outlining recommended steps for conducting inference in network experiments using the OLS estimator.
\cref{app:gamma} discusses how to estimate the nuisance parameters consistently,
\cref{app:lemmas} contains technical lemmas, \cref{app:proof} contains proofs, \cref{app:sim_additional} presents additional simulation results, and 
\cref{app:survey} lists the papers included in the survey of network experiment research presented in the Introduction.

\section{Model}
\label{sec:model}
In this section, we first outline our framework for modeling network experiments. We then introduce the estimands of interest, which are defined both for the entire population and for the sampled group, as well as the OLS estimator used to estimate these estimands.

\subsection{Population} As in \cite{abadie2020sampling}, we consider a sequence of finite populations.
There are finitely many units ($n<\infty$) in the population, denoted by $\mathcal{N}_{n}=\{1,...,n\}$. These units are connected through the network represented by an adjacency matrix $\bA_{n}=[A_{n,i,j}]_{i,j\in\mathcal{N}_{n}}\in\{0,1\}^{n\times n}$.
We define $A_{n,i,j}=1$ if there is a network link between units $i$ and $j$, and $A_{n,i,j}=0$ otherwise.
We assume that the network is undirected $(A_{n,i,j}=A_{n,j,i})$ and has no self-loops ($A_{n,i,i}=0$). Each unit $i$ is characterized by a vector of covariates $Z_{n,i}\in \mathcal{Z}_{n}\subset\mathbb{R}^{d_{Z}}$, potential outcomes $Y_{n,i}^{*}(\cdot)\in\mathcal{Y}_{n}\subset\mathbb{R}$ that depend on the entire vector of binary treatments $\bD_{n}=[D_{n,i}]_{i\in\mathcal{N}_{n}}\in \{0,1\}^{n}$. We consider the setup where the researcher assigns treatments only to the sampled units, but spillovers to non-sampled units are allowed. The covariates $Z_{n,i}$ include both network information (e.g., number of $i$'s neighbors, degree: $\text{deg}_{n,i}=\sum_{j\neq i}A_{n,i,j}$) and individual information (e.g., $i$'s age). Also, the potential outcomes may violate the Stable Unit Treatment Value Assumption (SUTVA) by allowing for others' treatment status as inputs. 

\subsection{Sampling} From a finite population of $n$ units, we draw a sample of $N=\sum_{i=1}^{n}R_{n, i}$ units (hence $n\geq N$), where $R_{n, i}\in\{0,1\}$ is the sampling indicator for the $i$-th unit: $R_{n,i}=1$ if $i$ is in the sample and otherwise $R_{n,i}=0$.
Given the sampling indicator vector $\bR_{n}$, partial elements of the true network $\bA_{n}$ are sampled. We denote the sampled network, given the sampling indicator vector $\bR_{n}$, as $\tilde{\bA}_{n}(\bR_{n})$.
When the dependence on $\bR_{n}$ is clear from context, we simply write $\tilde{\bA}_{n}$.
The sampled adjacency matrix $\tilde{\bA}_{n}=[\tilde{A}_{n,i,j}]_{i,j\in\mathcal{N}_{n}}\in\{0,1\}^{n\times n}$ has $(i,j)$-element $\tilde{A}_{n,i,j}$, which equals one if there is a true network link between units $i$ and $j$ and the link is sampled, and zero otherwise.

In this paper, we focus on two canonical network sampling methods: (i) \emph{induced subgraph sampling}, and (ii) \emph{star sampling}.
In the induced subgraph sampling case, we sample $\tilde{\bA}_{n}= \bR_{n}\bR_{n}'\odot \bA_{n}$ where $\odot$ is the element-wise product and the $(i,j)$-element of $\tilde{\bA}_{n}$, $\tilde{A}_{n,i,j}=R_{n,i}R_{n,j}A_{n,i,j}$ represents a network link between the units $i$ and $j$, which is sampled if both units are sampled.
In the star sampling case, we sample $\tilde{\bA}_{n}=\left(\boldsymbol{1}_{n}\boldsymbol{1}_{n}'-(\boldsymbol{1}_{n}-\bR_{n})(\boldsymbol{1}_{n}-\bR_{n})'\right)\odot \bA_{n}$, where $\tilde{A}_{n,i,j}=\max\{R_{n,i},R_{n,j}\}A_{n,i,j}$ represents a network link between the units $i$ and $j$, which is sampled if at least one of the two units is sampled.
Sampled networks under induced subgraph and star sampling are illustrated in \cref{fig:network_sampling}. 
In the figure, the sampled units are in blue, and the non-sampled units are in light gray. The sampled links are shown as solid black lines, and the non-sampled links as dashed gray lines.
In practice, if the researcher asks the sampled units to list their friends \textit{from the list of sampled units}, the induced subgraph sampling network is sampled (e.g., \citealp{Conley2010,dizon2020}; \citealp{carter2021subsidies}). If the researcher asks the sampled units to list their friends \textit{from the population}, the star sampling network is sampled (e.g., \citealp{banerjee2013diffusion}; \citealp{cai2015social}; \citealp{Beaman2021-hw}).
See Section 5.3 of \cite{kolaczyk2014statistical} for further examples of network sampling.

\begin{figure}[ht]
    \caption{Comparison of induced subgraph sampling (left) and star sampling (right).} 
    \begin{threeparttable}
        \centering
        \begingroup
        \begin{subfigure}{0.49\textwidth}
            \centering
            \includegraphics[width=\textwidth]{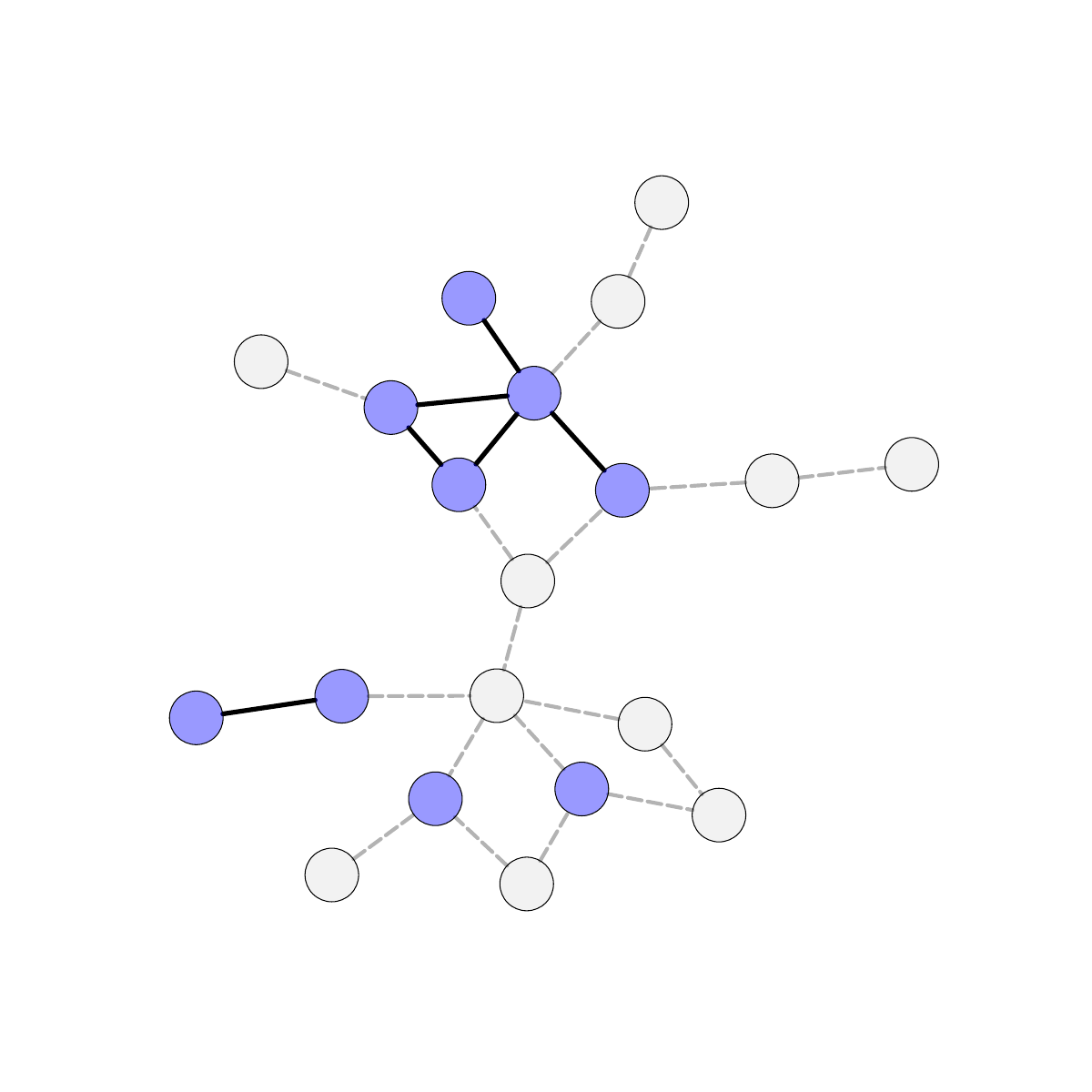}
            \caption{Induced subgraph sampling}
            \label{fig:in_sample}
        \end{subfigure}
        \hfill
        \begin{subfigure}{0.49\textwidth}
            \centering
            \includegraphics[width=\textwidth]{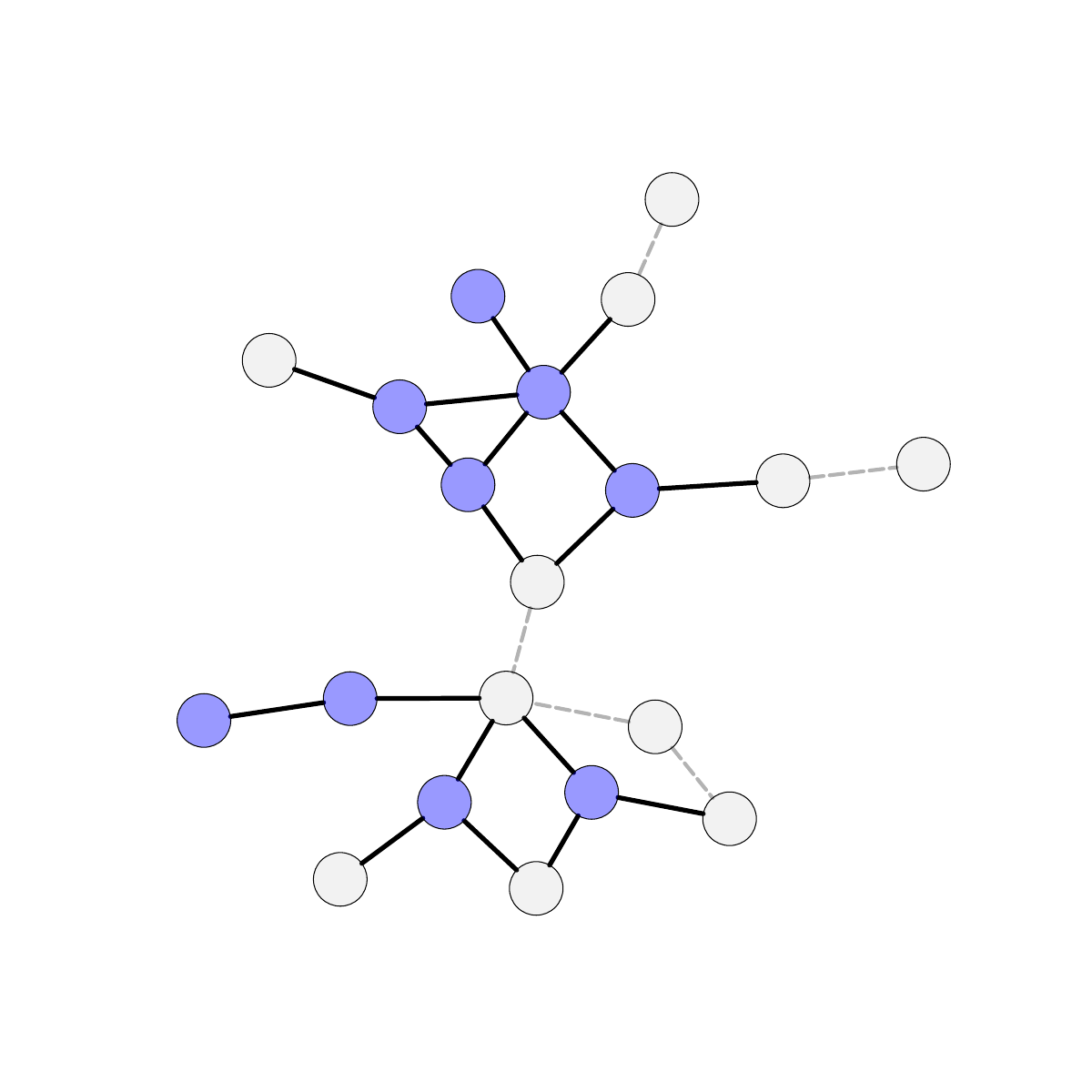}
            \caption{Star sampling}
            \label{fig:out_sample}
        \end{subfigure}
        \vspace{8pt}
        \begin{tablenotes}
            \footnotesize
            \item {\it Note:} Blue nodes indicate sampled units, while light gray nodes denote non-sampled units. Solid black links are observable to the researcher; dashed gray links are unobserved.
        \end{tablenotes}
    \endgroup
    \end{threeparttable}
           \label{fig:network_sampling}
\end{figure}

We denote the observed covariates by $\tilde{Z}_{n,i}$, which may differ from $Z_{n,i}$ due to network sampling. For example, if $Z_{n,i}$ includes $i$'s degree, then $\tilde{Z}_{n,i}$ contains $i$'s degree computed from the sampled network $\tilde{\bA}_{n}$: $\tilde{\text{deg}}_{n,i}=\sum_{j\neq i}\tilde{A}_{n,i,j}$. 
Note that we allow both $Z_{n,i}$ and $\tilde{Z}_{n,i}$ to depend on $\bR_n$.

Throughout the paper, we maintain the following assumption regarding the sampling process and the assignment mechanism.
\begin{assumption}   \label{asm:sampling}
    (i) Random sampling:
    \begin{align*}
        R_{n,i}\sim \text{Bernoulli}(\rho_{n}) \text{ i.i.d.},
    \end{align*}
    where $\rho_{n}\in(0,1]$ is a sequence of sampling probabilities such that $\rho_{n}\to\rho\in(0,1]$.
    
    (ii) Network sampling:
    Given a fixed entire network sequence $\bA_{n}\in\{0,1\}^{n\times n}$, the $(i,j)$-element of sampled network $\tilde{\bA}_{n}$ is generated by the induced subgraph sampling 
    $\tilde{A}_{n,i,j}=R_{n,i}R_{n,j}A_{n,i,j}$ or the star sampling $\tilde{A}_{n,i,j}=\max\{R_{n,i},R_{n,j}\}A_{n,i,j}$.
    
    (iii) Treatment assignment mechanism: Let $\bR_{n,-i}$ denote the vector $\bR_n$ excluding the $i$-th element, $R_{n,i}$.
    The assignment mechanism $D_{n,i}$ is independent of $\bR_{n,-i}$ and drawn independently (but not necessarily identically) from a known distribution. The distribution of $D_{n,i}$ is degenerate at $0$ if and only if $R_{n,i}=0$.
\end{assumption}

\cref{asm:sampling} (iii) implies $D_{n,i}=0$ if $R_{n,i}=0$, which means we treat only the sampled units. 
The simplest example is 
\begin{equation}
    D_{n,i}\sim \text{Bernoulli}(R_{n,i}p_{n,i}) \text{ independently}.
    \label{eq:bernoulli}
\end{equation}
While we use the Bernoulli assignment in \cref{eq:bernoulli} for all illustrations in the paper, our theoretical results accommodate more general assignment mechanisms, as specified in \cref{asm:sampling} (iii).
Since \cref{asm:sampling} (iii) does not require the identical draws, $p_{n,i}$ could depend on $\bA_{n}$, $Z_{n,i}$ or other observed characteristics of unit $i$. We can equivalently write \cref{asm:sampling} (iii) as $D_{n,i}=R_{n,i}D^*_{n,i}$, where $D^*_{n,i}$ is defined as the latent treatment indicator generated by $D^*_{n,i}\sim \text{Bernoulli}(p_{n,i})$ (or more general distribution satisfying the assumption) independently. Note also that the treatment assignment mechanism is known to the researcher, which is satisfied in a randomized controlled trial and commonly assumed in the design-based inference literature.

\begin{remark}
    \cref{asm:sampling} (i) rules out cluster sampling and multi-wave network sampling, because in such designs $R_{n,i}$ may depend on $R_{n,j}$ for some $j\neq i$ through the cluster or the network, respectively. 
    \cref{asm:sampling} (ii) prohibits censoring of $\tilde{\bA}_n$; when censoring occurs, it can be treated as misspecification of the exposure mapping (see \cref{ex:censor}). 
    \cref{asm:sampling} (iii) excludes complex assignment schemes, such as matched-pair or blocked randomization.
\end{remark}

\subsection{Potential Outcome}
As discussed above, each unit's potential outcome $Y^*_{n,i}(\cdot)$ is a function of the full treatment vector $\bD_{n}$. By \cref{asm:sampling} (iii), we can write $\bD_{n}=\bR_n\odot\bD_{n}^*$, where $\bD^*_{n}=[D^*_{n,i}]_{i\in\mathcal{N}_n}$. Following the literature (e.g., \citealp{aronow2017estimating}), we assume that there is an exposure mapping $T_{n,i}\in\mathcal{T}_n\subset\mathbb{R}^{d_T}$ that essentially determines $i$'s potential outcome by summarizing the network structure and the treatment status vector. See \cref{sec:exposure_mapping} below for a detailed definition and discussion on the exposure mapping.
We consider a linear potential outcome model, so that for each $t\in\mathcal{T}_{n}$, $Y_{n,i}^{*}(t)$ is defined as follows.
 \begin{assumption}\label{asm:linear}
    For all $t\in\mathcal{T}_{n}$, 
    \begin{align*}
        Y_{n,i}^{*}(t)=t'\theta_{n,i} + \nu_{n,i},
    \end{align*}
    where $\theta_{n,i}$ and $\nu_{n,i}$ are non-stochastic.
\end{assumption}

Note that $\theta_{n,i}$ is a vector of heterogeneous treatment effects. Each element $\theta_{n,i,(k)}$ represents the marginal effect of the $k$-th component of the exposure mapping $T_{n,i}$ on the potential outcome $Y_{n,i}^{*}(t)$. For example, if the $k$-th component of $T_{n,i}$ is the share of treated friends, then $\theta_{n,i,(k)}$ captures the causal spillover effect from $i$'s treated friends on $i$'s outcome. Since $\theta_{n,i}$ is non-stochastic, it may depend on the population network $\bA_{n}$, allowing for heterogeneity based on network structure and unit $i$'s position.

Although a linear model may seem restrictive, when $|\mathcal{T}_{n}|$ is finite (e.g., $\mathcal{T}_{n}=\{0,1\}^2$), this assumption is without loss of generality as discussed in \cite{abadie2020sampling}. 
The realized outcome is $Y_{n,i}=Y^*_{n,i}(T_{n,i})$. Thus, the outcome depends on $\bD_n$ only via the exposure mapping $T_{n,i}$.

\subsection{Exposure Mapping}
\label{sec:exposure_mapping}
Let the true exposure mapping be $T_{n,i}=g(i,\bD_{n},\bA_{n})\in\mathcal{T}_{n}\subset\mathbb{R}^{d_T}$, where $g:\mathcal{N}_{n}\times \{0,1\}^{n} \times \{0,1\}^{n\times n} \to \mathcal{T}_{n}$ is a function that generates the true exposure mapping for each unit. Specifically, for unit $i$, it takes (i) $i$'s index, (ii) the treatment vector $\bD_{n}$, and (iii) the true network $\bA_n$ as inputs, and returns a lower-dimensional vector of summary statistics for the outcome. For example, applied researchers use the presence of $i$'s treated friends and the share of $i$'s treated friends as the exposure mapping.

This paper allows the researcher to misspecify the functional form of $g$.
For example, if the researcher uses the presence of treated friends as the exposure mapping, while the true potential outcome is linear in the share of treated friends, then the exposure mapping is misspecified.
We denote this misspecified exposure mapping function by $\tilde{g}_n:\mathcal{N}_{n}\times \{0,1\}^{n}\times \{0,1\}^{n\times n}\to \tilde{\mathcal{T}}_{n}$, where $\tilde{\mathcal{T}}_{n}\in\mathbb{R}^{d_{\tilde{T}}}$. Note that the dimensions $d_T$ and $d_{\tilde{T}}$ may differ. 
The functional form $\tilde{g}_n$ could depend on the sample size $n$ (as in \cref{ex:cencormiss}), but for notational simplicity, we omit the subscript $n$.
We assume that dimensions $d_T$ and $d_{\tilde{T}}$ are constants independent of $n$.

If $\tilde{g}=g$, then the observed exposure mapping $\tilde{T}_{n,i}$ can be written as $\tilde{T}_{n,i}=g(i,\bD_{n},\tilde{\bA}_{n})$. That is, the only difference between the true exposure mapping and the observed exposure mapping is the network input, between $\bA_{n}$ and $\tilde{\bA}_{n}$.
More generally, if the researcher misspecifies $g$ as $\tilde{g}$, then the observed exposure mapping is $\tilde{T}_{n,i}=\tilde{g}(i,\bD_{n},\tilde{\bA}_{n})$.
In this case, the dimensions $d_T$ and $ d_{\tilde{T}}$ may differ.

Below, we provide some examples of exposure mappings. 
\begin{example}\label{ex:direct}
    Suppose that the true exposure mapping is $i$'s own treatment indicator:
    \begin{align*}
        T_{n,i}=g(i,\bD_{n},\bA_{n})&=D_{n,i}=R_{n,i}D^*_{n,i}.
    \end{align*}
    Note that the exposure mapping does not depend on the network information, and as long as the researcher correctly specifies the exposure mapping $g=\tilde{g}$, we have $T_{n,i}=\tilde{T}_{n,i}$ for all $i\in\mathcal{N}_n$.
\end{example}

\begin{example}\label{ex:withoutmiss}
    Suppose that the true exposure mapping is an indicator of the existence of at least one treated friend:
    \begin{align*}
        T_{n,i}=g(i,\bD_{n},\bA_{n})&=\mathds{1}\left\{\sum_{j\neq i}A_{n,i,j}R_{n,j}D^*_{n,j}>0\right\},
    \end{align*}
    and the researcher correctly specifies the exposure mapping as $\tilde{T}_{n,i}=g(i,\bD_{n},\tilde{\bA}_{n})$. Thus, for the induced subgraph sampling case ($\tilde{A}_{n,i,j}=R_{n,i}R_{n,j}A_{n,i,j}$),
    \begin{align*}
        \tilde{T}_{n,i}&=\mathds{1}\left\{\sum_{j\neq i}R_{n,i}R_{n,j}A_{n,i,j}R_{n,j}D^*_{n,j}>0\right\}=\mathds{1}\left\{R_{n,i}\sum_{j\neq i}A_{n,i,j}R_{n,j}D^*_{n,j}>0\right\}.
    \end{align*}
    Thus, when $R_{n,i}=1$, we have $T_{n,i}=\tilde{T}_{n,i}$. 
    For the star sampling case ($\tilde{A}_{n,i,j}=\max\{R_{n,i},R_{n,j}\}A_{n,i,j}$),
    \begin{align*}
        \tilde{T}_{n,i}&=\mathds{1}\left\{\sum_{j\neq i}\max\{R_{n,i},R_{n,j}\}A_{n,i,j}R_{n,j}D^*_{n,j}>0\right\}=\mathds{1}\left\{\sum_{j\neq i}A_{n,i,j}R_{n,j}D^*_{n,j}>0\right\},
    \end{align*}
    and we have $T_{n,i}=\tilde{T}_{n,i}$ for all $i\in\mathcal{N}_n$.
\end{example}

Although the two preceding examples correctly specify the exposure mapping, the subsequent example fails to do so.
\begin{example}\label{ex:cencormiss}
    Suppose that the true exposure mapping is a vector of a direct treatment, a spillover treatment through a fraction of treated peers, and their interaction term:
    \begin{align*}
        T_{n,i}=g(i,\bD_{n},\bA_{n})&=\left(R_{n,i}D^*_{n,i},\frac{\sum_{j\neq i}A_{n,i,j}R_{n,j}D^*_{n,j}}{\sum_{j\neq i}A_{n,i,j}},R_{n,i}D^*_{n,i}\times \frac{\sum_{j\neq i}A_{n,i,j}R_{n,j}D^*_{n,j}}{\sum_{j\neq i}A_{n,i,j}}\right).
    \end{align*}
    By convention, we usually set $\sum_{j\neq i}A_{n,i,j}R_{n,j}D^*_{n,j}/\sum_{j\neq i}A_{n,i,j}=0$ if $\sum_{j\neq i}A_{n,i,j}=0$ to negate the spillover effect.
    Suppose that the researcher misspecifies $\tilde{g}$ as
    \begin{align*}
        \tilde{T}_{n,i}=\tilde{g}(i,\bD_{n},\tilde{\bA}_{n})&=\left(R_{n,i}D^*_{n,i},\mathds{1}\left\{\sum_{j\neq i}\tilde{A}_{n,i,j}R_{n,j}D^*_{n,j}>0\right\}\right).
    \end{align*}
    In this specification, it is evident that $g \neq \tilde{g}$ because $d_{T} > d_{\tilde{T}}$.
    The misspecified $\tilde{g}$ accounts only for the direct effect and the spillover effect represented by an indicator of the presence of at least one treated friend. Consequently, not only do the dimensions differ, but the structures of the variables capturing spillover effects are also distinct.
    %
\end{example}

\subsection{Censored Network}
We can also treat censoring on a sampled network as arising from a misspecified exposure mapping as $\tilde{g}$ can specify which links in a sampled network $\tilde{\bA}$ to be used to compute the exposure mapping. This is empirically relevant as in practice, some studies impose a cap on the number of links each sampled unit can report, leading to a discrepancy between the sampled and censored networks. For example, in \cite{cai2015social}, each sampled unit was asked to report up to five closest friends, which potentially introduces censoring in the observed network. See also \cite{griffith2022name} for further examples and a detailed discussion of censoring in network data collection.
The following example illustrates how censoring can be framed as a misspecified exposure mapping:

\begin{example}\label{ex:censor}
    Let $g$ be the same as in \cref{ex:withoutmiss}. Suppose that the researcher misspecifies $\tilde{g}$ due to the censoring as 
    \begin{align*}
        \tilde{T}_{n,i}&=\tilde{g}(i,\bD_{n},\tilde{\bA}_{n})=g(i,\bD_{n},\bC_n(\tilde{\bA}_{n})\odot\tilde{\bA}_{n})=\mathds{1}\left\{\sum_{j\neq i}C_{n,i,j}(\tilde{\bA}_{n})\tilde{A}_{n,i,j}R_{n,j}D^*_{n,j}>0\right\},
    \end{align*}
    where $\bC_n(\tilde{\bA}_{n})$ is the censoring indicator matrix whose $(i,j)$-element is $C_{n,i,j}(\tilde{\bA}_{n})\in\{0,1\}$, a binary variable that indicates whether unit $j$ is censored from $i$'s perspective. The censoring indicator can be a random variable, as we allow it to be an unknown function of the sampled network $\tilde{\bA}_{n}$. For example, $C_{n,i,j}=1$ when unit $i$ (or $j$) with $R_{n,i}=1$ (or $R_{n,j}=1$) is asked to list their five closest friends and $j$ (or $i$) is one of them.\footnote{We can define $C_{n,i,i}(\tilde{\bA}_{n})$ arbitrarily because $A_{n,i,i}=0$.} In this example, $g\neq \tilde{g}$ in general and misspecification occurs due to the censoring.
\end{example}

We distinguish between the sampled network $\tilde{\bA}_{n}$ and the censored network $\bC_n(\tilde{\bA}_{n})\odot\tilde{\bA}_{n}$, and the discrepancy is framed as the misspecification of the exposure mapping.
This framework is useful for separating the sampling effect from the censoring. In the extreme case with $\rho_n=1$, we sample the entire network $\tilde{\bA}_{n}=\bA_{n}$, but the censoring still matters as we observe $\bC_n(\bA_{n})\odot\bA_{n}$.
For convenience, we will omit the notational dependence of $\bC_n$ on $\tilde{\bA}_{n}$.

The dependence of $\bC_n$ on $\tilde{\bA}_{n}$ is motivated as follows.
In practice, the censored induced subgraph sampling network is observed if the researcher asks the sampled unit to list a fixed number of closest friends \textit{from the sampled friends}. Thus, it usually depends on $[\tilde{A}_{n,i,j}]_{j\in\mathcal{N}_n}$.
The censored star sampling network is observed if the researcher asks $i$ with $R_{n,i}=1$ to list a fixed number of closest friends \textit{from their friends in population} $[A_{n,i,j}]_{j\in\mathcal{N}_n}$. Since $\tilde{A}_{n,i,j}=A_{n,i,j}$ holds for $R_{n,i}=1$ for the star sampling network, the censoring depends on $[\tilde{A}_{n,i,j}]_{j\in\mathcal{N}_n}$. We also allow the arbitrary dependence of $\bC_n$ on other deterministic variables, such as individuals' preferences regarding their friends, which is a benefit of the design-based framework.

\subsection{Estimands and Estimator}
To facilitate the introduction of our estimands and OLS estimator, we first transform the exposure mappings. 
Recall that the exposure mappings $T_{n,i}$ and $\tilde{T}_{n,i}$ are random vectors that depend on $\bR_n$ and $\bD_n$, and the covariates $Z_{n,i}$ and $\tilde{Z}_{n,i}$ are random vectors that depend only on $\bR_n$.
Define 
\begin{equation*}
    X_{n,i}=T_{n,i} - \Lambda_{n}Z_{n,i},\quad\text{ and }\quad\tilde{X}_{n,i}=\tilde{T}_{n,i} - \tilde{\Lambda}_{n}\tilde{Z}_{n,i},
\end{equation*}
where 
$$\Lambda_{n}=\left(\sum_{i=1}^n \mathbb{E}[T_{n,i}Z_{n,i}^{\prime}]\right)\left(\sum_{i=1}^n \mathbb{E}[Z_{n,i}Z_{n,i}^{\prime}]\right)^{-1},$$
and 
$$\tilde{\Lambda}_{n}=\left(\sum_{i=1}^n R_{n, i}\mathbb{E}[\tilde{T}_{n,i}|\bR_n]\tilde{Z}_{n,i}^{\prime}\right)\left(\sum_{i=1}^n R_{n, i}\tilde{Z}_{n,i}\tilde{Z}_{n,i}^{\prime}\right)^{-1}.$$
That is, $X_{n,i}$ is the population residual of the regression of $T_{n,i}$ on $Z_{n,i}$, and $\tilde{X}_{n,i}$ is the residual of the regression of $\tilde{T}_{n,i}$ on $\tilde{Z}_{n,i}$ using sampled units. Since we know the treatment assignment distribution with known $p_{n,i}$ and observe $\bR_n$, we can calculate $\mathbb{E}[\tilde{T}_{n,i}|\bR_n]$ analytically.

\cref{tab:lin_prop} summarizes the conditional expectation of widely used exposure mappings when the assignment probability is homogeneous: $D^*_{n,i}\sim \text{Bernoulli}(p_{n})$ i.i.d. The table focuses on the case where the exposure mapping is scalar. The researcher applies it element-wise for multi-dimensional cases. 
For the second neighborhood, the expectation can be calculated similarly. See also \cref{ex:cai} below for the modification on multi-dimensional cases with the second neighborhood.

\begin{table}[h]
    \caption{Conditional Expectation of Exposure Mappings Frequently Used in Applied Research}
    \label{tab:lin_prop}
    \begin{threeparttable}
    \centering
    \begingroup
    \setlength{\extrarowheight}{5pt}
        \begin{tabular}{c c c}
            \hline
            \hline
            Exposure Mapping & $\tilde{T}_{n,i}=g(i,\bD_{n},\tilde{\bA}_{n})$ & $\mathbb{E}\left[\tilde{T}_{n,i}\mid\bR_n\right]$\\
            \hline
            Individual Treatment & $R_{n,i}D^*_{n,i}$ & $R_{n,i}p_{n}$\\
            Treated Friends Share & $\frac{\sum_{j\neq i}\tilde{A}_{n,i,j}R_{n,j}D^*_{n,j}}{\sum_{j\neq i}\tilde{A}_{n,i,j}}$ & $p_{n}\times\frac{\sum_{j\neq i}\tilde{A}_{n,i,j}R_{n,j}}{\sum_{j\neq i}\tilde{A}_{n,i,j}}$\\
            Treated Friends Number & $\sum_{j\neq i}\tilde{A}_{n,i,j}R_{n,j}D^*_{n,j}$ & $p_{n}\times \sum_{j\neq i}\tilde{A}_{n,i,j}R_{n,j}$\\
            Treated Friends Existence & $\mathds{1}\left\{\sum_{j\neq i}\tilde{A}_{n,i,j}R_{n,j}D^*_{n,j}>0\right\}$ & $1-(1-p_{n})^{\sum_{j\neq i}\tilde{A}_{n,i,j}R_{n,j}}$\\
            \hline
        \end{tabular}
        \begin{tablenotes}
            \footnotesize
            \item {\it Note:} Assume that $R_{n,i}\sim \text{Bernoulli}(\rho_{n})$ i.i.d. and $D^*_{n,i}\sim \text{Bernoulli}(p_{n})$ i.i.d. By convention, we usually set $\sum_{j\neq i}\tilde{A}_{n,i,j}R_{n,j}D^*_{n,j}/\sum_{j\neq i}\tilde{A}_{n,i,j}=0$ if $\sum_{j\neq i}\tilde{A}_{n,i,j}=0$.
        \end{tablenotes}
    \endgroup
    \end{threeparttable}
\end{table}

To summarize relevant moments of the data, define the population matrix \(\Omega_n\) and the sample matrices \(\tilde Q_n\) and \(\tilde\Omega_n\):
$$
\Omega_n=\frac{1}{n} \sum_{i=1}^n\mathbb{E}\left[\left(\begin{array}{c}
    Y_{n,i}\\
    X_{n,i}\\
    Z_{n,i}
\end{array}\right)\left(\begin{array}{c}
    Y_{n,i}\\
    X_{n,i}\\
    Z_{n,i}
\end{array}\right)^{\prime}\right]\equiv\left(\begin{array}{ccc}
    \Omega_n^{YY} & \Omega_n^{YX} & \Omega_n^{YZ}\\
    \Omega_n^{XY} & \Omega_n^{XX} & \Omega_n^{XZ}\\
    \Omega_n^{ZY} & \Omega_n^{ZX} & \Omega_n^{ZZ}
\end{array}\right),
$$
$$
\tilde{Q}_n=\frac{1}{N} \sum_{i=1}^n R_{n, i}\left(\begin{array}{c}
    Y_{n,i}\\
    \tilde{X}_{n,i}\\
    \tilde{Z}_{n,i}
\end{array}\right)\left(\begin{array}{c}
    Y_{n,i}\\
    \tilde{X}_{n,i}\\
    \tilde{Z}_{n,i}
\end{array}\right)^{\prime}\equiv\left(\begin{array}{ccc}
    \tilde{Q}_n^{YY} & \tilde{Q}_n^{YX} & \tilde{Q}_n^{YZ}\\
    \tilde{Q}_n^{XY} & \tilde{Q}_n^{XX} & \tilde{Q}_n^{XZ}\\
    \tilde{Q}_n^{ZY} & \tilde{Q}_n^{ZX} & \tilde{Q}_n^{ZZ}
\end{array}\right),
$$
and
$$
\tilde{\Omega}_n=\frac{1}{N} \sum_{i=1}^n R_{n, i}\mathbb{E}\left[\left(\begin{array}{c}
    Y_{n,i}\\
    \tilde{X}_{n,i}\\
    \tilde{Z}_{n,i}
\end{array}\right)\left(\begin{array}{c}
    Y_{n,i}\\
    \tilde{X}_{n,i}\\
    \tilde{Z}_{n,i}
\end{array}\right)^{\prime}\mid\bR_n\right]\equiv\left(\begin{array}{ccc}
    \tilde{\Omega}_n^{YY} & \tilde{\Omega}_n^{YX} & \tilde{\Omega}_n^{YZ}\\
    \tilde{\Omega}_n^{XY} & \tilde{\Omega}_n^{XX} & \tilde{\Omega}_n^{XZ}\\
    \tilde{\Omega}_n^{ZY} & \tilde{\Omega}_n^{ZX} & \tilde{\Omega}_n^{ZZ}
\end{array}\right).
$$
Note that the expectation for $\Omega_n$ is taken over $\bD_n$ and $\bR_n$ while the conditional expectation for $\tilde{\Omega}_n$ is taken over $\bD_n$ conditional on $\bR_n$. 

Our estimands of interest are 
\begin{equation}
    \left(\begin{array}{c}
        \theta_{n}^{\causal}\\
        \gamma_{n}^{\causal}
    \end{array}\right)=\left(\begin{array}{cc}
        \Omega_n^{XX} & \Omega_n^{XZ}\\
        \Omega_n^{ZX} & \Omega_n^{ZZ}
    \end{array}\right)^{-1}\left(\begin{array}{c}
        \Omega_n^{XY}\\
        \Omega_n^{ZY}
    \end{array}\right),
\end{equation}
and
\begin{equation}
    \left(\begin{array}{c}
        \theta_{n}^{\causample}\\
        \gamma_{n}^{\causample}
    \end{array}\right)=\left(\begin{array}{cc}
        \tilde{\Omega}_n^{XX} & \tilde{\Omega}_n^{XZ}\\
        \tilde{\Omega}_n^{ZX} & \tilde{\Omega}_n^{ZZ}
    \end{array}\right)^{-1}\left(\begin{array}{c}
        \tilde{\Omega}_n^{XY}\\
        \tilde{\Omega}_n^{ZY}
    \end{array}\right).
\end{equation}
These are causal estimands in the sense specified by \cite{abadie2020sampling}. 
$(\theta_{n}^{\causal},\gamma_{n}^{\causal})'$ concerns the population-level causal effects of intervention while $(\theta_{n}^{\causample},\gamma_{n}^{\causample})$ concerns the sample-level causal effects when the sampling is governed by $\bR_{n}$.
$(\theta_{n}^{\causal},\gamma_{n}^{\causal})'$ is a solution for the population moment condition:
\begin{align}
    \label{eq:moment_causal}
    \frac{1}{n} \sum_{i=1}^n \mathbb{E}\left[\left(\begin{array}{c}
        X_{n,i}\\
        Z_{n,i}
    \end{array}\right)\left(Y_{n, i}-X_{n, i}^{\prime} \theta_{n}^{\causal} - Z_{n, i}^{\prime} \gamma_{n}^{\causal}\right)\right]=0,
\end{align}
and $(\theta_{n}^{\causample},\gamma_{n}^{\causample})$ is a solution for the sample moment condition:
\begin{align}
    \label{eq:moment_causalsample}
    \frac{1}{N} \sum_{i=1}^n R_{n, i}\mathbb{E}\left[\left(\begin{array}{c}
        \tilde{X}_{n,i}\\
        \tilde{Z}_{n,i}
    \end{array}\right)\left(Y_{n, i}-\tilde{X}_{n, i}^{\prime} \theta_{n}^{\causample} - \tilde{Z}_{n, i}^{\prime} \gamma_{n}^{\causample}\right)\mid\bR_n\right]=0.
\end{align}
We study (i) whether the sample-level estimand can be estimated consistently (internal validity), and, if so, (ii) how closely it approximates the population-level estimand (external validity). We will also discuss whether each element of these estimands admits a causal interpretation, namely, whether each OLS coefficient represents a convex combination of the corresponding heterogeneous treatment effects, which is not discussed in \cite{abadie2020sampling}.

For the sample-level causal estimand, we consider the ordinary least squares estimator:
\begin{equation}
    \left(\begin{array}{c}
        \hat{\theta}_{n}\\
        \hat{\gamma}_{n}
    \end{array}\right)=\left(\begin{array}{cc}
        \tilde{Q}_n^{XX} & \tilde{Q}_n^{XZ}\\
        \tilde{Q}_n^{ZX} & \tilde{Q}_n^{ZZ}
    \end{array}\right)^{-1}\left(\begin{array}{c}
        \tilde{Q}_n^{XY}\\
        \tilde{Q}_n^{ZY}
    \end{array}\right).
\end{equation}
Equivalently, the moment condition is
\begin{equation}
    \frac{1}{n} \sum_{i=1}^n R_{n, i}\left(\begin{array}{c}
        \tilde{X}_{n,i}\\
        \tilde{Z}_{n,i}
    \end{array}\right)\left(Y_{n, i}-\tilde{X}_{n, i}^{\prime} \theta - \tilde{Z}_{n, i}^{\prime} \gamma\right)=0.
\end{equation}

An alternative approach is to use the inverse probability weighting (IPW) estimator (e.g., \citealp{leung2022causal}; \citealp{gao2023causal}). A usual condition for the IPW estimator to work in a network experimental setting is the individual-level overlapping condition; in our notation, we need to have $\mathbb{P}[\tilde{T}_{n,i}=t|\bR_{n}]\in (\eta,1-\eta)$ almost surely for all $i\in \mathcal{N}_{n}$ and $t\in\mathcal{T}_{n}$ for some $\eta\in(0,1/2)$. This overlapping condition is difficult to maintain in the network sampling framework. For example, consider a population of two connected units. Suppose the first unit is sampled, while the second is not. The exposure mapping is defined as the number of treated neighbors. In this case, $\mathbb{P}[\tilde{T}_{n,1}=1|\bR_{n}]=0$, thereby violating the overlapping condition.
Also, it is notable that the IPW estimator typically targets a quantity that differs from our estimands, which are defined through moment conditions in \cref{eq:moment_causal} and \cref{eq:moment_causalsample}.

Throughout this section, we have defined the network sampling framework, the exposure mapping, and the potential outcome model. We have also defined the population- and sample-level estimands, which are the solutions to the population and sample moment conditions, respectively. The next section provides our main theoretical results within this framework.

\section{Main Results}
In this section, we present the main results of this paper. We first discuss the population- and sample-level estimands' causal interpretation, then derive the 
asymptotic properties of the OLS estimator for both.
Since we have assumed that the sequence of sampling probabilities $\rho_n$ is bounded away from $0$ (\cref{asm:sampling}-(i)), it follows that $N>0$ a.s. for large enough $n$ (\cref{lem:N_pos} in the Supplemental Appendix). Thus, there is no additional concern for the degeneracy of the estimands and the OLS estimator in a large population, relative to the standard design-based setting with $\rho_n=1$.\footnote{This is why we have a stronger statement than \cite{abadie2020sampling} who allow $\rho_n\to 0$ as $n\to\infty$ and use ``with probability approaching $1$'' instead of ``almost surely'' in their results. Note that \cref{lem:N_pos} allows $\rho_n\to 0$ as long as $\rho_n n\to \infty$.}
\label{sec:main}
\subsection{Interpretability of the Causal Estimands}
We impose the following regularity conditions for the causal estimands to be well-defined. 
These conditions require boundedness of the outcome, exposure mappings, and covariates, as well as full rank of the exposure mappings and covariates.

\begin{assumption}
    \label{asm:moment}
    \
    \begin{enumerate}
        \item (Uniform Boundedness): The sequence of potential outcomes $Y_{n,i}^{*}(\cdot)$ is uniformly bounded, i.e., there exists some constant $\bar{Y}>0$ such that $|Y_{n,i}^{*}(t)|\leq \bar{Y}<\infty$ for all $n$, $i\in \mathcal{N}_{n}$, and $t\in\mathcal{T}$.
        \item The sequences of exposure mappings $T_{n,i}$ and $\tilde{T}_{n,i}$ satisfy the following.
        \begin{enumerate}
            \item (Uniform Boundedness): There exists some constant $\bar{T}$ such that $\|T_{n,i}\|,\|\tilde{T}_{n,i}\|\leq \bar{T}<\infty$ almost surely for all $n,i\in\mathcal{N}_{n}$.
            \item (Variation): $(1/n)\times\sum_{i\in \mathcal{N}_{n}}\Var(T_{n,i})$ is invertible and $(1/N)\times\sum_{i\in \mathcal{N}_{n}}R_{n,i}\Var(\tilde{T}_{n,i}\mid\bR_n)$ is almost surely invertible for large enough $n$.
        \end{enumerate}
        \item The sequences of covariates $Z_{n,i}$ and $\tilde{Z}_{n,i}$ satisfy the following.
        \begin{enumerate}
            \item (Uniform Boundedness): There exists some constant $\bar{Z}$ such that $\|Z_{n,i}\|,\|\tilde{Z}_{n,i}\|\leq \bar{Z}<\infty$ almost surely for all $n,i\in\mathcal{N}_{n}$.
            \item (Full Rank): $(1/n)\times\sum_{i=1}^nZ_{n,i}Z_{n,i}'$ is almost surely full-rank for large enough $n$, and $(1/N)\times\sum_{i=1}^nR_{n,i}\allowbreak \tilde{Z}_{n,i}\tilde{Z}_{n,i}'$ is almost surely invertible for large enough $n$.
        \end{enumerate}
    \end{enumerate}
\end{assumption}

\cref{asm:moment} (iii) implies that the sequences of residualized exposure mappings $X_{n,i}$ and $\tilde{X}_{n,i}$ satisfy the following.
\begin{enumerate}[label=(\alph*)]
    \item (Uniform Boundedness): There exists some constant $\bar{X}$ such that $\|X_{n,i}\|,\|\tilde{X}_{n,i}\|\leq \bar{X}<\infty$ almost surely for all $n,i\in\mathcal{N}_{n}$.
    \item (Full Rank): $(1/n)\times\sum_{i\in \mathcal{N}_{n}}\mathbb{E}[X_{n,i}X_{n,i}']$ is invertible and $(1/N)\times\sum_{i=1}^nR_{n,i}\mathbb{E}[\tilde{X}_{n,i}\tilde{X}_{n,i}'|\bR_{n}]$ is almost surely invertible for large enough $n$.
\end{enumerate}

The uniform boundedness of the potential outcomes in \cref{asm:moment} (i) is a standard assumption in the literature (e.g., \citealp{leung2022causal,gao2023causal}). \cref{asm:moment} (ii-a) rules out some network statistics in a large, dense network (e.g., a diverging degree). \cref{asm:moment} (ii-b) requires that the exposure mappings are not degenerate across the units. For example,  in \cref{ex:withoutmiss},  \cref{asm:moment} (ii-b) is violated if the network is empty, $A_{n,i,j}=0$ for all $i,j\in \mathcal{N}_{n}$, as $\mathds{1}\{\sum_{j\neq i}R_{n,j}A_{n,i,j}D^*_{n,j}>0\}=0$ for all $i\in\mathcal{N}_{n}$.
\cref{asm:moment} (iii-b) does not exclude the constant term in $Z_{n,i}$ and $\tilde{Z}_{n,i}$.
\cref{asm:moment} (ii-b) and (iii-b) are not as restrictive as they seem since we have $N>0$ a.s. for large enough $n$.

We impose an additional condition on the exposure mapping:
\begin{assumption}
    \label{asm:linear_propensity}
    There exists a sequence of matrices $L_{n}$ such that
    \begin{align*}
        \mathbb{E}[T_{n,i}|\bR_{n}]=L_{n}Z_{n,i}\quad\text{a.s.}
    \end{align*}
    for large enough $n$. Similarly, there exists a sequence of matrices $\tilde{L}_{n}$ measurable with respect to $\sigma(\bR_{n})$ such that 
    \begin{align*}
        \mathbb{E}[\tilde{T}_{n,i}|\bR_{n}]=\tilde{L}_{n}\tilde{Z}_{n,i}\quad\text{a.s.}
    \end{align*}
    for large enough $n$.
\end{assumption}

This assumption is fairly weak, as it is automatically satisfied if $\mathbb{E}[T_{n,i}|\bR_{n}]$ and $\mathbb{E}[\tilde{T}_{n,i}|\bR_{n}]$ are included in $Z_{n,i}$ and $\tilde{Z}_{n,i}$, respectively. Typically, in a field experiment, the experimenter knows the assignment mechanism, so $\mathbb{E}[\tilde{T}_{n,i}|\bR_{n}]$ can be computed either analytically or numerically and included as covariates. As the following example shows, in some cases, it is sufficient to include some network statistics in the covariates to satisfy this assumption.

\begin{example}\label{ex:kramer}
    Consider a variant of the exposure mapping in \cite{miguel2004worms} that counts the number of treated friends:
    \begin{align*}
        T_{n,i}=g(i,\bD_{n},\bA_{n})=\sum_{j\neq i}A_{n,i,j}R_{n,j}D^*_{n,j}.
    \end{align*}
    If there is no censoring, then $\tilde{g}=g$.
    The conditional expectations of exposure mappings are derived as $\mathbb{E}[T_{n,i}|\bR_{n}]=\sum_{j\neq i}A_{n,i,j}R_{n,j}p_{n,j}$, and $\mathbb{E}[\tilde{T}_{n,i}|\bR_{n}]=\sum_{j\neq i}\tilde{A}_{n,i,j}R_{n,j}p_{n,j}=\sum_{j\neq i}A_{n,i,j}R_{n,j}p_{n,j}$ for $R_{n,i}=1$.
    Thus, \cref{asm:linear_propensity} holds if the weighted degree $\sum_{j\neq i}A_{n,i,j}R_{n,j}p_{n,j}$ is included in $Z_{n,i}$ and $\tilde{Z}_{n,i}$.
\end{example}

We obtain the following transformations of the estimands in terms of the individual causal effects $\theta_{n,i}$ in the linear potential outcome model in \cref{asm:linear}:
\begin{theorem}\label{thm:weakly_causal}
        Under \cref{asm:sampling,asm:linear,asm:moment,asm:linear_propensity}, for large enough $n$,
    \begin{align*}
        \theta_{n}^{\causal}=\left(\sum_{i=1}^n\mathbb{E}[X_{n,i}X_{n,i}']\right)^{-1}\sum_{i=1}^n\mathbb{E}[X_{n,i}X_{n,i}']\theta_{n,i},
    \end{align*}
    and
    \begin{align*}
        \theta_{n}^{\causample}=\left(\sum_{i=1}^nR_{n,i}\mathbb{E}[\tilde{X}_{n,i}\tilde{X}_{n,i}'|\bR_{n}]\right)^{-1}\sum_{i=1}^nR_{n,i}\mathbb{E}[\tilde{X}_{n,i}X_{n,i}'|\bR_{n}]\theta_{n,i}\quad\text{a.s.}
    \end{align*}
\end{theorem}
\cref{thm:weakly_causal} shows that $\theta_{n}^{\causal}$ is expressed as a weighted sum of causal effects $\theta_{n,i}$ induced by the exposure mapping. On the other hand, $\theta_{n}^{\causample}$ is not necessarily a weighted sum of $\theta_{n,i}$ because of the difference in $X_{n,i}$ and $\tilde{X}_{n,i}$ in the numerator.
Moreover, the dimension of $\theta_{n}^{\causample}$ is $d_{\tilde{T}}$, which can be different from $d_{T}$, the dimension of $\theta_{n,i}$.

In the absence of \cref{asm:linear_propensity}, it is known that the formula in \cref{thm:weakly_causal} does not hold due to the omitted variable bias (OVB).
\cref{asm:linear_propensity} and \cref{thm:weakly_causal} suggest a takeaway for practitioners: \textit{under the linear propensity scores, the researcher can select necessary controls
easily to avoid the OVB.}

The linear propensity score assumption \cref{asm:linear_propensity} is a weak assumption in design-based causal inference. This assumption also appears in \cite{abadie2020sampling} and \cite{borusyak2023nonrandom}. 
In the latter, the OVB is removed by using the recentered instruments. Theoretically, including the controls and using the recentered instruments are equivalent, but including the controls is more frequently used in practice. While \cite{borusyak2023nonrandom} focuses on homogeneous treatment effects, this paper allows for heterogeneous treatment effects.

Note that in general, the $k$-th elements of $\theta_{n}^{\causal}$ and $\theta_{n}^{\causample}$ do not directly correspond to the causal effect of changes in the $k$-th element of the exposure mapping on the outcomes. For example, if the exposure mapping is two-dimensional, we could have the first element of $\theta_{n}^{\causal}$ to be negative while the first element of $\theta_{n,i}$ is positive for all $i\in\mathcal{N}_{n}$ if the second element of it is significantly negative. 

\subsection{Causal Interpretation}
To provide a causal interpretation for each element $\theta_{n,(k)}^{\causal}$ and $\theta_{n,(k)}^{\causample}$, we develop an element-wise version of \cref{thm:weakly_causal}. To this end, we let $T_{n,i,(k)}$ denote the $k$-th element of $T_{n,i}$. Similarly, we write $\tilde{T}_{n,i,(k)},X_{n,i,(k)},\tilde{X}_{n,i,(k)}$. For each $k$, let $U_{n,i,(k)}$ be the residual when projecting $X_{n,i,(k)}$ onto the $X_{n,i,(-k)}=(X_{n,i,(l)})_{l\neq k}$:
\begin{equation*}
    U_{n,i,(k)}=X_{n,i,(k)}-\left(\sum_{i=1}^n\mathbb{E}[X_{n,i,(k)}X_{n,i,(-k)}']\right)\left(\sum_{i=1}^n\mathbb{E}[X_{n,i,(-k)}X_{n,i,(-k)}']\right)^{-1}X_{n,i,(-k)}.
\end{equation*}
Similarly, define 
\begin{equation*}
    \tilde{U}_{n,i,(k)}=\tilde{X}_{n,i,(k)}-\left(\sum_{i=1}^nR_{n,i}\mathbb{E}[\tilde{X}_{n,i,(k)}\tilde{X}_{n,i,(-k)}'|\bR_{n}]\right)\left(\sum_{i=1}^nR_{n,i}\mathbb{E}[\tilde{X}_{n,i,(-k)}\tilde{X}_{n,i,(-k)}'|\bR_{n}]\right)^{-1}\tilde{X}_{n,i,(-k)}.
\end{equation*}
Then, we have the following decompositions:
\begin{corollary}
    \label{cor:causal_element}
    Under \cref{asm:sampling,asm:linear,asm:moment,asm:linear_propensity}, for large enough $n$,
    \begin{align}
        \theta_{n,(k)}^{\causal}=\frac{\sum_{i=1}^n\mathbb{E}[U_{n,i,(k)}X_{n,i,(k)}]\theta_{n,i,(k)}}{\sum_{i=1}^n\mathbb{E}[U_{n,i,(k)}^{2}]}
        +\frac{\sum_{i=1}^n\mathbb{E}[U_{n,i,(k)}X_{n,i,(-k)}']\theta_{n,i,(-k)}}{\sum_{i=1}^n\mathbb{E}[U_{n,i,(k)}^{2}]}\label{eq:causal_decomp}
    \end{align}
    for each $k=1,...,d_{T}$, and
    \begin{align}
            \theta_{n,(k)}^{\causample}
            =\frac{\sum_{i=1}^nR_{n,i}\mathbb{E}[\tilde{U}_{n,i,(k)}X_{n,i}'|\bR_{n}]\theta_{n,i}}{\sum_{i=1}^nR_{n,i}\mathbb{E}[\tilde{U}_{n,i,(k)}^{2}|\bR_{n}]}\quad\text{a.s.}
            \label{eq:causample_decomp}
    \end{align}
    for each $k=1,...,d_{\tilde{T}}$.
    Under an additional assumption $d_{\tilde{T}}=d_T$, we can simplify it into
    \begin{align}
        \theta_{n,(k)}^{\causample}
        =&\frac{\sum_{i=1}^nR_{n,i}\mathbb{E}[\tilde{U}_{n,i,(k)}X_{n,i,(k)}|\bR_{n}]\theta_{n,i,(k)}}{\sum_{i=1}^nR_{n,i}\mathbb{E}[\tilde{U}_{n,i,(k)}^{2}|\bR_{n}]}\nonumber\\
        &+\frac{\sum_{i=1}^nR_{n,i}\mathbb{E}[\tilde{U}_{n,i,(k)}X_{n,i,(-k)}'|\bR_{n}]\theta_{n,i,(-k)}}{\sum_{i=1}^nR_{n,i}\mathbb{E}[\tilde{U}_{n,i,(k)}^{2}|\bR_{n}]}\quad\text{a.s.}
        \label{eq:causample_decomp_special}
    \end{align}
    for each $k=1,...,d_{T}$.
\end{corollary}

\cref{cor:causal_element} shows that $\theta_{n,(k)}^{\causal}$ and $\theta_{n,(k)}^{\causample}$ can be influenced by effects from other elements $\theta_{n,i,(l)}$ with $l\neq k$. 
However, the residualization does not eliminate contamination bias, because the definition of $U_{n,i,(k)}$ and $\tilde{U}_{n,i,(k)}$ only implies 
$$\sum_{i=1}^n\mathbb{E}[U_{n,i,(k)}X_{n,i,(-k)}']=0 \text{ and } \sum_{i=1}^nR_{n,i}\mathbb{E}[\tilde{U}_{n,i,(k)}X_{n,i,(-k)}'|\bR_{n}]=0,$$ 
respectively.
Moreover, $\mathbb{E}[U_{n,i,(k)}X_{n,i,(k)}]$ and $\mathbb{E}[\tilde{U}_{n,i,(k)}X_{n,i,(k)}|\bR_{n}]$ are not guaranteed to be non-negative.

\begin{example}
    Suppose the true exposure mapping is the number of treated friends, $T_{n,i}= \sum_{j\neq i}A_{n,i,j}R_{n,j}D^*_{n,j}$, and that $T_{n,i}$ takes three possible values $1$, $2$, or $3$. 
    Suppose the researcher misspecifies the exposure mapping as dummy variables: $\tilde{T}_{n,i}=(\tilde{T}_{n,i,(1)},\tilde{T}_{n,i,(2)},\tilde{T}_{n,i,(3)})$, where $\tilde{T}_{n,i,(k)}=\allowbreak\mathds{1}\{\sum_{j\neq i}\tilde{A}_{n,i,j}R_{n,j}D^*_{n,j}=k\}$.
    Thus, $d_{\tilde{T}}=3>d_T=1$.
    For simplicity, consider star sampling, which provides all network links in the first neighborhood. Then, $\tilde{T}_{n,i,(k)}=\allowbreak\mathds{1}\{T_{n,i}=k\}$. Equation \cref{eq:causample_decomp} in \cref{cor:causal_element} implies that the coefficient for $\tilde{T}_{n,i,(k)}$ has no contamination term because $X_{n,i}$ is a scalar in this example. However, each element of $\theta_{n}^{\causample}$ captures a different weighted sum of $\theta_{n,i}$; thus, the interpretation is unclear.
\end{example}

\begin{remark}
    \textbf{(i)} Assuming $d_{\tilde{T}}=d_T$ requires the researcher to correctly specify the dimension of the exposure mapping ($d_{T}=d_{\tilde{T}}$).
    However, this assumption allows the researcher to misspecify the shape of $\tilde{g}\neq g$ or mismeasure the network.

    \textbf{(ii)} Our result for $\theta_n^{\causal}$ is a design-based analogue of Proposition 1 in \cite{goldsmith2022contamination}. The main differences are that their analysis is model-based and focuses on mutually exclusive treatment indicators (e.g., K-arms).\footnote{Mutually exclusive treatments guarantee that each treatment's own effect receives a non-negative weight.}\textsuperscript{,}\footnote{
        \cite{goldsmith2022contamination} propose three approaches to eliminate contamination bias, but all require modeling the conditional expectation of heterogeneous treatment effects based on observed covariates. In a design-based setting with deterministic treatment effects $\theta_{n,i}$, such modeling is not appropriate. Even if the modeling assumption is justified, their methods may be unreliable for network experiments due to weak overlap in propensity scores, which is often violated for common exposure mappings.
    }
    In contrast, we allow more flexible treatments, including network spillovers. Our decomposition for $\theta_n^{\causample}$ additionally accommodates both misspecification of the exposure mapping and mismeasurement of the network.

    \textbf{(iii)} If the distribution of $T_{n,i}$ does not depend on $i$, a result in \cref{cor:causal_element} can be strengthened to 
    \begin{align*}
        \theta_{n,(k)}^{\causal}=\frac{\sum_{i=1}^n\mathbb{E}[U_{n,i,(k)}X_{n,i,(k)}]\theta_{n,i,(k)}}{\sum_{i=1}^n\mathbb{E}[U_{n,i,(k)}^{2}]}
    \end{align*}
    for any $k$.
    That is, we do not have a contamination bias. 
    However, the weight can be negative.
    Moreover, the homogeneous requirement of the treatment variable $T_{n,i}$ is usually violated in design-based network experiments since the exposure mapping depends on the network information for each $i$ and the population network $\bA_n$ is treated as non-random.
    
    \textbf{(iv)} The weight for $\theta_{n}^{\causal}$ is clearly non-negative if the dimension of the treatment variable $T_{n,i}$ is one ($d_T=1$) because no contamination occurs when $d_{T}=1$.
    This result is consistent with \cite{borusyak2024negative}, but our result in \cref{cor:causal_element} is more general ($d_T>1$).
\end{remark}

\subsection{When Can We Avoid the Contamination Bias?}
The following statement provides sufficient conditions to avoid contamination bias.
Define the conditional covariance for randomvariables $W_1$ and $W_2$ given $\bR_n$ as $\Cov(W_1,W_2|\bR_n)=\mathbb{E}[(W_1-\mathbb{E}[W_1|\bR_n])(W_2-\mathbb{E}[W_2|\bR_n])|\bR_n]$.

\begin{corollary}
    \label{cor:no_contamination}
    Assume that \cref{asm:sampling,asm:linear,asm:moment,asm:linear_propensity} and $d_{\tilde{T}}=d_T$ hold.
    Suppose that $\mathbb{E}[\Cov(T_{n,i,(k)},T_{n,i,(l)}|\bR_n)]=0$ for all $i\in\mathcal{N}_{n}$ and for any $l\neq k$.
    Then, for large enough $n$, there is no contamination bias for $\theta_{n,(k)}^{\causal}$, i.e., 
    \begin{align*}
        \theta_{n,(k)}^{\causal}=\frac{\sum_{i=1}^n\mathbb{E}[X_{n,i,(k)}^2]\theta_{n,i,(k)}}{\sum_{i=1}^n\mathbb{E}[X_{n,i,(k)}^{2}]}
    \end{align*}
    for each $k=1,...,d_{T}$. Suppose that  $\Cov(\tilde{T}_{n,i,(k)},T_{n,i,(l)}|\bR_{n})=0$ for all $i\in\mathcal{N}_{n}$ with $R_{n,i}=1$ and for any $l\neq k$.
    Then, for large enough $n$, there is no contamination bias for $\theta_{n,(k)}^{\causample}$, i.e.,
    \begin{align*}
            \theta_{n,(k)}^{\causample}
            =&\frac{\sum_{i=1}^nR_{n,i}\mathbb{E}[\tilde{X}_{n,i,(k)}X_{n,i,(k)}|\bR_{n}]\theta_{n,i,(k)}}{\sum_{i=1}^nR_{n,i}\mathbb{E}[\tilde{X}_{n,i,(k)}^{2}|\bR_{n}]}\quad\text{a.s.}
    \end{align*}
    for each $k=1,...,d_{T}$.

    The weights of $\theta_{n,(k)}^{\causal}$ for $\theta_{n,i,(k)}$ are always non-negative.
    If we further assume that
    $\Cov(\tilde{T}_{n,i,(k)},T_{n,i,(k)}|\bR_{n})\geq0$ a.s. for all $i\in\mathcal{N}_{n}$ with $R_{n,i}=1$ and for all $k=1,...,d_{T}$, then the weights of $\theta_{n,(k)}^{\causample}$ for $\theta_{n,i,(k)}$ are non-negative, i.e., 
    \begin{equation*}
        \frac{R_{n,i}\mathbb{E}[\tilde{X}_{n,i,(k)}X_{n,i,(k)}|\bR_{n}]}{\sum_{i=1}^nR_{n,i}\mathbb{E}[\tilde{X}_{n,i,(k)}^{2}|\bR_{n}]}\geq0\quad\text{a.s.}
    \end{equation*}
    for all $i\in\mathcal{N}_{n}$ and each $k=1,...,d_{T}$.
\end{corollary}

The zero conditional covariance assumption is satisfied if elements of $T_{n,i}$ and $\tilde{T}_{n,i}$ are mutually independent.
The positive conditional covariance assumption is satisfied under the censored network (see \cref{ex:censor_contamination} below).

\begin{remark}
    Under homogeneous treatment effects $\theta_{n,i}=\theta_{n}$, we have $\theta_{n}^{\causal}=\theta_{n}$, but 
    \begin{equation*}
        \theta_{n}^{\causample}=\theta_{n}-\left(\sum_{i=1}^{n}R_{n,i}\mathbb{E}[\tilde{X}_{n,i}\tilde{X}_{n,i}'|\bR_{n}]\right)^{-1}\sum_{i=1}^{n}R_{n,i}\mathbb{E}[\tilde{X}_{n,i}(X_{n,i}-\tilde{X}_{n,i})'|\bR_{n}]\theta_{n}.
    \end{equation*}
    Thus, $\theta_{n}^{\causal}$ does not have contamination bias for homogeneous treatment effects, but $\theta_{n}^{\causample}$ does.
    Under homogeneous treatment effects and $X_{n,i}=\tilde{X}_{n,i}$, we have $\theta_{n}^{\causal}=\theta_{n}^{\causample}=\theta_{n}$.
\end{remark}

\begin{example}
    \label{ex:censor_contamination}
    Consider the exposure mapping in \cref{ex:censor}. 
    The misspecified exposure mapping is $\tilde{T}_{n,i}=\tilde{g}(i,\bD_{n},\tilde{\bA}_{n})=\mathds{1}\left\{\sum_{j\neq i}C_{n,i,j}\tilde{A}_{n,i,j}R_{n,j}D^*_{n,j}>0\right\}$.
    Assume that $D^*_{n,i} \sim \text{Bernoulli}(p_{n})$ for $i = 1, \dots, n$ independently.
    By adapting \cref{cor:no_contamination}, $\theta_{n}^{\causample}$ is a convex combination of $\theta_{n,i}$. 
    Indeed, we can calculate
    \begin{align*}
            \theta_{n}^{\causample}
            =&\frac{\sum_{i=1}^nR_{n,i}\left(1-(1-p_n)^{\sum_{j\neq i}C_{n,i,j}\tilde{A}_{n,i,j}R_{n,j}}\right)\theta_{n,i,(1)}}{\sum_{i=1}^nR_{n,i}\left(1-(1-p_n)^{\sum_{j\neq i}C_{n,i,j}\tilde{A}_{n,i,j}R_{n,j}}\right)},
    \end{align*}
    and the weights are non-negative.
    In general, if both mappings 
    $\tilde{T}_{n,i,(k)}$ and $T_{n,i,(k)}$ are weakly increasing (or both weakly decreasing) in $\{D^*_{n,i}\}_{i\in\mathcal{N}_n}$, then the weights are non-negative.
    Thus, censoring does not cause negative weight problems when $g$ is weakly monotone on $\{D^*_{n,i}\}_{i\in\mathcal{N}_n}$ for the first neighborhood exposure mapping.    
\end{example}

\subsection{More Examples}
\begin{example}
    \label{ex:q_function}
    Consider a general form of exposure mapping. For some function $q:\mathbb{R}^2\to\mathbb{R}$, let $T_{n,i}\allowbreak=(R_{n,i}D^*_{n,i},\allowbreak q(\sum_{j\neq i}A_{n,i,j}R_{n,j}D^*_{n,j},\allowbreak\sum_{j\neq i}A_{n,i,j}R_{n,j}))$ and 
    $\tilde{T}_{n,i}\allowbreak=(R_{n,i}D^*_{n,i},\allowbreak q(\sum_{j\neq i}\tilde{A}_{n,i,j}R_{n,j}D^*_{n,j},\allowbreak\sum_{j\neq i}\tilde{A}_{n,i,j}R_{n,j}))$. For example, the share of treated friends is covered by the following $q$:
    \begin{equation*}
        q\left(\sum_{j\neq i}A_{n,i,j}R_{n,j}D^*_{n,j},\sum_{j\neq i}A_{n,i,j}R_{n,j}\right)=\frac{\sum_{j\neq i}A_{n,i,j}R_{n,j}D^*_{n,j}}{\sum_{j\neq i}A_{n,i,j}R_{n,j}}.
    \end{equation*} 
    It also covers the indicator function as in \cref{ex:withoutmiss}.
    Since $D^*_{n,i}\indep D^*_{n,j}$, this satisfies the no-correlation conditions. If $q$ is non-decreasing with respect to the first argument, then $\tilde{T}_{n,i}$ and $T_{n,i}$ are positively correlated, giving $\theta_{n}^{\causample}$ a clear causal interpretation. This type of exposure mapping is used in \cite{cai2015social} and \cite{carter2021subsidies}.
    As we illustrated above in the special case, the censoring $\tilde{T}_{n,i}=(R_{n,i}D^*_{n,i},q(\sum_{j\neq i}C_{n,i,j}\tilde{A}_{n,i,j}R_{n,j}D^*_{n,j},\sum_{j\neq i}C_{n,i,j}\tilde{A}_{n,i,j}R_{n,j}))$ does not cause negative weight problems since the exposure mapping $g$ is weakly monotone on $\{D^*_{n,i}\}_{i\in\mathcal{N}_n}$.
\end{example}

\begin{example}
    \label{ex:aronow}
    Let 
    $T_{n,i}=(R_{n,i}D^*_{n,i}G_{n,i}, R_{n,i}D^*_{n,i}(1-G_{n,i}),(1-R_{n,i}D^*_{n,i})G_{n,i})$, where $G_{n,i}=\mathds{1}\{\sum_{j\neq i}A_{n,i,j}\allowbreak R_{n,j}D^*_{n,j}>0\}$. The exposure mapping categorizes each unit $i$ into one of three mutually exclusive exposure types, based on their own treatment status and the presence of treated friends.\footnote{The slope of the OLS estimator captures the effect associated with the group of units that are untreated and have no treated friends.} The elements are mutually exclusive but dependent, so the no-correlation conditions are violated, and we have a contamination bias. This exposure mapping is used in \cite{aronow2017estimating}. For the exposure mapping with dependence among its elements, we recommend using the inverse propensity score weighting (IPW) estimators to avoid contamination bias.
\end{example}

\begin{remark}(Comparison with IPW estimators)
    The causal estimand for the IPW estimators is the average treatment effect (ATE), $(1/n)\sum_{i=1}^n Y_{n,i}^*(t)$ for each $t$.
    In other words, the IPW estimator and the regression estimator are for different causal estimands.
    While the IPW estimator works well for cases like \cref{ex:aronow}, it is not suitable for cases like \cref{ex:cai} because the overlapping condition of the propensity score is easily violated.
    For example, suppose that $T_{n,i}$ is the treated friends share $(\sum_{j\neq i}A_{n,i,j}R_{n,j}D^*_{n,j})/(\sum_{j\neq i}A_{n,i,j})$, and there are two units having three and two friends in the population network, respectively. The former can take $T_{n,i}=1/3$ with positive probability, but the latter never takes the value. Thus, the overlapping condition fails to hold. 
    Moreover, the overlapping condition can be violated in the sampled network even if it is satisfied in the population network, since the sampled network is a sub-network of the population one.

    The choice between the IPW estimator and the regression should be decided by the exposure mapping formula that the researcher wants to use.
    We recommend using the IPW estimators to avoid contamination bias when the overlapping condition is satisfied.
    On the other hand, if there is any doubt about the overlapping condition or the exposure mapping takes (nearly) continuous values, we suggest using the regression model since it does not require the overlapping condition. We leave a more detailed comparison between the IPW estimator and the OLS estimator for future research.
\end{remark}

\begin{example}
    \label{ex:cai}
    Consider an exposure mapping
    $$T_{n,i}=\left(R_{n,i}D^*_{n,i},\frac{\sum_{j\neq i}A_{n,i,j}R_{n,j}D^*_{n,j}}{\sum_{j\neq i}A_{n,i,j}},\frac{\sum_{j\neq i}\sum_{k\neq i,j}A_{n,i,j}A_{n,j,k}R_{n,k}D^*_{n,k}}{\sum_{j\neq i}\sum_{k\neq i,j}A_{n,i,j}A_{n,j,k}}\right),$$
    where the first element indicates whether unit $i$ is directly treated or not, the second element captures the treated friends share among $i$'s first neighbors, and the third element captures the treated friends share among $i$'s second neighbors.
    There are overlaps in $\bD_{n}$ in the second and third elements if there are triangles in the network, so no-correlation conditions are generally violated. \cref{fig:sub_triangle} shows an example of a network with triangles. The second element of $T_{n,i}$ is the average of the neighbors' treatment status including $D_{n,i_{1}}$ and $D_{n,i_{2}}$. The third element is the average of the first neighbors' treatment status, including $D_{n,i_{1}}$ and $D_{n,i_{2}}$, again. Thus, the second and third elements are correlated.
    This setting is employed in \cite{cai2015social}.
    An easy way to avoid contamination bias is to modify the exposure mapping $g$ to eliminate the double counting. For example, we can use 
    \begin{align}
        T_{n,i}=\left(R_{n,i}D^*_{n,i},\frac{\sum_{j\neq i}A_{n,i,j}R_{n,j}D^*_{n,j}}{\sum_{j\neq i}A_{n,i,j}},\frac{\sum_{j\neq i}\sum_{k\neq i,j}A_{n,i,j}A_{n,j,k}(1-A_{n,i,k})R_{n,k}D^*_{n,k}}{\sum_{j\neq i}\sum_{k\neq i,j}A_{n,i,j}A_{n,j,k}(1-A_{n,i,k})}\right),
        \label{eq:triangle}
    \end{align} 
    instead.
    Although we miss some of the second-order links, we still manage to avoid the double counting and hence contamination bias.
    \begin{figure}[h]
        \caption{Networks with triangle links}
        \centering
        
        \begin{subfigure}[b]{0.38\textwidth}
            \centering
            \begin{tikzpicture}[
                scale=.75,
                transform shape,
                every node/.style={
                    circle,
                    draw=blue,
                    fill=blue!30,
                    minimum size=1cm
                }
            ]
                \node (center) at (0,0) {$i$};
    
                \node (A) at (2.7,0) {$i_1$};
                \node (B) at (-2.5,-1) {};
                \node (C) at (2,2) {$i_2$};
                \node (D) at (-1,2) {};
    
                \draw (center) -- (A);
                \draw (center) -- (B);
                \draw (center) -- (C);
                \draw (center) -- (D);
    
                \node (A1) at (5,-0.8) {};
                \node (A2) at (5,0.8)  {};
                \draw (A) -- (A1);
                \draw (A) -- (A2);
                \draw (A1) -- (A2);
                \draw (A) -- (C);
    
                \draw (A1) -- ++(1,-0.2);
                \draw (A1) -- ++(1,0.2);
                \draw (A2) -- ++(1,0.7);
                \draw (center) -- ++(0,-1.5);
                \draw (A) -- ++(0.3,-1.5);
                \draw (B) -- ++(-1,-0.2);
                \draw (C) -- ++(0.2,1);
                \draw (D) -- ++(0.2,0.9);
                \draw (D) -- ++(-2,0.7);
            \end{tikzpicture}
            \caption{Without censored links}
            \label{fig:sub_triangle}
        \end{subfigure}
        \hspace{50pt}
        \begin{subfigure}[b]{0.38\textwidth}
            \centering
            \begin{tikzpicture}[
                scale=.75,          
                transform shape,    
                every node/.style={
                    circle,
                    draw=blue,
                    fill=blue!30,
                    minimum size=1cm
                }
            ]
                \node (center) at (0,0) {$i$};
    
                \node (A) at (2.7,0) {$i_1$};
                \node (B) at (-2.5,-1) {};
                \node (C) at (2,2) {$i_2$};
                \node (D) at (-1,2) {};
    
                \draw (center) -- (A);
                \draw (center) -- (B);
                \draw[dashed, gray] (center) -- (C);   
                \draw (center) -- (D);
    
                \node (A1) at (5,-0.8) {};
                \node (A2) at (5,0.8)  {};
                \draw (A) -- (A1);
                \draw (A) -- (A2);
                \draw (A1) -- (A2);
                \draw (A) -- (C);
    
                \draw (A1) -- ++(1,-0.2);
                \draw (A1) -- ++(1,0.2);
                \draw (A2) -- ++(1,0.7);
                \draw (center) -- ++(0,-1.5);
                \draw[dashed, gray] (A) -- ++(0.3,-1.5);
                \draw (B) -- ++(-1,-0.2);
                \draw (C) -- ++(0.2,1);
                \draw (D) -- ++(0.2,0.9);
                \draw (D) -- ++(-2,0.7);
            \end{tikzpicture}
            \caption{With censored (dashed) links}
            \label{fig:sub_triangle_censor}
        \end{subfigure}
        
        \label{fig:triangle_combined}
    \end{figure}

\end{example}

\begin{example}
    \label{ex:cai_censor}
    Consider the setup in \cref{ex:cai} but with censoring caused by naming up to four friends.
    As illustrated in \cref{fig:sub_triangle_censor}, suppose that the sampled network link between $i_1$ and $i_2$ is not observed due to the censoring. Then, $i_2$ is misclassified as a second neighborhood friend in the observed network while $i_2$ is a first neighborhood friend in the population network.
    Thus, if we consider the true exposure mapping $T_{n,i}$ as in \cref{eq:triangle}, and a misspecified exposure mapping for the sampled network
    \begin{align*}
        \tilde{T}_{n,i}=&\left(R_{n,i}D^*_{n,i},
        \frac{\sum_{j\neq i}C_{n,i,j}\tilde{A}_{n,i,j}R_{n,j}D^*_{n,j}}{\sum_{j\neq i}C_{n,i,j}\tilde{A}_{n,i,j}},\right.\\
        &\left.\qquad
        \frac{\sum_{j\neq i}\sum_{k\neq i,j}C_{n,i,j}\tilde{A}_{n,i,j}C_{n,j,k}\tilde{A}_{n,j,k}(1-C_{n,i,k}\tilde{A}_{n,i,k})R_{n,k}D^*_{n,k}}{\sum_{j\neq i}\sum_{k\neq i,j}C_{n,i,j}\tilde{A}_{n,i,j}C_{n,j,k}\tilde{A}_{n,j,k}(1-C_{n,i,k}\tilde{A}_{n,i,k})}\right),
    \end{align*}
    then, there is a correlation between $T_{n,i,(2)}$ and $\tilde{T}_{n,i,(3)}$.
    An easy way to avoid contamination bias is to modify the exposure mapping $\tilde{g}$ so that $\tilde{T}_{n,i,(3)}$ equals zero for individuals subject to censoring.
    For example, if the censoring happens by asking up to four friends, we can eliminate the individuals with four observed links from consideration
    \begin{equation*}
        \tilde{T}_{n,i,(3)}=\frac{\sum_{j\neq i}\sum_{k\neq i,j}C_{n,i,j}\tilde{A}_{n,i,j}C_{n,j,k}\tilde{A}_{n,j,k}(1-C_{n,i,k}\tilde{A}_{n,i,k})R_{n,k}D^*_{n,k}}{\sum_{j\neq i}\sum_{k\neq i,j}C_{n,i,j}\tilde{A}_{n,i,j}C_{n,j,k}\tilde{A}_{n,j,k}(1-C_{n,i,k}\tilde{A}_{n,i,k})}\mathds{1}\left\{\sum_{j\neq i}C_{n,i,j}\tilde{A}_{n,i,j}<4\right\}.
    \end{equation*}
    Note that the censoring for $i$ does not matter for the first neighborhood element $\tilde{T}_{n,i,(2)}$ by the same logic as \cref{ex:q_function}. Moreover, the censoring for $i_1$ does not matter for the second neighborhood element $\tilde{T}_{n,i,(3)}$ of $i$ because it does not introduce any misclassification.
\end{example}

\subsection{Asymptotic Theory}
We mostly follow the notation of \cite{kojevnikov2021limit}. Let $\mathcal{N}_{n}=\{1,...,n\}$ be the set of population units and $d_{n}(i,j)$ be the shortest distance between $i,j\in \mathcal{N}_{n}$ on $\bA_{n}$ (set $d_{n}(i,i)=0$; set $d_{n}(i,j)=\infty$ if there are no paths between $i$ and $j$). Define $\mathcal{L}_{v}=\{\mathcal{L}_{v,a}:a\in\mathbb{N}\}$, where $\mathcal{L}_{v,a} = \{f:\mathbb{R}^{v\times a}\to \mathbb{R}:\|f\|_{\infty}<\infty,\Lip(f)<\infty\}$, $\|\cdot\|_{\infty}$ is the sup-norm, and $\Lip(f)$ is the Lipschitz constant of $f$. Let $\mathcal{P}_{n}(a,b;s)=\{(A,B):A,B\subset \mathcal{N}_{n},|A|=a,|B|=b,d_{n}(A,B)\geq s\}$, where $d_{n}(A,B)=\min_{i\in A}\min_{j\in B}d_{n}(i,j)$.
For each $A\subset \mathcal{N}_{n}$ and triangular array $(U_{n,i})$, let us write $U_{n,A}=(U_{n,i})_{i\in A}$.
\begin{definition}
    A triangular array $\{U_{n,i}\}, n\geq 1,U_{n,i}\in \mathbb{R}^{v}$, is called {\it conditionally $\psi$-dependent given $\bR_n$}, if for each $n\in\mathbb{N}$, there exists a $\sigma(\bR_n)$-measurable sequence $\xi_{n}=\{\xi_{n,s}\}_{s\geq 0},\xi_{n,0}=1$, and a collection of nonrandom functions $(\psi_{a,b})_{a,b\in\mathbb{N}},\psi_{a,b}:\mathcal{L}_{v,a}\times\mathcal{L}_{v,b}\to[0,\infty)$ such that for all $(A,B)\in\mathcal{P}_{n}(a,b;s)$ with $s>0$ and all $f\in\mathcal{L}_{v,a}$ and $g\in\mathcal{L}_{v,b}$,
    \begin{align*}
        |\Cov(f(U_{n,A}),g(U_{n,B}))|\leq \psi_{a,b}(f,g)\xi_{n,s} \quad\text{a.s.}
    \end{align*}
\end{definition}

Define
\begin{align*}
    \mathcal{N}_{n}(i;s) =\{j\in \mathcal{N}_{n}:d_{n}(i,j)\leq s\},
\end{align*}
which is the set of $i$'s neighborhood within $s$-distance.
First, we assume that the network dependence of the exposure mappings is local. 
\begin{assumption}\label{asm:local}
    There exists some $K\in\mathbb{N}$ such that for any $i\in \mathcal{N}_{n},n\in\mathbb{N}$ and $\bd_{n},\bd_{n}^{\prime}\in\{0,1\}^{n}$ such that $\bd_{n,\mathcal{N}_{n}(i,K)}=\bd'_{n,\mathcal{N}_{n}(i,K)}$,
    \begin{align*}
        g(i,\bd_{n},\bA_{n})&=g(i,\bd_{n}^{\prime},\bA_{n}),\quad\text{ and }\quad
        \tilde{g}(i,\bd_{n},\tilde{\bA}_{n})=\tilde{g}(i,\bd_{n}^{\prime},\tilde{\bA}_{n})\quad\text{a.s.}
    \end{align*}
\end{assumption}

Let $\tilde{d}_{n}(i,j)$ be the shortest distance between $i,j\in \mathcal{N}_{n}$ on $\tilde{\bA}_{n}$.
\cref{asm:local,asm:sampling} imply that $T_{n,i}\indep T_{n,j}$ if $d_{n}(i,j)>2K$. They also imply that $\tilde{T}_{n,i}\indep \tilde{T}_{n,j}$ if $d_{n}(i,j)>2K$ because $\tilde{d}_{n}(i,j)\geq d_{n}(i,j)$ almost surely and because $i$ and $j$ do not share $R_{n,k}$ and $D^*_{n,k}$ for any $k\neq i,j$ in their $K$-neighborhoods.

Under the correctly specified exposure mapping, $g=\tilde{g}$, the condition $g(i,\bd_{n},\bA_{n})=g(i,\bd_{n}^{\prime},\bA_{n})$ automatically implies $\tilde{g}(i,\bd_{n},\tilde{\bA}_{n})=\tilde{g}(i,\bd_{n}^{\prime},\tilde{\bA}_{n})$ a.s. because the distance on a sampled network is always weakly longer than that on the population network: $\tilde{d}_{n}(i,j)\geq d_{n}(i,j)$. For the same reason, the distance on a censored network is always weakly longer than that on sampled or population networks.

Define $\mathcal{N}_{n}^{\partial}(i;s)=\{j\in \mathcal{N}_{n}:d_{n}(i,j)=s\}$, which is the set of $i$'s neighborhood with exact $s$-distance, and its $p$-th sample moment $\delta_{n}^{\partial}(s;p)=n^{-1}\sum_{i\in \mathcal{N}_{n}}|\mathcal{N}_{n}^{\partial}(i;s)|^{p}$.
The next assumption requires that the sum of these $p$-th sample moments within $2K$-distance is bounded.
\begin{assumption}
    \label{asm:sparsity}
    The sequence of networks $(\bA_{n})$ satisfies
    \begin{align*}
        \sum_{1\leq s\leq 2K}\delta_{n}^{\partial}(s;1)=O(1).
    \end{align*}
\end{assumption}
By a simple calculation and $\rho>0$, we can show that \cref{asm:sparsity} is equivalent to $(n\rho_n)^{-1}\sum_{i=1}^{n}\allowbreak\sum_{j\in\mathcal{N}_n(i;2K)}1=O(1)$. Also note that \cref{asm:sparsity} is weaker than the bounded network degree since this assumption only requires the boundedness on average.

Then, we show that our estimator is consistent for the sample-level causal estimand:
\begin{theorem}\label{thm:consistency}
    Under \cref{asm:linear,asm:sampling,asm:moment,asm:linear_propensity,asm:local,asm:sparsity}, $$\hat{\theta}_{n}-\theta_{n}^{\causample}\overset{p^R}{\longrightarrow}0
    \quad\text{ and }\quad \hat{\theta}_{n}-\theta_{n}^{\causample}\overset{p}{\to}0,$$
    where $\overset{p^R}{\longrightarrow}$ denotes convergence in probability conditional on $\bR_n$, that is, for any $\varepsilon>0$,
    $$
        \mathbb{P}\left(\|\hat{\theta}_{n}-\theta_{n}^{\causample}\|\leq\varepsilon\mid\bR_n\right)\overset{a.s.}{\longrightarrow}1
    $$ 
    as $n\to\infty$.
\end{theorem}
\cref{thm:consistency} establishes the internal validity of our network experiment. However, in general, $\hat{\theta}_{n}-\theta_{n}^{\causal}\not\overset{p}{\to}0$ because $\theta^{\causal}_{n}-\theta_{n}^{\causample}\not\overset{p}{\to}0$ due to misspecification of the exposure mapping. Moreover, as shown in \cref{cor:causal_element}, $\theta_{n}^{\causample}$ does not have a clear causal interpretation. Consequently, \cref{thm:consistency} does not guarantee the external validity of our network experiment.

Ideally, our network experiment would satisfy $\hat{\theta}_{n}-\theta_{n}^{\causal}\overset{p}{\to}0$ so that each element of $\hat{\theta}_{n}$ can be interpreted as a causal spillover effect. 
We show that this consistency is achieved when there is no misspecification and no mismeasurement ($\tilde{T}_{n,i}=T_{n,i}$ for each $i\in\mathcal{N}_{n}$) and the observed covariates coincide with those in the population ($\tilde{Z}_{n,i}=Z_{n,i}$ for each $i\in\mathcal{N}_{n})$. 
We are essentially assuming that each $\tilde{T}_{n,i}$ is computed by $g(i,\bD_{n},\bA_{n})=T_{n,i}$ where we replace $\tilde{g}$ with $g$ and $\tilde{\bA}_{n}$ with $\bA_{n}$. 
Under the linear propensity scores, we can show that $X_{n,i}=\tilde{X}_{n,i}$ a.s. (\cref{lem:Lambda}).

\begin{assumption}
    \label{asm:exposure}
    \
    \begin{enumerate}
        \item We have the following equalities almost surely for $R_{n,i}=1$: $\tilde{T}_{n,i}=T_{n,i}$ and $\tilde{Z}_{n,i}=Z_{n,i}$ for all $i\in\mathcal{N}_{n}$ and $n\in\mathbb{N}$.
        \item Each element of $T_{n,i}$ and $Z_{n,i}$ either does not depend on $R_{n,i}$, or depends on it only through a multiplicative form.
        \item At most one element of $T_{n,i}$ depends on $i$'s own treatment $R_{n,i}D^*_{n,i}$ and the element does not depend on $R_{n,j}$ and $D_{n,j}$ for any $j\neq i$.
    \end{enumerate}
\end{assumption}

\cref{asm:exposure} (i) holds when there is no misspecification and no mismeasurement for sampled units, i.e., $g = \tilde{g}$ and $\tilde{\bA}_n = \bA_n$ locally. For example, under star sampling with an exposure mapping restricted to the first neighborhood, all relevant links are correctly observed for sampled units. However, \cref{asm:exposure} (i) may not hold for exposure mappings with higher-order dependence, since more global measurement of the network is then required up to the relevant order. Nonetheless, in such cases, researchers can apply our results to $\theta_{n}^{\causample}$ instead of $\theta_{n}^{\causal}$ and interpret it as a convex combination of heterogeneous treatment effects.
\cref{asm:exposure} (ii) means some components of the covariate vector are independent of the sampling indicator, while others incorporate $R_{n, i}$ in a multiplicative way---for example, $Z_{n,i,(k)}=R_{n,i}p_{n,i}$.
We can always pick covariates $Z_{n,i}$ having \cref{asm:exposure} (ii) since $R_{n,i}$ enters only multiplicatively for $T_{n,i}$ by \cref{asm:sampling} if we include the direct effect without any transformation. Thus, we can choose covariates $Z_{n,i}$ satisfying \cref{asm:exposure} and \cref{asm:linear_propensity} simultaneously.
\cref{asm:exposure} (iii) is satisfied if we do not include the cross term of the direct effect $R_{n,i}D^*_{n,i}$ and a spillover effect.
Excluding the cross term is also used to guarantee no contamination (\cref{cor:no_contamination}).\footnote{We can allow the violation of \cref{asm:exposure} (iii) if we modify $\hat{\theta}_n$ in the same manner as $\tilde{\gamma}_n$ in \cref{app:gamma}.}

Under \cref{asm:exposure}, we can show the consistency of the OLS estimator $\hat{\theta}_{n}$ for the population-level causal estimand $\theta_{n}^{\causal}$:
\begin{theorem}
    \label{thm:causal_consistency}
    Under \cref{asm:linear,asm:sampling,asm:moment,asm:linear_propensity,asm:local,asm:sparsity,asm:exposure}, $$\hat{\theta}_{n}-\theta_{n}^{\causal}\overset{p}{\to}0.$$
\end{theorem}

It is worth noting that \cref{thm:causal_consistency} does not hold if $\tilde{Z}_{n,i}\neq Z_{n,i}$, since we cannot ensure $\tilde{X}_{n,i}\sim X_{n,i}$ asymptotically. Instead, under no misspecification, \cref{cor:causal_element} implies
\begin{align*}
    \theta_{n,(k)}^{\causample}=\frac{\sum_{i=1}^nR_{n,i}\mathbb{E}[\tilde{U}_{n,i,(k)}^{2}|\bR_{n}]\theta_{n,i}}{\sum_{i=1}^nR_{n,i}\mathbb{E}[\tilde{U}_{n,i,(k)}^{2}|\bR_{n}]}
\end{align*}
for each $k=1,...,d_{\tilde{T}}$. Thus, although the consistency for $\theta_{n}^{\causal}$ may fail in this setting, the absence of misspecification alone recovers the causal interpretability of $\theta_{n}^{\causample}$, and by extension, that of $\hat{\theta}_{n}$.

Next, we consider the asymptotic distribution of $\hat{\theta}_{n}$. 
Now, we introduce additional dependence measures of the network.
Define $\Delta_{n}(s,m;k)=\frac{1}{n}\sum_{i\in \mathcal{N}_{n}}\max_{j\in \mathcal{N}_{n}^{\partial}(i;s)}|\mathcal{N}_{n}(i;m)\setminus \mathcal{N}_{n}(j;s-1)|^{k}$, and $c_{n}(s,m;k)=\inf_{\alpha >1}[\Delta_{n}(s,m;k\alpha)]^{1/\alpha}\left[\delta^{\partial}_{n}\left(s;\frac{\alpha}{\alpha-1}\right)\right]^{1-1/\alpha}$.
$c_{n}(s,m;k)$ measures the density of the network and is used as a sufficient condition for the CLT.

Define
\begin{align*}
    \tilde{\varepsilon}_{n,i}&=Y_{n,i}-\tilde{X}_{n,i}'\theta_{n}^{\causample}-\tilde{Z}_{n,i}'\gamma_{n}^{\causample},\\
    \varepsilon_{n,i}&=Y_{n,i}-X_{n,i}'\theta_{n}^{\causal}-Z_{n,i}'\gamma_{n}^{\causal},
\end{align*}
and 
\begin{align*}
    \tilde{\Sigma}_{n}
    =\Var\left(\sum_{i=1}^{n}R_{n,i}\tilde{X}_{n,i}\tilde{\varepsilon}_{n,i}\mid\bR_n\right),\qquad
    \Sigma_{n}
    =\Var\left(\sum_{i=1}^{n}R_{n,i}X_{n,i}\varepsilon_{n,i}\right).
\end{align*}

We impose the following assumption, which requires a weak dependence structure in the network and rules out overly dense networks.

\begin{assumption}
    \label{asm:dependence}
    There exists a positive sequence $m_{n}\to\infty$ such that for $p=1,2$,
    \begin{align*}
        n\tilde{\Sigma}_{n}^{-(1+p/2)}\sum_{s=0}^{2K}c_{n}(s,m_{n};p)\overset{a.s.}{\longrightarrow} 0,\qquad
        n\Sigma_{n}^{-(1+p/2)}\sum_{s=0}^{2K}c_{n}(s,m_{n};p)\to 0.
    \end{align*}
\end{assumption}
Then, we show that $\hat{\theta}_{n}$ is asymptotically normal relative to $\theta_{n}^{\causample}$:
\begin{theorem}
    \label{thm:normality}
    Under \cref{asm:linear,asm:sampling,asm:moment,asm:linear_propensity,asm:local,asm:sparsity,asm:dependence},
    \begin{align*}
        \tilde{\Sigma}_{n}^{-1/2}\tilde{Q}_n^{XX}(\hat{\theta}_{n}-\theta_{n}^{\causample})\overset{d^R}{\longrightarrow}\mathrm{N}(0,I_{d_{\tilde{T}}})
        \quad\text{and}\quad
        \tilde{\Sigma}_{n}^{-1/2}\tilde{Q}_n^{XX}(\hat{\theta}_{n}-\theta_{n}^{\causample})\overset{d}{\to}\mathrm{N}(0,I_{d_{\tilde{T}}}),
    \end{align*}
    where $\overset{d^R}{\to}$ denotes convergence in distribution conditional on $\bR_n$, that is,
    $$
        \left|\mathbb{P}\left(\tilde{\Sigma}_{n}^{-1/2}\tilde{Q}_n^{XX}(\hat{\theta}_{n}-\theta_{n}^{\causample}) \leq t\mid\bR_n\right)-F(t)\right|\overset{a.s.}{\longrightarrow}0
    $$
    as $n\to\infty$ for any $t\in\mathbb{R}^{d_{\tilde{T}}}$ letting $F(t)$ be the distribution function of $\mathrm{N}(0,I_{d_{\tilde{T}}})$.
\end{theorem}

We also show that the absence of misspecification and access to the variables in the population yield asymptotic normality of $\hat{\theta}_{n}$ relative to $\theta_{n}^{\causal}$:
\begin{theorem}
    \label{thm:normality_causal}
    Under \cref{asm:exposure,asm:linear,asm:sampling,asm:moment,asm:linear_propensity,asm:local,asm:sparsity,asm:dependence}, we have
    \begin{align*}
        \Sigma_{n}^{-1/2}\tilde{Q}^{XX}_{n}(\hat{\theta}_{n}-\theta_{n}^{\causal})\overset{d}{\to}\mathrm{N}(0,I_{d_{T}}).
    \end{align*}
\end{theorem}

\begin{remark}
    When we have a homogeneous effect $\theta_{n,i}=\theta_{n}$, we have $\theta_{n}^{\causample}=\theta_{n}^{\causal}\quad\text{a.s.}$
    for large enough $n$ under $X_{n,i}=\tilde{X}_{n,i}$.
    Hence, we can use the same asymptotic distribution among them.
\end{remark}

\section{Variance Estimation}
\label{sec:variance}
In this section, we provide a conservative network heteroskedasticity- and autocorrelation-consistent (HAC) variance estimator for $\hat{\theta}_{n}$.
Note that even when treatments and samples are randomly assigned and drawn, dependence can persist within a $2K$-neighborhood because exposure mappings $T_{n,i}$ may share elements of $\bD_{n}$ and $\bR_{n}$. As a result, the variance estimator must account for this local dependence structure. However, for any pair $i, j$ with $d_n(i, j) > 2K$, the exposure mappings $T_{n,i}$ and $T_{n,j}$ are independent. When the exposure mapping is correctly specified ($\tilde{g}=g$), the researcher can directly choose a finite $K$ based on the functional form. If there is potential misspecification in $\tilde{g}$, $K$ should be selected conservatively, reflecting the maximum range over which the exposure mapping may induce dependence.

Define
\begin{align*}
    \tilde{\mathcal{N}}_{n}(i;s) =\{j\in \mathcal{N}_{n}:\tilde{d}_{n}(i,j)\leq s\},
\end{align*}
which is the set of $i$'s neighborhood within $s$-distance on a sampled network $\tilde{\bA}_{n}$.
Note that $\tilde{\mathcal{N}}_{n}(i;s)$ is a random set because $\tilde{d}_{n}(i,j)$ is a random variable depending on $\bR_n$.
On the other hand, $d_{n}(i,j)$ and $\mathcal{N}_{n}(i;s)$ are non-random.
Recall that we also have $\tilde{d}_{n}(i,j)\geq d_{n}(i,j)$ a.s., thus, $\tilde{\mathcal{N}}_{n}(i;s)\subseteq \mathcal{N}_{n}(i;s)$ a.s.

Let 
\begin{align*}
    \hat{\varepsilon}_{n,i}&=Y_{n,i}-\tilde{X}_{n,i}^{\prime} \hat{\theta}_n-\tilde{Z}_{n,i}^{\prime}\tilde{\gamma}_n,
\end{align*}
$\Psi_{n,i}=X_{n,i}\varepsilon_{n,i}$, $\tilde{\Psi}_{n,i}=\tilde{X}_{n,i}\tilde{\varepsilon}_{n,i}$, and $\hat{\Psi}_{n,i}=\tilde{X}_{n,i}\hat{\varepsilon}_{n,i}$, where we define $\tilde{\gamma}_n$ later in \cref{thm:var_consistency,thm:var_consistency_eig}.
By orthogonality conditions, $\sum_{i=1}^{n}\mathbb{E}\left[\Psi_{n,i}\right]=0$, $\sum_{i=1}^{n}R_{n,i}\mathbb{E}\left[\tilde{\Psi}_{n,i}\mid\bR_n\right]=0$, and $\sum_{i=1}^{n}R_{n,i}\hat{\Psi}_{n,i}=0$.

Then, the variances of interest can be written as
\begin{align*}
    \frac{1}{n\rho_n}\tilde{\Sigma}_{n}
    =&\Var\left(\frac{1}{\sqrt{n\rho_n}}\sum_{i=1}^{n}R_{n,i}\tilde{X}_{n,i}\tilde{\varepsilon}_{n,i}\mid\bR_n\right)\\
    =&\frac{1}{n\rho_n}\sum_{i=1}^{n}\sum_{j\in\tilde{\mathcal{N}}_{n}(i,2K)}R_{n,i}R_{n,j}\mathbb{E}\left[\left(\tilde{\Psi}_{n,i}-\mathbb{E}\left[\tilde{\Psi}_{n,i}\mid\bR_n\right]\right)\left(\tilde{\Psi}_{n,j}-\mathbb{E}\left[\tilde{\Psi}_{n,j}\mid\bR_n\right]\right)^{\prime}\mid\bR_n\right],
\end{align*}
and
\begin{align*}
    \frac{1}{n\rho_n}\Sigma_{n}
    =&\Var\left(\frac{1}{\sqrt{n\rho_n}}\sum_{i=1}^{n}R_{n,i}X_{n,i}\varepsilon_{n,i}\right)\\
    =&\frac{1}{n\rho_n}\sum_{i=1}^{n}\sum_{j\in\mathcal{N}_{n}(i,2K)}\mathbb{E}\left[\left(R_{n,i}\Psi_{n,i}-\rho_n\mathbb{E}\left[\Psi_{n,i}\right]\right)\left(R_{n,j}\Psi_{n,j}-\rho_n\mathbb{E}\left[\Psi_{n,j}\right]\right)^{\prime}\right].
\end{align*}

We consider the following feasible estimator:
\begin{align*}
    \frac{1}{N}\hat{\Sigma}_n
    =&\frac{1}{N}\sum_{i=1}^{n}\sum_{j\in\tilde{\mathcal{N}}_{n}(i,2K)}R_{n,i}R_{n,j}\hat{\Psi}_{n,i}\hat{\Psi}_{n,j}^{\prime}.
\end{align*}

To show the consistency of the variance estimator, we assume an additional sparsity condition.
The assumption requires a few more notations. 
Let $\delta_{n}(s;p)$ be the $p$-th sample moment of the set of $i$'s neighborhood within $s$-distance: $\delta_{n}(s;p)=n^{-1}\sum_{i=1}^n|\mathcal{N}_{n}(i;s)|^{p}$.
We also define $\mathcal{J}_n(s,m)$ as the set of quadruples $(i,j,i^{\prime},j^{\prime})$ such that $i^{\prime}$ and $j^{\prime}$ are $m$-neighbors of $i$ and $j$, respectively, and the distance between $i$ and $j$ is exactly $s$:
\begin{align*}
    \mathcal{J}_n(s,m)=\left\{(i,j,i^{\prime},j^{\prime})\in\mathcal{N}_n^4:i^{\prime}\in\mathcal{N}_n(i,m),j^{\prime}\in\mathcal{N}_n(j,m), d_n(i,j)=s\right\}.
\end{align*}
\begin{assumption}
    \label{asm:var_dependence}
    (i) $\delta_{n}(2K;2)=o(n)$. 
    (ii) $\sum_{s=0}^{2K}|\mathcal{J}_n(s,2K)|=o(n^2)$.
\end{assumption}
\cref{asm:var_dependence} is a version of Assumptions 7c and 7d of \cite{leung2022causal}. This assumption is satisfied if network links are not too dense.

\begin{theorem}
    \label{thm:var_consistency}
    Let $\tilde{\gamma}_n=\hat{\gamma}_n$.
    Under \cref{asm:linear,asm:sampling,asm:moment,asm:linear_propensity,asm:local,asm:sparsity,asm:dependence,asm:var_dependence}, we have
    \begin{align*}
        \frac{1}{N}\hat{\Sigma}_n
        &=\frac{1}{n\rho_n}\tilde{\Sigma}_n+\tilde{B}_{n}+o_{p^{R}}(1),
    \end{align*}
    where $U_n=o_{p^{R}}(1)$ means $U_n\overset{p^R}{\longrightarrow}0$, and
    \begin{align*}
        \tilde{B}_{n}
        =&\frac{1}{n\rho_n}\sum_{i=1}^{n}\sum_{j\in\tilde{\mathcal{N}}_{n}(i,2K)}R_{n,i}R_{n,j}\mathbb{E}\left[\tilde{\Psi}_{n,i}\mid\bR_n\right]\mathbb{E}\left[\tilde{\Psi}_{n,j}\mid\bR_n\right]^{\prime}.
    \end{align*}

     Let $\tilde{\gamma}_n=\gamma_n^{\causal}+o_p(1)$. If, in addition, we assume \cref{asm:exposure} and $\tilde{d}_n(i,j)=d_n(i,j)$ a.s. for all $(i,j)\in\mathcal{N}_n^2$ with $R_{n,i}=1$ and $R_{n,j}=1$ and for all $n\in\mathbb{N}$, then,

    \begin{align*}
        \frac{1}{N}\hat{\Sigma}_n
        &=\frac{1}{n\rho_n}\Sigma_n+\hat{B}_{n}+o_p(1),
    \end{align*}
    where 
    \begin{align*}
        \hat{B}_{n}
        =&\frac{1}{n}\sum_{i=1}^{n}\sum_{j\in\mathcal{N}_{n}(i,2K)}\rho_n\mathbb{E}\left[\Psi_{n,i}\right]\mathbb{E}\left[\Psi_{n,j}\right]^{\prime}.
    \end{align*}
\end{theorem}

An estimator satisfying $\tilde{\gamma}_n=\gamma_n^{\causal}+o_p(1)$ is given in \cref{app:gamma}. In general, $\hat{\gamma}_n\neq\gamma_n^{\causal}+o_p(1)$, and we need a modification on $\hat{\gamma}_n$. The condition $\tilde{d}_n(i,j)=d_n(i,j)$ a.s. for units with $R_{n,i}=1$ and $R_{n,j}=1$ is satisfied, for example, when the network is sampled using star sampling and the exposure mapping is restricted to the first neighborhood.

\cref{thm:var_consistency} implies that we can only estimate the variance up to the ones with bias terms $\tilde{B}_{n}$ and $\hat{B}_{n}$ since there is no hope to estimate each heterogeneous expectation consistently.
This bias is inevitable in heterogeneous treatment effect settings (\citealp{abadie2020sampling}; \citealp{leung2020treatment}; \citealp{gao2023causal}).
Combining this convergence and the asymptotic normality, we can estimate the variance of $\hat{\theta}_{n}$ by
\begin{equation}\label{eq:se}
    {\left(\tilde{Q}^{XX}_{n}\right)}^{-1}\left(\frac{1}{N}\hat{\Sigma}_n\right){\left(\tilde{Q}^{XX}_{n}\right)}^{-1}.
\end{equation}

The above variance estimator has a problem because we cannot guarantee conservativeness.
Indeed, bias matrices $\hat{B}_{n}$ and $\tilde{B}_{n}$ are not necessarily positive semi-definite.\footnote{Alternatively, we can implement the randomized inference as \cite{borusyak2023nonrandom}. For multidimensional $\hat{\theta}_n$, the randomized inference do not guarantee conservativeness, too.}
Conservative guarantee modification is possible.
We can write
$(1/N)\hat{\Sigma}_n=(1/N)\hat{R\Psi}^{\prime}_n\tilde{K}_n\hat{R\Psi}_n$,
where 
\begin{align*}
    \hat{R\Psi}_n&=\left(R_{n,1}\tilde{X}_{n,1}\hat{\varepsilon}_{n,1},\cdots,R_{n,n}\tilde{X}_{n,n}\hat{\varepsilon}_{n,n}\right)^{\prime},\\
    \tilde{K}_{n}&=[\mathds{1}\{\tilde{d}_n(i, j)\leq2K\}]_{i,j}.
\end{align*}
Eigendecomposition gives $\tilde{K}_n=\mathcal{Q}_n\Xi_n\mathcal{Q}_n^{\prime}$.
By replacing $\tilde{K}_n$ by $\tilde{K}_n^+=\mathcal{Q}_n\max\{0,\Xi_n\}\mathcal{Q}_n^{\prime}$ ($\max$ is taken element-wise), the variance matrix estimator 
\begin{align*}
    \frac{1}{N}\hat{\Sigma}_n^+
    =\frac{1}{N}\hat{R\Psi}^{\prime}_n\tilde{K}_n^{+}\hat{R\Psi}_n=\frac{1}{N}\sum_{i=1}^{n}\sum_{j=1}^{n}R_{n,i}R_{n,j}\hat{\Psi}_{n,i}\hat{\Psi}_{n,j}^{\prime}\tilde{K}_{n,i,j}^{+}.
\end{align*}
becomes positive semi-definite.
We also have $\tilde{K}_n^{-}=\mathcal{Q}_n|\min\{0,\Xi_n\}|\mathcal{Q}_n^{\prime}=\tilde{K}_n^{+}-\tilde{K}_n$.
This modification is provided by \cite{gao2023causal}. The modified variance estimator is given by
\begin{equation}\label{eq:se_mod}
    {\left(\tilde{Q}^{XX}_{n}\right)}^{-1}\left(\frac{1}{N}\hat{\Sigma}^{+}_n\right){\left(\tilde{Q}^{XX}_{n}\right)}^{-1}.
\end{equation}

Define $K_n=[\mathds{1}\{d_n(i, j)\leq2K\}]_{i,j}$, and define $K_n^{+}$ and $K_n^{-}$ in a similar manner to $\tilde{K}_n^{+}$ and $\tilde{K}_n^{-}$.
Define
\begin{align*}
    \tilde{\delta}_{n}^{-}(2K;p)=\frac{1}{n}\sum_{i=1}^{n}\left(\sum_{j=1}^{n}|\tilde{K}_{n,i,j}^{-}|\right)^{p},
    \qquad
    \delta_{n}^{-}(2K;p)=\frac{1}{n}\sum_{i=1}^{n}\left(\sum_{j=1}^{n}|K_{n,i,j}^{-}|\right)^{p},
\end{align*}
and
\begin{align*}
    &|\tilde{\mathcal{J}}_n^{-}(s,2K)|=\sum_{i=1}^{n}\sum_{j=1}^{n}\mathds{1}\{d_n(i, j)=s\}\left(\sum_{i^{\prime}=1}^{n}|\tilde{K}_{n,i,i^{\prime}}^{-}|\right)\left(\sum_{j^{\prime}=1}^{n}|\tilde{K}_{n,j,j^{\prime}}^{-}|\right),\\
    &|\mathcal{J}_n^{-}(s,2K)|=\sum_{i=1}^{n}\sum_{j=1}^{n}\mathds{1}\{d_n(i, j)=s\}\left(\sum_{i^{\prime}=1}^{n}|K_{n,i,i^{\prime}}^{-}|\right)\left(\sum_{j^{\prime}=1}^{n}|K_{n,j,j^{\prime}}^{-}|\right).
\end{align*}

\begin{assumption}
    \label{asm:var_dependence_eig}
    (i) $\tilde{\delta}_{n}^{-}(2K;1)=O_{\text{a.s.}}(1)$ and $\delta_{n}^{-}(2K;1)=O(1)$.
    (ii) $\tilde{\delta}_{n}^{-}(2K;2)=O_{\text{a.s.}}(n)$ and $\delta_{n}^{-}(2K;2)=O(n)$.
    (iii) $\sum_{s=0}^{2K}|\tilde{\mathcal{J}}_n^{-}(s,2K)|=O_{\text{a.s.}}(n^2)$ and $\sum_{s=0}^{2K}|\mathcal{J}_n^{-}(s,2K)|=O(n^2)$.
\end{assumption}
\cref{asm:var_dependence_eig} is a version of Assumptions 7b-7d of \cite{gao2023causal}. The assumption is a modified version of \cref{asm:var_dependence} for the eigenvalue modification.

\begin{theorem}
    \label{thm:var_consistency_eig}
    Let $\tilde{\gamma}_n=\hat{\gamma}_n$.
    Under \cref{asm:linear,asm:sampling,asm:moment,asm:linear_propensity,asm:local,asm:sparsity,asm:dependence,asm:var_dependence_eig}, we have
    \begin{align*}
        \frac{1}{N}\hat{\Sigma}_n^{+}
        &=\frac{1}{n\rho_n}\tilde{\Sigma}_n+\tilde{B}_{n}^{+}+o_p^R(1),
    \end{align*}
    where 
    \begin{align*}
        \tilde{B}_{n}^{+}
        =&\frac{1}{n\rho_n}\sum_{i=1}^{n}\sum_{j=1}^{n}R_{n,i}R_{n,j}\mathbb{E}\left[\tilde{\Psi}_{n,i}\mid\bR_n\right]\mathbb{E}\left[\tilde{\Psi}_{n,j}\mid\bR_n\right]^{\prime}\tilde{K}_{n,i,j}^{+}\\
        &+\frac{1}{n\rho_n}\sum_{i=1}^{n}\sum_{j=1}^{n}R_{n,i}R_{n,j}\mathbb{E}\left[\left(\tilde{\Psi}_{n,i}-\mathbb{E}\left[\tilde{\Psi}_{n,i}\mid\bR_n\right]\right)\left(\tilde{\Psi}_{n,j}-\mathbb{E}\left[\tilde{\Psi}_{n,j}\mid\bR_n\right]\right)^{\prime}\mid\bR_n\right]\tilde{K}_{n,i,j}^{-}
    \end{align*}

    Let $\tilde{\gamma}_n=\gamma_n^{\causal}+o_p(1)$.
    If, in addition, we assume \cref{asm:exposure} and $\tilde{d}_n(i,j)=d_n(i,j)$ a.s. for all $(i,j)\in\mathcal{N}_n^2$ with $R_{n,i}=1$ and $R_{n,j}=1$ and for all $n\in\mathbb{N}$, then,
    \begin{align*}
        \frac{1}{N}\hat{\Sigma}_n^{+}
        &=\frac{1}{n\rho_n}\Sigma_n+\hat{B}_{n}^{+}+o_p(1),
    \end{align*}
    where 
    \begin{align*}
        \hat{B}_{n}^{+}
        =&\frac{1}{n}\sum_{i=1}^{n}\sum_{j=1}^{n}\rho_n\mathbb{E}\left[\Psi_{n,i}\right]\mathbb{E}\left[\Psi_{n,j}\right]^{\prime}K_{n,i,j}^{+}\\
        &+\frac{1}{n\rho_n}\sum_{i=1}^{n}\sum_{j=1}^{n}\mathbb{E}\left[\left(R_{n,i}\Psi_{n,i}-\rho_n\mathbb{E}\left[\Psi_{n,i}\right]\right)\left(R_{n,j}\Psi_{n,j}-\rho_n\mathbb{E}\left[\Psi_{n,j}\right]\right)^{\prime}\right]K_{n,i,j}^{-}.
    \end{align*}
    
\end{theorem}

\section{Simulation}\label{sec:simulation}
In this section, we conduct a simulation exercise to illustrate the potential severity of contamination bias. We focus on a case where $T_{n,i} \neq \tilde{T}_{n,i}$ and contamination bias can arise. See \cref{app:sim_additional} for results when $T_{n,i} = \tilde{T}_{n,i}$, where no contamination bias occurs and our inference procedure is valid under correct model specification.

In the following exercise, we use network link information from \cite{banerjee2013diffusion} to simulate variables based on a real-world network structure, rather than on an artificially generated population network.
That study conducted a network survey among randomly selected respondents across 75 villages in rural southern India, where respondents were asked to name 5 to 8 contacts across 12 interaction dimensions (e.g., house visits, borrowing goods). We focus on the borrowing network among individuals, specifically whether a person borrows rice or kerosene from others.\footnote{\cite{banerjee2013diffusion} also collected household-level network data; we use individual-level network data, which is sparser than the household-level networks.} To illustrate the applicability of our framework to a single large network without relying on many clusters, we focus on the largest village and use its borrowing network as the population network $\bA_{n}$. Basic network statistics for this village are presented in \cref{tab:network_info}:

\begin{table}[htbp]
    \caption{Network Information}
    \label{tab:network_info}
    \centering
    \begin{tabular}{cccc}
      \hline
      \hline
      \textbf{Nodes} & \textbf{Edges} & \textbf{Mean Degree} & \textbf{Mean 2nd Order Degree} \\\hline
        1770 &  5556 & 6.28 & 11.44 \\\hline
    \end{tabular}
    \caption*{\footnotesize \textit{Notes}: \textbf{Nodes} reports the number of individuals in the village; \textbf{Edges} reports the number of links based on borrowing relationships; \textbf{Mean Degree} reports the mean degree; \textbf{Mean 2nd Order Degree} reports the mean count of friends-of-friends not directly connected to node $i$.}
\end{table}

In this exercise, we consider a scenario in which the true and observed exposure mappings differ. The main objective is to quantify the severity of contamination bias. Specifically, we focus on a case where there is no contamination bias at the population level, but bias can arise due to the choice of $\tilde{g}$. We specify the exposure mapping as in \cref{ex:cai}:
\begin{align*}
    T_{n,i}&=\left(R_{n,i}D^*_{n,i},\frac{\sum_{j\neq i}A_{n,i,j}R_{n,j}D^*_{n,j}}{\sum_{j\neq i}A_{n,i,j}},\frac{\sum_{j\neq i}\sum_{k\neq i,j}A_{n,i,j}A_{n,j,k}(1-A_{n,i,k})R_{n,k}D^*_{n,k}}{\sum_{j\neq i}\sum_{k\neq i,j}A_{n,i,j}A_{n,j,k}(1-A_{n,i,k})}\right)\\
    &=:(D_{n,i},\text{net}_{n,i},\text{weak}_{n,i}),
\end{align*}
and $\tilde{T}_{n,i}$ is the same as $T_{n,i}$ except that its second and third elements are replaced by
\begin{align*}
    \tilde{\text{net}}_{n,i}=\frac{\sum_{j\neq i}A_{n,i,j}R_{n,j}D_{n,j}^{*}}{\sum_{j\neq i}R_{n,j}A_{n,i,j}};\quad \tilde{\text{weak}}_{n,i}=\frac{\sum_{j\neq i}\sum_{k\neq i,j}R_{n,j}A_{n,i,j}A_{n,j,k}(1-A_{n,i,k})R_{n,k}D^*_{n,k}}{\sum_{j\neq i}\sum_{k\neq i,j}R_{n,j}A_{n,i,j}R_{n,k}A_{n,j,k}(1-A_{n,i,k})}.
\end{align*}
For comparison, we also consider $\tilde{T}_{n,i}^{\text{overlap}}$, which is the same as $\tilde{T}_{n,i}$ except that each $1-A_{n,i,k}$ in $\tilde{\text{weak}}_{n,i}$ is replaced by $1$. As discussed in \cref{ex:cai}, due to overlaps in the second and third elements, the sample-level causal estimand based on $\tilde{T}_{n,i}^{\text{overlap}}$ will be contaminated. In contrast, the estimands based on $T_{n,i}$ and $\tilde{T}_{n,i}$ are not, as they are free of such overlaps and correlations.

We implement the following simulation design. First, we set individual-specific parameters as
$\theta_{n,i,(1)}\sim\text{Exponential}(1/3)$ i.i.d., $\theta_{n,i,(2)}=M_{n,i}$, $\theta_{n,i,(3)}=0$, and $\nu_{n,i}\sim N(0,2)$ i.i.d., where $M_{n,i}$ is a clustering coefficient given by $M_{n,i}=(100/n)\times \sum_{k\neq i}\left(\sum_{j\neq i,k}A_{n,i,j}A_{n,j,k}\right)^{2}$.
Specifically, we draw these $\theta_{n,i}$ and $\nu_{n,i}$ once and treat them as fixed for each Monte Carlo iteration to simulate design-based and sampling-based uncertainties.
We choose $\theta_{n,i,(2)}=M_{n,i}$ to mechanically maximize contamination bias, as $M_{n,i}$ correlates with the contamination weights appearing in \cref{cor:causal_element}. The average spillover effect from $\text{net}_{n,i}$ (i.e., the average of $M_{n,i}$) is about $1/2$. We also set $\theta_{n,i,(3)}=0$ for all $i$, so any deviation from $0$ can be interpreted as contamination bias.
Given the fixed population adjacency matrix $\bA_{n}$ from \cite{banerjee2013diffusion}, we can calculate the population-based causal estimand $\theta_{n}^{\causal}$.

Next, for each iteration, we draw $D^*_{n,i}\sim \text{Bernoulli}(0.5)$ i.i.d., and $R_{n,i}\sim \text{Bernoulli}(\rho_{n})$ i.i.d. for varying sampling probabilities $\rho_{n}\in\{0.1,0.5,1.0\}$ to see the impact of sampling uncertainty on inference. For each realization of $\bR_{n}$, we compute $\theta_{n}^{\causample}$. Subsequently, using each realization of $\bR_{n}$ and $\bD_{n}$, we estimate $\hat{\theta}_{n}$ from the regression $Y_{n,i}\sim \tilde{X}_{n,i} + \tilde{Z}_{n,i}$,
where $\tilde{Z}_{n,i}=(R_{n,i}p_{n},p_{n}\mathds{1}\{\sum_{j\neq i}R_{n,j}A_{n,i,j}>0\},p_{n}\mathds{1}\{\sum_{j\neq i}\sum_{k\neq i,j}R_{n,j}A_{n,i,j}R_{n,k}A_{n,j,k}(1-A_{n,i,k})>0\})$, restricted to units with $R_{n,i}=1$. Finally, we compute the standard errors based on \cref{eq:se_mod} with $\tilde{\gamma}_{n}=\hat{\gamma}_{n}$ for $\theta^{\causample}$ and with $\tilde{\gamma}_{n}$ from \cref{app:gamma} for $\theta^{\causal}$, as well as the conventional Eicker-Huber-White (EHW) standard errors, which are computed from the following variance estimator:
\begin{align*}
    \left(\tilde{Q}_{n}^{XX}\right)^{-1}\left(\frac{1}{N}\sum_{i=1}^{n}R_{n,i}\tilde{X}_{n,i}\tilde{X}_{n,i}'\hat{\varepsilon}_{n,i}^{2}\right)\left(\tilde{Q}_{n}^{XX}\right)^{-1}.
\end{align*}
When computing the standard errors based on \cref{eq:se_mod}, we use the observed network $\tilde{\bA}_{n}=[R_{n,i}\times R_{n,j}\times A_{n,i,j}]_{i,j}$, which is the sampled network with induced subgraph links. We repeat this process 2,000 times. The overlapping case is implemented in the same manner, except that we use $\tilde{T}_{n,i}^{\text{overlap}}$ instead of $\tilde{T}_{n,i}$, and the third element of $\tilde{Z}_{n,i}$ is replaced by $p_{n}\mathds{1}\{\sum_{j\neq i}\sum_{k\neq i,j}R_{n,j}A_{n,i,j}R_{n,k}A_{n,j,k}>0\}$.

Simulation results for $\rho_{n}\in\{0.1,0.5,1.0\}$ are summarized in \cref{tab:simulationTcontam}. In Panel A, we use $\tilde{T}_{n,i}$ whose $\tilde{\text{weak}}_{n,i}$ does not have an overlap in $D^*_{n,j}$ for any $j$ with $\tilde{\text{net}}_{n,i}$. In Panel B, we use $\tilde{T}^{\text{overlap}}_{n,i}$ whose $\tilde{\text{weak}}^{\text{overlap}}_{n,i}$ does share some $D^*_{n,j}$ with $\tilde{\text{net}}_{n,i}$. Also note that, in both panels, the true exposure mapping is fixed to $T_{n,i}$ defined above. Hence, the population-level causal estimands $\theta_{n}^{\causal}$ are the same regardless of which $\tilde{T}_{n,i}$ or $\tilde{T}_{n,i}^{\text{overlap}}$ is used.

From Panel A of \cref{tab:simulationTcontam} (no overlap case), we can observe that the sample-level estimand and estimator largely deviate from the population-level estimand for $\text{net}_{n,i}$. This deviation is driven not by contamination, but by the difference between $\text{net}_{n,i}$ and $\tilde{\text{net}}_{n,i}$:
\begin{align*}
    \text{net}_{n,i}=\frac{\sum_{j\neq i}A_{n,i,j}R_{n,j}D^*_{n,j}}{\sum_{j\neq i}A_{n,i,j}}\neq\frac{\sum_{j\neq i}A_{n,i,j}R_{n,j}D^*_{n,j}}{\sum_{j\neq i}R_{n,j}A_{n,i,j}}=\tilde{\text{net}}_{n,i}.
\end{align*}
When $\rho_{n}$ is small, the denominator of $\tilde{\text{weak}}_{n,i}$ tends to be smaller than that of $\text{net}_{n,i}$, which results in a downward bias. 

Because of the bias, the coverage probabilities against $\theta_{n}^{\causal}$ are close to $0$ with both EHW standard errors and those based on \cref{eq:se_mod}, especially when $\rho_{n}$ is small. 
However, as $\rho_{n}$ increases, the bias and coverage probabilities tend to improve with our proposed standard errors \cref{eq:se_mod} because the difference between $T_{n,i}$ and $\tilde{T}_{n,i}$ becomes smaller and the standard errors are designed to be conservative. In contrast, the EHW standard errors fail to capture the dependence structure and thus severely under-cover the causal estimands as $\rho_{n}$ increases.

From Panel B of \cref{tab:simulationTcontam} (with overlap case), we can observe a similar pattern as in Panel A when $\rho_{n}$ is small. However, a crucial difference arises when $\rho_{n}=1.0$. We can observe that $\theta_{n,(3)}^{\causample}$ and $\hat{\theta}_{n,(3)}$ are largely biased downward compared with $\theta_{n,(3)}^{\causal}$, with a magnitude similar to that of $\theta_{n,(2)}^{\causal}$. Since the true $\theta_{n,i,(3)}=0$ for all $i$, this bias is mainly driven by contamination, as suggested by \cref{cor:causal_element}. The contamination bias is also reflected in the average absolute deviation of the estimator and the coverage probabilities against $\theta_{n,(3)}^{\causal}$ for both EHW standard errors and those based on \cref{eq:se_mod}, resulting in under-coverage. 

In summary, the simulation results in \cref{tab:simulationTcontam} show that the deviation of $\tilde{T}_{n,i}$ from $T_{n,i}$ can lead to severe bias and under coverage for the population causal estimands. The results also highlight the potential severity of contamination bias when there is a small overlap in elements of $\tilde{T}_{n,i}$, whose size can be comparable to the true spillover effects. This emphasizes the importance of choice of $\tilde{g}$ in practice and calls for caution when interpreting the results based on the linear regression framework. In the next section, we discuss whether the contamination bias is present in the real data application.

\begin{table}[htbp]
  \caption{Simulation Results: $T_{n,i}\neq\tilde{T}_{n,i}$ case}
  \label{tab:simulationTcontam}
  \scalebox{0.9}{
  \begin{threeparttable}
    \centering
    \begingroup
  \begin{tabular}{lccccccccc}
    \hline
    \hline
    \multicolumn{10}{c}{\textbf{Panel A: No Overlaps}} \\
    \hline
    & \multicolumn{3}{c}{\(\rho = 0.1\)} & \multicolumn{3}{c}{\(\rho = 0.5\)} & \multicolumn{3}{c}{\(\rho = 1.0\)} \\
    & D & net & weak & D & net & weak & D & net & weak \\
    \hline
    $\theta^{\causal}$ & 0.348 & 0.567 & 0.0 & 0.348 & 0.567 & 0.0 & 0.348 & 0.567 & 0.0 \\
    $\theta^{\causample}$ & 0.347 & 0.153 & 0.0 & 0.348 & 0.282 & 0.0 & 0.348 & 0.567 & 0.0 \\
    $\hat{\theta}$ & 0.347 & 0.139 & 0.009 & 0.346 & 0.28 & -0.004 & 0.347 & 0.565 & -0.01 \\
    $\text{SE EHW}$ & 0.163 & 0.251 & 0.475 & 0.087 & 0.113 & 0.128 & 0.068 & 0.11 & 0.111 \\
    $\text{SE \eqref{eq:se_mod} } \theta^{\causal}$ & 0.165 & 0.263 & 0.549 & 0.102 & 0.135 & 0.175 & 0.108 & 0.191 & 0.233 \\
    $\text{SE \eqref{eq:se_mod} } \theta^{\causample}$ & 0.163 & 0.248 & 0.398 & 0.1 & 0.133 & 0.153 & 0.104 & 0.198 & 0.173 \\
    $|\hat{\theta}-\theta^{\causal}|$ & 0.182 & 0.47 & 0.696 & 0.08 & 0.292 & 0.159 & 0.058 & 0.141 & 0.153 \\
    $|\hat{\theta}-\theta^{\causample}|$ & 0.18 & 0.289 & 0.696 & 0.08 & 0.124 & 0.159 & 0.058 & 0.141 & 0.153 \\
    $\text{Coverage EHW }\theta^{\causal}$ & 0.844 & 0.56 & 0.703 & 0.908 & 0.335 & 0.797 & 0.938 & 0.775 & 0.74 \\
    $\text{Coverage EHW }\theta^{\causample}$ & 0.846 & 0.819 & 0.703 & 0.909 & 0.836 & 0.797 & 0.938 & 0.775 & 0.74 \\
    $\text{Coverage \eqref{eq:se_mod} }\theta^{\causal}$ & 0.847 & 0.577 & 0.768 & 0.942 & 0.443 & 0.922 & 0.997 & 0.968 & 0.964 \\
    $\text{Coverage \eqref{eq:se_mod} }\theta^{\causample}$ & 0.844 & 0.813 & 0.618 & 0.939 & 0.898 & 0.87 & 0.995 & 0.971 & 0.915 \\
    \hline
    \multicolumn{10}{c}{\textbf{Panel B: With Overlaps}} \\
    \hline
    & \multicolumn{3}{c}{\(\rho = 0.1\)} & \multicolumn{3}{c}{\(\rho = 0.5\)} & \multicolumn{3}{c}{\(\rho = 1.0\)} \\
    & D & net & weak & D & net & weak & D & net & weak \\
    \hline
    $\theta^{\causal}$ & 0.348 & 0.567 & 0.0 & 0.348 & 0.567 & 0.0 & 0.348 & 0.567 & 0.0 \\
    $\theta^{\causample}$ & 0.347 & 0.149 & 0.032 & 0.348 & 0.279 & 0.008 & 0.348 & 0.773 & -0.356 \\
    $\hat{\theta}$ & 0.347 & 0.135 & 0.037 & 0.346 & 0.28 & 0.0 & 0.347 & 0.783 & -0.374 \\
    $\text{SE EHW}$ & 0.163 & 0.269 & 0.447 & 0.087 & 0.155 & 0.176 & 0.068 & 0.249 & 0.264 \\
    $\text{SE \eqref{eq:se_mod} } \theta^{\causal}$ & 0.165 & 0.279 & 0.492 & 0.102 & 0.186 & 0.22 & 0.105 & 0.419 & 0.452 \\
    $\text{SE \eqref{eq:se_mod} } \theta^{\causample}$ & 0.163 & 0.265 & 0.416 & 0.1 & 0.18 & 0.204 & 0.104 & 0.394 & 0.395 \\
    $|\hat{\theta}-\theta^{\causal}|$ & 0.181 & 0.476 & 0.59 & 0.08 & 0.297 & 0.204 & 0.058 & 0.279 & 0.439 \\
    $|\hat{\theta}-\theta^{\causample}|$ & 0.179 & 0.295 & 0.589 & 0.08 & 0.147 & 0.204 & 0.058 & 0.211 & 0.295 \\
    $\text{Coverage EHW }\theta^{\causal}$ & 0.845 & 0.584 & 0.756 & 0.908 & 0.528 & 0.828 & 0.936 & 0.854 & 0.65 \\
    $\text{Coverage EHW }\theta^{\causample}$ & 0.84 & 0.845 & 0.752 & 0.91 & 0.896 & 0.828 & 0.936 & 0.928 & 0.834 \\
    $\text{Coverage \eqref{eq:se_mod} }\theta^{\causal}$ & 0.848 & 0.597 & 0.8 & 0.938 & 0.653 & 0.918 & 0.996 & 0.987 & 0.902 \\
    $\text{Coverage \eqref{eq:se_mod} }\theta^{\causample}$ & 0.84 & 0.837 & 0.714 & 0.936 & 0.933 & 0.886 & 0.995 & 0.995 & 0.96 \\
    \hline
  \end{tabular}
        \begin{tablenotes}
            \footnotesize
            \item {\it Note:} Panel A reports the results when $\tilde{T}_{n,i}$ is used while Panel B reports the results when $\tilde{T}_{n,i}^{\text{overlap}}$ is used. The first three rows report the averages of the population and sample-level causal estimands and the OLS estimator. The fourth and fifth rows report the averages of the EHW standard errors and our proposed standard errors based on \cref{eq:se_mod}. The sixth and seventh rows report the average absolute deviations of the estimator from the two causal estimands. The last four rows report the coverage probabilities of the $95\%$ confidence intervals constructed using the EHW standard errors and the standard errors based on our proposed method \cref{eq:se_mod} for the two causal estimands.
        \end{tablenotes}
    \endgroup
    \end{threeparttable}  
    }
\end{table}

\section{Empirical Illustration}
\label{sec:empirical}
In an influential study, \cite{cai2015social} conducted a large-scale network experiment in which they randomly assigned information sessions on weather insurance products to rice farmers in rural villages in China. Out of 185 randomly selected villages, all rice farmers were invited to participate, and approximately 90\% agreed to attend. The researchers administered both a household survey (to gather farmer characteristics) and a network survey (to collect friendship links). In the network survey, household heads were asked to list their five closest friends with whom they discussed rice production and financial matters, which provides a star sampling network. They were allowed to list friends outside of their village.\footnote{\cite{cai2015social} conducted a pilot network survey in two villages without limiting the number of friends, but found that most farmers listed five or fewer friends. We take this analysis at face value and assume that there is no concern about censoring the number of friends.}

The information sessions were conducted in two rounds (first and second) and with varying intensity (simple or intensive). Farmers were randomly assigned to one of four possible sessions. The main outcome here, $Y_{n,i}$, is a test score measuring understanding of the insurance product, taking 10 values between 0 and 1 (\textbf{test}). The treatment variable, $D_{n,i}$, indicates whether a farmer was assigned to an intensive session (\textbf{intensive}). To measure the spillover/diffusion effects of the information sessions on farmers' knowledge, the researchers focused on a subsample of farmers who were not invited in the first round and defined (i) the fraction of a farmer's friends who attended an intensive session in the first round (\textbf{net}) and (ii) the fraction of those friends' friends who attended an intensive session in the first round (\textbf{weak}). 

As discussed in \cref{ex:cai} and the simulation section, including first-order overlaps between \textbf{net} and \textbf{weak} can significantly affect inference through induced contamination bias.\footnote{We found that \cite{cai2015social} included such overlaps in their version of \textbf{weak}; see the data/do/rawnet.do file in their replication folder: \url{https://www.openicpsr.org/openicpsr/project/113593/version/V1/view;jsessionid=743ABAC8AEBB3E612D4250D02BE40429}.} Here, we empirically examine whether such overlaps make a significant difference by comparing results when these overlaps are included or excluded in \textbf{net} and \textbf{weak}. Specifically, we run the following regression for the overlap and no-overlap specifications:\footnote{Note that \cite{cai2015social} specified the exposure mapping as either $(\text{intensive},\text{net})$ or $(\text{weak})$, running regressions separately. Here, we consider a hypothetical scenario where both \textbf{net} and \textbf{weak} are included in the regression simultaneously, rather than replicating their original results.}
$$
    \text{test} \sim \text{intensive} + \text{net} + \text{weak} + \text{controls}.
$$

For estimation, unlike in the simulation exercise above, we use all the available villages in the sample, as done in \cite{cai2015social}. We control for household characteristics, village fixed effects, and network information (degree dummy) to satisfy \cref{asm:linear_propensity}. Standard errors are calculated via our proposed method \cref{eq:se_mod}, with \(K=2\).

\begin{table}[htbp]
\centering
\caption{Regression Results for \cite{cai2015social}'s data}
\begin{tabular}{lccc}
\hline
\hline
 & With Overlaps & No Overlaps \\
\hline
intensive & 0.0752  & 0.0734  \\
  & (0.0159)  & (0.0164)  \\
net & 0.3110  & 0.2879  \\
  & (0.0527)  & (0.0500)  \\
weak & -0.1511  & -0.0741  \\
  & (0.0453)  & (0.0383)  \\
\hline
\end{tabular}
\caption*{\footnotesize \textit{Notes}: The number of villages is $47$, and the total sample size is $1247$. The first and second columns report estimates with and without overlaps in first-order links between \textbf{net} and \textbf{weak}. All regressions include household characteristics, village fixed effects, and network information as controls. Standard errors, computed using our proposed method \cref{eq:se_mod} with $\tilde{\gamma}_{n}=\hat{\gamma}_{n}$, are reported in parentheses.}
\label{table:empirical_regression} 
\end{table}

\cref{table:empirical_regression} reports the OLS estimator $\hat{\theta}_{n}$ and its standard errors, both with and without overlaps in the exposure mappings. When overlaps are included, the coefficient for \textbf{net} remains largely unchanged, but the estimate for \textbf{weak} becomes substantially more negative. 
Specifically, the coefficient on \textbf{weak} is statistically significant at the 95\% confidence level under the overlap specification, and its magnitude nearly doubles compared to the no-overlap specification\textemdash becoming comparable in size (but opposite in sign) to that of \textbf{net}.
This highlights the risk of overstating the effect of weak connections due to contamination bias, even when the true effect may be small or absent.

This pattern in the empirical results is consistent with the simulation findings in \cref{tab:simulationTcontam}, where overlaps in the exposure mapping lead to substantial contamination bias in the estimates of \textbf{weak}, while the estimates of \textbf{net} remain largely unaffected. 
Overall, this exercise highlights that correlations among elements of the exposure mapping can potentially lead to misleading assessments of causal spillover effects.

\section{Conclusion}
\label{sec:conclusion}
In this paper, we study a linear regression framework for estimating causal spillover effects in network experiments. We show that, due to contamination bias, the OLS estimator for spillover effects does not bear a causal interpretation unless the exposure mapping is free of correlation among its elements. We also develop a novel asymptotic theory for inference on causal spillover effects, allowing for explicit sampling of units and networks, as well as network dependence.

Based on our theoretical analysis and simulation/empirical exercises, we recommend that researchers follow the flowchart in \cref{fig:flowchart} when estimating causal spillover effects in network experiments using linear regression. A crucial step is to ensure that the exposure mapping is free of correlations among its elements to avoid contamination bias and to ensure a causal interpretation of the OLS estimator. If the exposure mapping implied by plausible economic theories is not free of correlations but is sufficiently discrete (e.g., binary) to satisfy the overlap condition, we suggest avoiding the OLS estimator and instead using alternative methods, such as inverse probability weighting (e.g., \citealp{aronow2017estimating}; \citealp{leung2022causal}; \citealp{gao2023causal}), to directly estimate the causal treatment effects. 

\begin{figure}[ht]
\begin{spacing}{1}
    \caption{Flowchart for Valid Inference with Linear Regression}
    \label{fig:flowchart}
    \centering
    \scalebox{0.9}{
\begin{tikzpicture}[node distance=1.2cm and 2.7cm, scale=0.95, transform shape, every node/.style={font=\small}]

  \node (start)  [startstop]               {Start};
  \node (dec1)   [decision, below=of start] {Covariates satisfy \cref{asm:linear_propensity}?};
  \node (dec2)   [decision, below=of dec1]  {Regressors satisfy the no correlation condition in \cref{cor:no_contamination}?};
  \node (dec3)   [decision, below=of dec2]  {Exposure mapping correctly specified and relevant network information observed (\cref{asm:exposure}(i))?};
  \node (dec4)   [decision, below=of dec3]  {Network HAC estimator correctly applied?};
  \node (stop)   [startstop, below=of dec4] {Valid inference: confidence interval >95\% asymptotically};

  \node (proc1)  [process, right=of dec1]  {Include\\ required covariates};
  \node (proc2)  [process, right=of dec2]  {Modify the exposure mapping or flag results as potentially contaminated};
  \node (proc3)  [process, right=of dec3]  {Interpret as inference for $\theta^{\causample}$, not $\theta^{\causal}$};
  \node (proc4)  [process, right=of dec4]  {Apply correct network HAC};

  \coordinate (mid12) at ($(dec1.south)!0.5!(dec2.north)$);
  \coordinate (mid23) at ($(dec2.south)!0.5!(dec3.north)$);
  \coordinate (mid34) at ($(dec3.south)!0.5!(dec4.north)$);
  \coordinate (mid45) at ($(dec4.south)!0.5!(stop.north)$);

  \draw [arrow] (start) -- (dec1);
  \draw [arrow] (dec1) -- node[left]{Yes} (dec2);
  \draw [arrow] (dec2) -- node[left]{Yes} (dec3);
  \draw [arrow] (dec3) -- node[left]{Yes} (dec4);
  \draw [arrow] (dec4) -- node[left]{Yes} (stop);

  \draw [arrow] (dec1.east) -- node[above]{No} (proc1.west);
  \draw [arrow] (proc1.south) |- (mid12);

  \draw [arrow] (dec2.east) -- node[above]{No} (proc2.west);
  \draw [arrow] (proc2.south) |- (mid23);

  \draw [arrow] (dec3.east) -- node[above]{No} (proc3.west);
  \draw [arrow] (proc3.south) |- (mid34);

  \draw [arrow] (dec4.east) -- node[above]{No} (proc4.west);
  \draw [arrow] (proc4.south) |- (mid45);

\end{tikzpicture}
}
\end{spacing}
\end{figure}

While this paper establishes a comprehensive framework for network experiments on sampled networks, several avenues for future research emerge.
First, relaxing the sampling assumptions to accommodate cluster and multi-wave designs, as well as allowing more complex assignment mechanisms, would broaden applicability. The present analysis permits assignment conditional on observed covariates but excludes matched-pair and blocked randomization.
Second, a systematic comparison between regression-based estimators and inverse-probability-weighting approaches for spillover effects in network experiments is important, but lies beyond the scope of this paper.

\addcontentsline{toc}{section}{References}
\putbib[list_ref]
\end{bibunit}

\newpage
\appendix

\clearpage
\begin{bibunit}[apecon]

\begin{LARGE}
\begin{center}
Supplement to ``Design-based and Network Sampling-Based Uncertainties in Network Experiments''
\end{center}
\end{LARGE}

\begin{large}
\begin{center}
Kensuke Sakamoto and Yuya Shimizu\\
University of Wisconsin-Madison
\end{center}
\end{large}

\begin{large}
\begin{center}
\today
\end{center}
\end{large}

\bigskip

This supplementary appendix contains proofs of the results in the main text as well as auxiliary results.
\cref{app:gamma} discusses how to estimate the nuisance parameters consistently.
\cref{app:lemmas} contains technical lemmas. \cref{app:proof} contains proofs. 
\cref{app:sim_additional} presents additional simulation results.
\cref{app:survey} lists the papers included in the survey of network experiment research presented in the Introduction.

\section{Example for $\tilde{\gamma}_n=\gamma_n^{\causal}+o_p(1)$}
\label{app:gamma}
In \cref{thm:var_consistency,thm:var_consistency_eig}, we need some $\tilde{\gamma}_n$ satisfying $\tilde{\gamma}_n=\gamma_n^{\causal}+o_p(1)$.
By \cref{asm:exposure} (ii), we, without loss of generality, assume that the first $m$ elements of $Z_{n,i}$ depend on $R_{n,i}$ multiplicatively.\footnote{In the usual applications, it is enough to consider the $m=1$ case.}
We allow general heterogeneous treatment assignment in \cref{asm:sampling} (iii).
We also assume that the researcher knows $\rho_n$ or the population size $n$.
Let $Z_{n,i}=(Z_{(1:m),n,i}',Z_{-(1:m),n,i}')'$, where $Z_{(1:m),n,i}$ is the first $m$ elements of $Z_{n,i}$ and $Z_{-(1:m),n,i}$ are the remaining elements.
Recall that $\tilde{Z}_{n,i}=Z_{n,i}$ under \cref{asm:exposure}.

Define 
\begin{equation}
    \tilde{\gamma}_{n}=(\tilde{P}_n^{ZZ})^{-1}\tilde{P}_n^{ZY},\label{eq:gamma}
\end{equation}
where
\begin{align*}
    \tilde{P}_n^{ZZ}&=\frac{1}{N} \sum_{i=1}^n R_{n, i}\left(\begin{array}{cc}
    \rho_nZ_{(1:m),n,i}Z_{(1:m),n,i}'&\rho_nZ_{(1:m),n,i}Z_{-(1:m),n,i}'\\
    \rho_nZ_{-(1:m),n,i}Z_{(1:m),n,i}'&Z_{-(1:m),n,i}Z_{-(1:m),n,i}'
    \end{array}\right),\\
    \tilde{P}_n^{ZY}&=\frac{1}{N} \sum_{i=1}^n R_{n, i}\left(\begin{array}{c}
    \rho_nZ_{(1:m),n,i}\\
    Z_{-(1:m),n,i}
    \end{array}\right)Y_{n,i}.
\end{align*}
Note that some elements of $\tilde{P}_n^{ZZ}$ and $\tilde{P}_n^{ZY}$ are rescaled by $\rho_n$ from $\tilde{Q}_n^{ZZ}$ and $\tilde{Q}_n^{ZY}$. $\rho_n$ can be replaced with its consistent estimator $N/n$. The consistency of $\tilde{\gamma}_{n}$ is shown in \cref{lem:gamma_consistency}.

\section{Preliminary Results}
\label{app:lemmas}
Remember that for each $i\in\mathcal{N}_{n}$,
\begin{align*}
    T_{n,i}&=g(i,\bD_{n},\bA_{n});\\
    \tilde{T}_{n,i}&=\tilde{g}(i,\bD_{n},\tilde{\bA}_{n});\\
    X_{n,i}&=T_{n,i}-\Lambda_{n}Z_{n,i};\\
    \tilde{X}_{n,i}&=\tilde{T}_{n,i}-\tilde{\Lambda}_{n}\tilde{Z}_{n,i},
\end{align*}
where \begin{align*}
    \Lambda_{n}&=(\sum_{i=1}^n \mathbb{E}[T_{n,i}Z_{n,i}^{\prime}])(\sum_{i=1}^n \mathbb{E}[Z_{n,i}Z_{n,i}^{\prime}])^{-1};\\
    \tilde{\Lambda}_{n}&=(\sum_{i=1}^n R_{n, i}\mathbb{E}[\tilde{T}_{n,i}|\bR_n]\tilde{Z}_{n,i}^{\prime})(\sum_{i=1}^n R_{n, i}\tilde{Z}_{n,i}\tilde{Z}_{n,i}^{\prime})^{-1},
\end{align*}
and
$$
\Omega_n=\frac{1}{n} \sum_{i=1}^n\mathbb{E}\left[\left(\begin{array}{c}
    Y_{n,i}\\
    X_{n,i}\\
    Z_{n,i}
\end{array}\right)\left(\begin{array}{c}
    Y_{n,i}\\
    X_{n,i}\\
    Z_{n,i}
\end{array}\right)^{\prime}\right]\equiv\left(\begin{array}{ccc}
    \Omega_n^{YY} & \Omega_n^{YX} & \Omega_n^{YZ}\\
    \Omega_n^{XY} & \Omega_n^{XX} & \Omega_n^{XZ}\\
    \Omega_n^{ZY} & \Omega_n^{ZX} & \Omega_n^{ZZ}
\end{array}\right);
$$
$$
\tilde{Q}_n=\frac{1}{N} \sum_{i=1}^n R_{n, i}\left(\begin{array}{c}
    Y_{n,i}\\
    \tilde{X}_{n,i}\\
    \tilde{Z}_{n,i}
\end{array}\right)\left(\begin{array}{c}
    Y_{n,i}\\
    \tilde{X}_{n,i}\\
    \tilde{Z}_{n,i}
\end{array}\right)^{\prime}\equiv\left(\begin{array}{ccc}
    \tilde{Q}_n^{YY} & \tilde{Q}_n^{YX} & \tilde{Q}_n^{YZ}\\
    \tilde{Q}_n^{XY} & \tilde{Q}_n^{XX} & \tilde{Q}_n^{XZ}\\
    \tilde{Q}_n^{ZY} & \tilde{Q}_n^{ZX} & \tilde{Q}_n^{ZZ}
\end{array}\right);
$$
$$
\tilde{\Omega}_n=\frac{1}{N} \sum_{i=1}^n R_{n, i}\mathbb{E}\left[\left(\begin{array}{c}
    Y_{n,i}\\
    \tilde{X}_{n,i}\\
    \tilde{Z}_{n,i}
\end{array}\right)\left(\begin{array}{c}
    Y_{n,i}\\
    \tilde{X}_{n,i}\\
    \tilde{Z}_{n,i}
\end{array}\right)^{\prime}\mid\bR_n\right]\equiv\left(\begin{array}{ccc}
    \tilde{\Omega}_n^{YY} & \tilde{\Omega}_n^{YX} & \tilde{\Omega}_n^{YZ}\\
    \tilde{\Omega}_n^{XY} & \tilde{\Omega}_n^{XX} & \tilde{\Omega}_n^{XZ}\\
    \tilde{\Omega}_n^{ZY} & \tilde{\Omega}_n^{ZX} & \tilde{\Omega}_n^{ZZ}
\end{array}\right).
$$

\subsection{Preliminary Lemmas}
We will use the following results from \cite{kojevnikov2021limit}. We will only state the conditional version of the results, but also use the unconditional version of the results, which can be understood analogously.

Define
\begin{align*}
    \sigma^{2}_{n}=\Var(S_{n}\mid\bR_n),
\end{align*}
where $S_{n}=\sum_{i\in \mathcal{N}_{n}}U_{i,n}$.
\begin{condition}
    \label{cond:clt}
    A triangular array $\{U_{n,i}\}$ is conditionally $\psi$-dependent given $\bR_n$ with $\xi_{n}$ satisfying
    \begin{itemize}
        \item For some constant $C>0$,
        \begin{align*}
            \psi_{a,b}(f,g)\leq C\times ab(\|f\|_{\infty}+\Lip(f))(\|g\|_{\infty}+\Lip(g)).
        \end{align*}
        \item $\sup_{n}\max_{s\geq1}\xi_{n,s}<\infty$ a.s.
        \item For some $p>4$, $\sup_{n\geq 1}\max_{i\in \mathcal{N}_{n}}\mathbb{E}[|U_{n,i}|^{p}\mid\bR_n]<\infty$ a.s.
        \item There exists a positive sequence $m_{n}\to\infty$ such that for $k=1,2$,
        \begin{align*}
            &\frac{n}{\sigma_{n}^{2+k}}\sum_{s\geq 0}c_{n}(s,m_{n};k)\xi_{n,s}^{1-\frac{2+k}{p}}\overset{a.s.}{\longrightarrow} 0,\\
            &\frac{n^{2}\xi_{n,m_{n}}^{1-1/p}}{\sigma_{n}}\overset{a.s.}{\longrightarrow} 0.
        \end{align*}
        \item $\mathbb{E}[U_{n,i}\mid\bR_n]=0$.
    \end{itemize}
\end{condition}

\begin{lemma}[CLT, Theorem 3.2 in \citealp{kojevnikov2021limit}] Under \cref{cond:clt}, 
\begin{align*}
    \sup_{t\in\mathbb{R}}\left|\mathbb{P}\left\{\frac{S_{n}}{\sigma_{n}}\leq t\mid\bR_n\right\}-\Phi(t)\right|\overset{a.s.}{\longrightarrow}0\text{ as }n\to\infty,
\end{align*}
    where $\Phi$ denotes the distribution function of $\mathcal{N}(0,1)$.
\end{lemma}

\begin{lemma}[Linear Transformation, Lemma 2.1 in \citealp{kojevnikov2021limit}]
    \label{lem:linear}
    For each $n \geq 1$, let $\left\{a_{n, i}\right\}_{i \in \mathcal{N}_n}$ be a sequence of $\sigma(\bR_n)$-measurable vectors such that $\max _{i \in \mathcal{N}_n}\left\|a_{n, i}\right\| \leq 1$ a.s. 
    Under the first condition of \cref{cond:clt}, 
    the array $a_{n, i}^{\prime} U_{n, i}$ is conditionally $\psi$-dependent given $\bR_n$ with the dependence coefficients $\left\{\xi_n\right\}$.
\end{lemma}

\begin{condition}
    \label{cond:var_consistency}
    Let $\omega(x)=\mathds{1}\{|x|\leq1\}$.
    There exists $p>4$ such that
    \begin{itemize}
        \item $\sup _{n \geq 1} \max _{i \in \mathcal{N}_n}\mathbb{E}[|U_{n,i}|^{p}\mid\bR_n]<\infty$ a.s.
        \item $\lim _{n \rightarrow \infty} \sum_{s \geq 1}\left|\omega(s/2K)-1\right| \delta_n^{\partial}(s,1) \xi_{n,s}^{1-(2 / p)}=0$ a.s.
        \item $\lim _{n \rightarrow \infty} n^{-1} \sum_{s \geq 0} c_n\left(s, 2K ; 2\right) \xi_{n,s}^{1-(4 / p)}=0$ a.s.
    \end{itemize}
\end{condition}

\begin{lemma}[Variance Consistency, $2K$ Local Case of Proposition 4.1. in \citealp{kojevnikov2021limit}]
    \label{lem:var_consistency}
    Suppose that \cref{cond:clt,cond:var_consistency} hold. Then as $n \rightarrow \infty$,
    $$
        \mathbb{E}\left[\left\|\frac{1}{n}\sum_{i=1}^n \sum_{j \in \tilde{\mathcal{N}}_n(i;2K)} U_{n,i}U_{n,j}^{\prime}-\Var\left(\frac{S_n}{\sqrt{n}}\mid\bR_n\right)\right\|_F
        \mid\bR_n\right] \overset{a.s.}{\longrightarrow} 0,
    $$
    where $\|\cdot\|_F$ is the Frobenius norm.
    By Markov's inequality, we also have
    $$
        \frac{1}{n}\sum_{i=1}^n \sum_{j \in \tilde{\mathcal{N}}_n(i;2K)} U_{n,i}U_{n,j}^{\prime}-\Var\left(\frac{S_n}{\sqrt{n}}\mid\bR_n\right)\overset{p^R}{\longrightarrow}0.
    $$
\end{lemma}

\subsection{Main Lemmas}
\begin{lemma}
\label{lem:N_pos}
    Under $\rho_n n\to\infty$,
    \begin{align*}
        N>0 \text{ a.s. for large enough $n$}
    \end{align*}
\end{lemma}
\begin{proof}
    Since the result is trivial for $\rho_n=1$, we focus on the case $\rho_{n}\in(0,1)$.
    By the inequality $1-x \leq e^{-x}$ for $x \in(0,1)$, we have $\left(1-\rho_n\right)^n \leq e^{-n \rho_n}$.
    Thus,
    $$
        \sum_{n=1}^{\infty} \mathbb{P}(N=0)=\sum_{n=1}^{\infty} \mathbb{P}\left(\sum_{i=1}^{n}R_{n,i}=0\right) 
        =\sum_{n=1}^{\infty}\left(1-\rho_n\right)^n
        \leq \sum_{n=1}^{\infty} e^{-n \rho_n}.
    $$
    $\rho_n n\to\infty$ implies the right-hand side is bounded.
    By the Borel-Cantelli lemma, we can conclude.
\end{proof}

\begin{lemma}
\label{lem:nNconv}
    Under $\rho_n^2 n\to\infty$,
    \begin{align*}
        \frac{N}{n\rho_n}\overset{a.s.}{\longrightarrow}1
    \end{align*}
\end{lemma}
as $n\to\infty$.
\begin{proof}
    Pick any $\varepsilon>0$.
    By Hoeffding's inequality with $R_i\in[0,1]$,
    \begin{align*}
        \mathbb{P}\left(\left|\frac{N}{n \rho_n}-1\right|>\varepsilon\right)
        &=\mathbb{P}\left(\left|N-n \rho_n\right|>\varepsilon n \rho_n\right)
        =\mathbb{P}\left(\left|\sum_{i=1}^n R_i-n \rho_n\right|>\varepsilon n \rho_n\right)\\
        &\leq 2 \exp \left(-\frac{2 (\varepsilon n \rho_n)^2}{n}\right)=2 \exp \left(-2\varepsilon^2 n \rho_n^2\right).
    \end{align*}
    $\rho_n^2 n\to\infty$ implies $\sum_{n=1}^{\infty}\mathbb{P}\left(\left|\frac{N}{n \rho_n}-1\right|>\varepsilon\right)$ is bounded.
    From the Borel-Cantelli lemma, we can conclude.
    \end{proof}

\begin{lemma}
    \label{lem:as_rep}
    Assume that \cref{asm:moment,asm:linear_propensity} hold. Then, for large enough $n$,
    $$
        \Lambda_{n}=L_n,\quad X_{n,i}=T_{n,i}-\mathbb{E}[T_{n,i}|\bR_n]\quad\text{a.s.},
    $$
    and
    $$              
        \tilde{\Lambda}_{n}=\tilde{L}_{n},\quad\tilde{X}_{n,i}=\tilde{T}_{n,i}-\mathbb{E}[\tilde{T}_{n,i}|\bR_{n}]\quad\text{a.s.}
    $$
\end{lemma}
\begin{proof}
    Observe that $\Lambda_{n}=L_{n}$ a.s. for large enough $n$ as
    \begin{align*}
        \Lambda_{n}&=\sum_{i=1}^{n}\mathbb{E}[\mathbb{E}[T_{n,i}|\bR_n]Z_{n,i}']\left(\sum_{i=1}^{n}\mathbb{E}[Z_{n,i}Z_{n,i}']\right)^{-1}\\
        &=L_{n}\sum_{i=1}^{n}\mathbb{E}[Z_{n,i}Z_{n,i}']\left(\sum_{i=1}^{n}\mathbb{E}[Z_{n,i}Z_{n,i}']\right)^{-1}\\
        &=L_{n},
    \end{align*}
    where $\Lambda_{n}$ is well-defined by \cref{asm:moment} and the second equality holds
    by \cref{asm:linear_propensity}. Similarly, $\tilde{\Lambda}_{n}=\tilde{L}_{n}$ a.s. for large enough $n$ as
    \begin{align*}
        \tilde{\Lambda}_{n}&=\sum_{i=1}^{n}R_{n,i}\mathbb{E}[\tilde{T}_{n,i}|\bR_{n}]\tilde{Z}_{n,i}'\left(\sum_{i=1}^{n}R_{n,i}\tilde{Z}_{n,i}\tilde{Z}_{n,i}'\right)^{-1}\\
        &=\tilde{L}_{n}\sum_{i=1}^{n}R_{n,i}\tilde{Z}_{n,i}\tilde{Z}_{n,i}\left(\sum_{i=1}^{n}R_{n,i}\tilde{Z}_{n,i}\tilde{Z}_{n,i}'\right)^{-1}\\
        &=\tilde{L}_{n},
    \end{align*}
    where $\tilde{\Lambda}_{n}$ is well-defined by \cref{asm:moment} and the second equality holds by \cref{asm:linear_propensity}.
    
    Since we define $X_{n,i}=T_{n,i}-\Lambda_n Z_{n,i}$ and $\tilde{X}_{n,i}=\tilde{T}_{n,i}-\tilde{\Lambda}_{n}\tilde{Z}_{n,i}$, \cref{asm:linear_propensity} and the above two displayed qualities imply for large enough $n$, $X_{n,i}=T_{n,i}-\mathbb{E}[T_{n,i}|\bR_n]$ a.s. and $\tilde{X}_{n,i}=\tilde{T}_{n,i}-\mathbb{E}[\tilde{T}_{n,i}|\bR_{n}]$ a.s.
\end{proof}

\begin{lemma}
    \label{lem:Lambda}
    Suppose that $\tilde{T}_{n,i}=T_{n,i}$ and $\tilde{Z}_{n,i}=Z_{n,i}$ for all $i\in\mathcal{N}_{n}$ and $n\in\mathbb{N}$. Under \cref{asm:moment,asm:linear_propensity}, (i) $\tilde{\Lambda}_{n}=\Lambda_{n}$ a.s. and (ii) $\tilde{X}_{n,i}=X_{n,i}$ a.s.
\end{lemma}
\begin{proof}
    The results follow directly from \cref{lem:as_rep}.
\end{proof}
    
\begin{lemma}
    \label{lem:psi_dependency}
    Assume that \cref{asm:sampling,asm:linear,asm:moment,asm:linear_propensity,asm:local} hold.
    The following sequences of triangular arrays are $\psi$-dependent with $\xi_{n,s}=\mathds{1}\{s\leq 2K\}$:
    \begin{align*}
        X_{n,i}Z_{n,i}^{\prime}, \ 
        X_{n,i}X_{n,i}^{\prime}, \ 
        X_{n,i}Y_{n,i}, \ 
        Z_{n,i}Z_{n,i}^{\prime}, \ 
        Z_{n,i}Y_{n,i}.
    \end{align*}
    The following sequences of triangular arrays are conditionally $\psi$-dependent given $\bR_n$ with $\xi_{n,s}=\mathds{1}\{s\leq 2K\}$:
    \begin{align*}
        R_{n,i}\tilde{X}_{n,i}\tilde{Z}_{n,i}^{\prime}, \
        R_{n,i}\tilde{X}_{n,i}\tilde{T}_{n,i}^{\prime}, \
        R_{n,i}\tilde{X}_{n,i}Y_{n,i}, \ 
        R_{n,i}\tilde{Z}_{n,i}\tilde{Z}_{n,i}^{\prime}, \
        R_{n,i}\tilde{Z}_{n,i}Y_{n,i}.
    \end{align*}
\end{lemma}
\begin{proof}
    By \cref{asm:local}, we can set $\xi_{n,s}=\mathds{1}\{s\leq 2K\}$ for $s\geq 1$ since if $d_{n}(A,B)>2K$, $f(U_{n,A})\indep g(U_{n,B})$ for any $f\in\mathcal{L}_{v,a}$ and $g\in\mathcal{L}_{v,b}$ as long as $U_{n,i}$ are based on $\tilde{T}_{n,i},T_{n,i},\tilde{Z}_{n,i},Z_{n,i},\tilde{Y}_{n,i},Y_{n,i}$. 
    For large enough $n$, \cref{lem:as_rep} implies $X_{n,i}=T_{n,i}-\mathbb{E}[T_{n,i}|\bR_n]$ and $\tilde{X}_{n,i}=\tilde{T}_{n,i}-\mathbb{E}[\tilde{T}_{n,i}|\bR_{n}]$ almost surely.
    Thus, for large enough $n$, $X_{n,i}$ and  $\tilde{X}_{n,i}$ also have the local dependence with $2K$.
    By \cref{asm:moment}, each element is uniformly bounded. Thus, we can set $\psi_{a,b}(f,g)=2\|f\|_{\infty}\|g\|_{\infty}$ for any  $f\in\mathcal{L}_{v,a}$ and $g\in\mathcal{L}_{v,b}$. This completes the proof.
\end{proof}

\begin{lemma}
    \label{lem:eps_bound}
    Under \cref{asm:moment},
    $$\max_{i}|\tilde{
        \varepsilon
        }_{n,i}|<\infty\text{ a.s. }\qquad\text{and}\qquad\max_{i}|
        \varepsilon_{n,i}|<\infty\text{ a.s. }$$
\end{lemma}
\begin{proof}
    Under the uniform boundedness and the invertibility condition (\cref{asm:moment}), $\|\theta_{n}^{\causample}\|<\infty$ a.s. and $\|\gamma_{n}^{\causample}\|<\infty$ a.s.
    Thus, by the Schwarz Inequality,
    \begin{align*}
        |\tilde{\varepsilon}_{n,i}|&\leq\max_{i}|Y_{n,i}|+\max_{i}\|\tilde{X}_{n,i}\|\times\|\theta_{n}^{\causample}\|+\max_{i}\|\tilde{Z}_{n,i}\|\times\|\gamma_{n}^{\causample}\|\\
        &<\infty\quad\text{a.s.}
    \end{align*}
    for all $i$.
    The bound for $|\varepsilon_{n,i}|$ can be derived similarly.
\end{proof}

\begin{lemma}
    \label{lem:LLN}
    Under \cref{asm:linear,asm:sampling,asm:moment,asm:linear_propensity,asm:local,asm:sparsity}, 
    \begin{align*}
        \tilde{Q}_{n}-\tilde{\Omega}_{n}\overset{p^R}{\longrightarrow}0
        \quad\text{and}\quad
        \tilde{Q}_{n}-\tilde{\Omega}_{n}\overset{p}{\to}0.
    \end{align*}
\end{lemma}
\begin{proof}
    Let $W_{n,i}\equiv (Y_{n,i},\tilde{X}_{n,i},\tilde{Z}_{n,i})'$. Then,
    \begin{align*}
        \tilde{Q}_{n}-\tilde{\Omega}_{n}&=\frac{1}{N}\sum_{i=1}^{n}R_{n,i}(W_{n,i}W_{n,i}'-\mathbb{E}[W_{n,i}W_{n,i}'|\bR_{n}])\\
        &=\frac{n\rho_{n}}{N}\times \frac{1}{n\rho_{n}}\sum_{i=1}^{n}R_{n,i}\left(W_{n,i}W_{n,i}'-\mathbb{E}[W_{n,i}W_{n,i}'|\bR_{n}]\right).
    \end{align*}
    Since $(n\rho_{n})/N\overset{a.s.}{\longrightarrow}1$ (\cref{lem:nNconv}) implies $(n\rho_{n})/N\overset{p^R}{\longrightarrow}1$, it suffices to show that
    \begin{align*}
         \frac{1}{n\rho_{n}}\sum_{i=1}^{n}R_{n,i}\left(W_{n,i}W_{n,i}'-\mathbb{E}[W_{n,i}W_{n,i}'|\bR_{n}]\right)\overset{p^R}{\longrightarrow}0.
    \end{align*}
    We will show it by verifying
    \begin{align*}
        \mathbb{E}\left[\left(\frac{1}{n\rho_{n}}\sum_{i=1}^{n}R_{n,i}\left(W_{n,i,(k)}W_{n,i,(\ell)}-\mathbb{E}[W_{n,i,(k)}W_{n,i,(\ell)}\mid\bR_{n}]\right)\right)^{2}\mid\bR_{n}\right]\overset{a.s.}{\longrightarrow}0
    \end{align*}
    for all $k,\ell=1,\dots,d_{\tilde{T}}$.
    Observe that 
    \begin{align}
        &\mathbb{E}\left[\left(\frac{1}{n\rho_{n}}\sum_{i=1}^{n}R_{n,i}\left(W_{n,i,(k)}W_{n,i,(\ell)}-\mathbb{E}[W_{n,i,(k)}W_{n,i,(\ell)}\mid\bR_{n}]\right)\right)^{2}\mid\bR_{n}\right]\nonumber\\
        =& \frac{1}{n^{2}\rho_{n}^{2}}\sum_{i=1}^{n}R_{n,i}\mathbb{E}\left[(W_{n,i,(k)}W_{n,i,(\ell)}-\mathbb{E}[W_{n,i,(k)}W_{n,i,(\ell)}|\bR_{n}])^{2}\mid\bR_{n}\right]\label{eq:WW_var}\\
        &+\frac{1}{n^{2}\rho_{n}^{2}}\sum_{i\neq j}R_{n,i}R_{n,j}\mathbb{E}[(W_{n,i,(k)}W_{n,i,(\ell)}-\mathbb{E}[W_{n,i,(k)}W_{n,i,(\ell)}|\bR_{n}])\nonumber\\
        &\hspace{80pt}
        \times(W_{n,j,(k)}W_{n,j,(\ell)}-\mathbb{E}[W_{n,j,(k)}W_{n,j,(\ell)}|\bR_{n}])\mid\bR_{n}]\label{eq:WW_cov}
    \end{align}
    
    For \cref{eq:WW_var}, since there is some absolute constant $C$ such that $|W_{n,j,(k)}W_{n,j,(\ell)}|<C$ by \cref{asm:moment},
    \begin{align*}
        \cref{eq:WW_var}
        &\leq \frac{1}{n^{2}\rho_{n}^{2}}\sum_{i=1}^n(2C)^2
        = 4C^{2}\times \frac{1}{n\rho_{n}^2}\to0
    \end{align*} 
    where the inequality and the convergence do not depend on $\bR_n$.
    
    For $\cref{eq:WW_cov}$, note that if $d_{n}(i,j)>2K$, then
    \begin{align*}
        \mathbb{E}[(W_{n,i,(k)}W_{n,i,(\ell)}-\mathbb{E}[W_{n,i,(k)}W_{n,i,(\ell)}|\bR_{n}])
        (W_{n,j,(k)}W_{n,j,(\ell)}-\mathbb{E}[W_{n,j,(k)}W_{n,j,(\ell)}|\bR_{n}])\mid\bR_{n}]=0
    \end{align*}
    as $R_{n,i}$ is i.i.d and $(T_{n,i},\tilde{T}_{n,i})\indep (T_{n,j},\tilde{T}_{n,j})$ with no overlap in $\bD_{n}$ and $\bR_{n}$. Thus,
    \begin{align*}
        \cref{eq:WW_cov} &= \frac{1}{n^{2}\rho^{2}_{n}}\sum_{i=1}^{n}\sum_{j\in\mathcal{N}(i,2K)\setminus\{i\}}R_{n,i}R_{n,j}\mathbb{E}[(W_{n,i,(k)}W_{n,i,(\ell)}-\mathbb{E}[W_{n,i,(k)}W_{n,i,(\ell)}|\bR_{n}])\nonumber\\
        &\hspace{120pt}
        \times(W_{n,j,(k)}W_{n,j,(\ell)}-\mathbb{E}[W_{n,j,(k)}W_{n,j,(\ell)}|\bR_{n}])\mid\bR_{n}]\\
        &\leq 4C^{2}\times\frac{1}{n\rho_{n}^2}\sum_{1\leq s\leq 2K}\delta_{n}^{\partial}(s;1)\to0,
    \end{align*}
    where the last line holds by \cref{asm:sparsity}, and the inequality and the convergence do not depend on $\bR_n$.
    
    Thus, by Markov's inequality for $\left(\frac{1}{n\rho_{n}}\sum_{i=1}^{n}R_{n,i}\left(W_{n,i,(k)}W_{n,i,(\ell)}-\mathbb{E}[W_{n,i,(k)}W_{n,i,(\ell)}\mid\bR_{n}]\right)\right)^{2}$, 
    $$
    \frac{1}{n\rho_{n}}\sum_{i=1}^{n}R_{n,i}\left(W_{n,i,(k)}W_{n,i,(\ell)}-\mathbb{E}[W_{n,i,(k)}W_{n,i,(\ell)}\mid\bR_{n}]\right)\overset{p^R}{\longrightarrow}0,
    $$
    and
    \begin{align*}
        \tilde{Q}_{n}-\tilde{\Omega}_{n}\overset{p^R}{\longrightarrow} 0.
    \end{align*}

    Unconditional consistency can be shown easily from this result.
    Since a conditional probability is bounded, the dominated convergence theorem and the law of iterated expectations imply $\tilde{Q}_{n}-\tilde{\Omega}_{n}\overset{p}{\to} 0.$
\end{proof}

\begin{lemma}
    \label{lem:LLN_W}
    Let $W_{n,i}$ be a scalar random variable satisfying $|W_{n,i}|\leq\bar{W}<\infty$ a.s.
    We allow $W_{n,i}$ to depend on $\bR_n$ and $\bD_n$, but assume that $W_{n,i}\indep R_{n,i}$ and $W_{n,i}\indep W_{n,j}$ if $d_n(i,j)>2K$. Then, under \cref{asm:sampling,asm:sparsity},
    \begin{align*}
        \frac{1}{N}\sum_{i=1}^nR_{n,i}\mathbb{E}[W_{n,i}|\bR_{n}]-\frac{1}{n}\sum_{i=1}^n\mathbb{E}[W_{n,i}]\overset{p}{\to} 0.
    \end{align*}
\end{lemma}
\begin{proof}
    By \cref{lem:nNconv}, 
    \begin{align*}
        \frac{1}{N}\sum_{i=1}^nR_{n,i}\mathbb{E}[W_{n,i}|\bR_{n}]=\frac{1}{n}\sum_{i=1}^n\frac{R_{n,i}}{\rho_n}\mathbb{E}[W_{n,i}|\bR_{n}]+o_p(1).
    \end{align*}
    Thus, it suffices to show that
    \begin{align}
       \mathbb{E}\left[\left(\frac{1}{n}\sum_{i=1}^n\frac{R_{n,i}}{\rho_n}\mathbb{E}[W_{n,i}|\bR_{n}]-\frac{1}{n}\sum_{i=1}^n\mathbb{E}[W_{n,i}]\right)^{2}\right]\to0
       \label{eq:consistency_causal}
    \end{align}
    The left-hand side of \cref{eq:consistency_causal} is given by
    \begin{align}
        &\frac{1}{n^{2}}\sum_{i=1}^n\mathbb{E}\left[\left(\frac{R_{n,i}}{\rho_n}\mathbb{E}[W_{i,n}|\bR_{n}]-\mathbb{E}[W_{n,i}]\right)^{2}\right]\label{eq:cc_var}\\
        &+\frac{1}{n^{2}}\sum_{i\neq j}\mathbb{E}\left[\left(\frac{R_{n,i}}{\rho_n}\mathbb{E}[W_{n,i}|\bR_{n }]-\mathbb{E}[W_{n,i}]\right)\left(\frac{R_{n,j}}{\rho_n}\mathbb{E}[W_{n,j}|\bR_{n}]-\mathbb{E}[W_{n,j}]\right)\right]\label{eq:cc_cov}
    \end{align}
    For \cref{eq:cc_var}, we have
        \begin{align*}
            \cref{eq:cc_var}&\leq \frac{2}{n^{2}}\sum_{i=1}^n\mathbb{E}\left[\left(\frac{R_{n,i}}{\rho_n}\right)^{2}(\mathbb{E}[W_{i,n}|\bR_{n}])^{2}+(\mathbb{E}[W_{n,i}])^{2}\right]\\
            &\leq \frac{2\bar{W}^2}{n}\left[\mathbb{E}\left[\left(\frac{R_{n,i}}{\rho_n}\right)^{2}\right]+1\right],
        \end{align*}
        where the first inequality holds from the inequality $(a-b)^{2}\leq 2(a^{2}+b^{2})$ for any $a,b\in\mathbb{R}$ and the second inequality holds by the uniform boundedness. Note that
        \begin{align*}
            \frac{1}{n}\mathbb{E}\left[\left(\frac{R_{n,i}}{\rho_n}\right)^{2}\right]&=\frac{1}{n\rho_n}=o(1).
        \end{align*}
    
    For \cref{eq:cc_cov},
    \begin{align*}
        \cref{eq:cc_cov}
        &=\frac{1}{n^{2}}\sum_{i=1}^n\sum_{j\in\mathcal{N}_{n}(i,2K)\setminus\{i\}}\mathbb{E}\left[\left(\frac{R_{n,i}}{\rho_n}\mathbb{E}[W_{n,i}|\bR_{n}]-\mathbb{E}[W_{n,i}]\right)\left(\frac{R_{n,j}}{\rho_n}\mathbb{E}[W_{n,j}|\bR_{n}]-\mathbb{E}[W_{n,j}]\right)\right]\\
        &\leq\frac{\bar{W}^2}{n^{2}}\sum_{i=1}^{n}\sum_{j\in\mathcal{N}_{n}(i,2K)\setminus\{i\}}\mathbb{E}\left[\left|\frac{R_{n,i}}{\rho_n}-1\right|\cdot\left|\frac{R_{n,j}}{\rho_n}-1\right|\right]\\
        &\leq\frac{\bar{W}^2}{n^{2}}\sum_{i=1}^{n}\sum_{j\in\mathcal{N}_{n}(i,2K)\setminus\{i\}}\mathbb{E}\left[\left(\frac{R_{n,i}}{\rho_n}-1\right)^2\right]\\
        &=\left(\frac{1}{\rho_n}-1\right)\frac{\bar{W}^2}{n^{2}}\sum_{i=1}^{n}\sum_{j\in\mathcal{N}_{n}(i,2K)\setminus\{i\}}1=O\left(\frac{1}{n\rho_n}\right)\sum_{1\leq s\leq2K}\delta_n^{\partial}(s;1)=o(1),
    \end{align*}
    where the first equality holds by $W_{n,i}\indep R_{n,j}$, $W_{n,i}\indep W_{n,j}$ if $d_{n}(i,j)>2K$, and \cref{asm:sampling}, the first inequality holds by the uniform boundedness, the next inequality holds by the Cauchy-Schwarz inequality, and the last step follows from \cref{asm:sparsity}.
    
    Combining the arguments for $\cref{eq:cc_var}$ and \cref{eq:cc_cov}, we have shown the convergence \cref{eq:consistency_causal} as $n\to\infty$.
\end{proof}

\begin{lemma}
    \label{lem:LLN_RW}
    Let $W_{n,i}$ be a scalar random variable satisfying $|W_{n,i}|\leq\bar{W}<\infty$ a.s.
    We allow $W_{n,i}$ to depend on $\bR_n$ and $\bD_n$, but assume that $W_{n,i}\indep R_{n,i}$ and $W_{n,i}\indep W_{n,j}$ if $d_n(i,j)>2K$. Then, under \cref{asm:sampling,asm:sparsity},
    \begin{align*}
        \frac{1}{N}\sum_{i=1}^nR_{n,i}\mathbb{E}[R_{n,i}W_{n,i}|\bR_{n}]-\frac{1}{n\rho_n}\sum_{i=1}^n\mathbb{E}[R_{n,i}W_{n,i}]
        &=\frac{1}{N}\sum_{i=1}^nR_{n,i}\mathbb{E}[W_{n,i}|\bR_{n}]-\frac{1}{n}\sum_{i=1}^n\mathbb{E}[W_{n,i}]\\
        &\overset{p}{\to}0.
    \end{align*}
\end{lemma}
\begin{proof}
    The result follows by the same logic as \cref{lem:LLN_W}.
\end{proof}

\begin{lemma}\label{lem:CLT} 
    Under \cref{asm:sampling,asm:linear,asm:moment,asm:linear_propensity,asm:local,asm:sparsity,asm:dependence},
    \begin{align*}
        \tilde{\Sigma}_{n}^{-1/2}
        \sum_{i=1}^{n}R_{n,i}\tilde{X}_{n,i}\tilde{\varepsilon}_{n,i} \overset{d^R}{\longrightarrow} \mathrm{N}(0,I_{d_{\tilde{T}}}).
    \end{align*}
\end{lemma}

\begin{proof}
    We use the Cramer-Wold device and verify \cref{cond:clt} for any given $a\in\mathbb{R}^{|\mathcal{T}|}$.

    First, we will transform the statistics and verify the zero expectation condition.
    The orthogonality condition for $\theta_{n}^{\causample}$ \cref{eq:moment_causalsample} implies
    \begin{align}
        \sum_{i=1}^n R_{n,i}\mathbb{E}\left[\tilde{X}_{n,i}\tilde{\varepsilon}_{n,i} \mid \bR_n\right] &=0.
    \end{align}
    
    Define $U_{n,i}\equiv R_{n,i}\tilde{X}_{n,i}\tilde{\varepsilon}_{n,i}-\mathbb{E}\left[R_{n,i}\tilde{X}_{n,i}\tilde{\varepsilon}_{n,i}\mid\bR_n\right]$.
    Then, $\tilde{\Sigma}_{n}^{-1/2}\sum_{i=1}^{n}R_{n,i}\tilde{X}_{n,i}\tilde{\varepsilon}_{n,i}
    =\tilde{\Sigma}_{n}^{-1/2}\sum_{i=1}^{n}U_{n,i}$
    and we have $\mathbb{E}[U_{n,i}\mid\bR_n]=0$ for all $i$.
    
    By the Cramer-Wold device, it suffices to show that 
    $$
    \frac{\sum_{i=1}^n a^{\prime}U_{n, i}}{\sqrt{a^{\prime}\tilde{\Sigma}_{n}a}} \overset{d^R}{\longrightarrow} \mathrm{N}(0,1)
    $$
    for any $a\in\mathbb{R}^{d_{\tilde{T}}}$ with $a^{\prime}a=1$.

    By \cref{lem:psi_dependency,lem:linear}, $a^{\prime}U_{n, i}$ is conditionally $\psi$-dependent with $\xi_{n,s}=\mathds{1}\{s\leq 2K\}$ given $\bR_n$.
    The other conditions are assumed in \cref{asm:dependence} or automatically satisfied under the local dependence (\cref{asm:local}).
\end{proof}

\begin{lemma}\label{lem:CLT_causal} 
    Under \cref{asm:exposure,asm:linear,asm:sampling,asm:moment,asm:linear_propensity,asm:local,asm:sparsity,asm:dependence},
    \begin{align*}
        \Sigma_{n}^{-1/2}
        \sum_{i=1}^{n}R_{n,i}X_{n,i}\varepsilon_{n,i} \overset{d}{\to} \mathrm{N}(0,I_{d_{T}}).
    \end{align*}
\end{lemma}
\begin{proof}
    An orthogonality condition for $\theta_{n}^{\causal}$ \cref{eq:moment_causal} implies
    \begin{align}
        \sum_{i=1}^n \mathbb{E}\left[X_{n,i}\varepsilon_{n,i}\right]  &=0.
    \end{align}
    By \cref{asm:linear,asm:exposure} (i),
    \begin{align*}
        X_{n,i}\varepsilon_{n,i}&=X_{n,i}(Y_{n,i}-X_{n,i}'\theta_{n}^{\causal}-Z_{n,i}'\gamma_{n}^{\causal})\\
        &=X_{n,i}T_{n,i}'\theta_{n,i}+X_{n,i}\nu_{n,i}-X_{n,i}X_{n,i}'\theta_{n}^{\causal}
        -X_{n,i}Z_{n,i}'\gamma_{n}^{\causal}.
    \end{align*}
    By \cref{lem:as_rep,asm:exposure} (ii), $R_{n,i}$ enters only multiplicatively for $T_{n,i}$ and $X_{n,i}=T_{n,i}-\mathbb{E}[T_{n,i}|\bR_n]$.
    By \cref{asm:exposure} (ii), each element of $Z_{n,i}$ is multiplicatively in $R_{n,i}$.
    Thus, each element of $X_{n,i}\varepsilon_{n,i}$ is multiplicatively in $R_{n,i}$ by $R_{n,i}^2=R_{n,i}$.
    Combining it with the orthogonality, 
    \begin{align*}
        \sum_{i=1}^{n}\mathbb{E}[R_{n,i}X_{n,i}\varepsilon_{n,i}]=0.
    \end{align*}
    
    Define $U_{n,i}=R_{n,i}X_{n,i}\varepsilon_{n,i}-\mathbb{E}[R_{n,i}X_{n,i}\varepsilon_{n,i}]$. Then, we have 
    \begin{align*}
        \Sigma_{n}^{-1/2}\sum_{i=1}^{n}R_{n,i}X_{n,i}\varepsilon_{n,i}
        &=\Sigma_{n}^{-1/2}\sum_{i=1}^{n}U_{n,i}+\Sigma^{-1/2}_{n}\sum_{i=1}^{n}\mathbb{E}[R_{n,i}X_{n,i}\varepsilon_{n,i}]\\
        &=\Sigma_{n}^{-1/2}\sum_{i=1}^{n}U_{n,i},
    \end{align*}
    and $\mathbb{E}[U_{n,i}]=0$.
    
    The remaining parts of the proof are similar to \cref{lem:CLT}.
\end{proof}

\begin{lemma}
    \label{lem:gamma_consistency}
    Under \cref{asm:linear,asm:sampling,asm:moment,asm:linear_propensity,asm:local,asm:sparsity}, $$\hat{\gamma}_{n}-\gamma_{n}^{\causample}\overset{p^R}{\longrightarrow}0.$$
    If we assume \cref{asm:exposure} additionally, 
    $$\tilde{\gamma}_{n}-\gamma_{n}^{\causal}\overset{p}{\to}0,$$
    where $\tilde{\gamma}_{n}$ is defined in \cref{eq:gamma}.
\end{lemma}

\begin{proof}
    We can show $\hat{\gamma}_{n}-\gamma_{n}^{\causample}\overset{p^R}{\longrightarrow}0$ by \cref{lem:LLN} as the proof for \cref{thm:consistency}.

    Next, we show $\tilde{\gamma}_{n}-\gamma_{n}^{\causal}\overset{p}{\to}0$. 
    By \cref{lem:LLN}, $\tilde{P}_n^{ZZ}\overset{p^R}{\longrightarrow}\mathbb{E}[\tilde{P}_n^{ZZ}|\bR_n]$ and $\tilde{P}_n^{ZY}\overset{p^R}{\longrightarrow}\mathbb{E}[\tilde{P}_n^{ZY}|\bR_n]$.
    By \cref{lem:LLN_W} and \cref{lem:LLN_RW},
    $\mathbb{E}[\tilde{P}_n^{ZZ}|\bR_n]\overset{p}{\to}\Omega_n^{ZZ}$ and $\mathbb{E}[\tilde{P}_n^{ZY}|\bR_n]\overset{p}{\to}\Omega_n^{ZY}$. Thus, we can conclude by the continuous mapping theorem.
\end{proof}
    
\section{Proofs}
\label{app:proof}
\subsection{Proof of \cref{thm:weakly_causal}}
\begin{proof}
    \cref{lem:as_rep} implies that
    \begin{align*}
        \Omega_{n}^{XZ}=0=\mathbb{E}[(T_{n,i}-\mathbb{E}[T_{n,i}|\bR_n])Z_{n,i}']=0
    \end{align*}
     for large enough $n$.
    Similarly, 
    \begin{align*}
        \tilde{\Omega}_{n}^{XZ}=\mathbb{E}[\tilde{X}_{n,i}\tilde{Z}_{n,i}'|\bR_{n}] 
        = \mathbb{E}[(\tilde{T}_{n,i}-\mathbb{E}[\tilde{T}_{n,i}|\bR_{n}])\tilde{Z}_{n,i}'|\bR_{n}] = 0\quad\text{a.s.}
    \end{align*}
    for large enough $n$ since $\tilde{Z}_{n,i}$ is measurable with respect to $\sigma(\bR_{n})$.

    Therefore, for large enough $n$,
    \begin{align*}
        \theta_{n}^{\causal}&=\left(\Omega_{n}^{XX}\right)^{-1}\Omega_{n}^{XY},
    \end{align*}
    and
    \begin{align*}
        \theta_{n}^{\causample}&=\left(\tilde{\Omega}_{n}^{XX}\right)^{-1}\tilde{\Omega}_{n}^{XY}\quad\text{a.s.}
    \end{align*}
    They are well-defined under \cref{asm:moment}.
    Then, it suffices to show that for large enough $n$,
    \begin{align*}
        \mathbb{E}[X_{n,i}Y_{n,i}]&=\mathbb{E}[X_{n,i}X_{n,i}']\theta_{n,i},
       \end{align*}
    and
    \begin{align*}
        \mathbb{E}[\tilde{X}_{n,i}Y_{n,i}|\bR_{n}]&=\mathbb{E}[\tilde{X}_{n,i}X_{n,i}'|\bR_{n}]\theta_{n,i}\quad\text{a.s.}
    \end{align*}
    
    The following transformations hold for large enough $n$:
    \begin{align*}
        \mathbb{E}[X_{n,i}Y_{n,i}]&=\mathbb{E}[X_{n,i}T_{n,i}']\theta_{n,i}+ \mathbb{E}[X_{n,i}]\nu_{n,i}\\
        &=\mathbb{E}[X_{n,i}X_{n,i}']\theta_{n,i}+\mathbb{E}[X_{n,i}(T_{n,i}-X_{n,i})']\theta_{n,i}\\
        &=\mathbb{E}[X_{n,i}X_{n,i}']\theta_{n,i}+\mathbb{E}[X_{n,i}]Z_{n,i}'\Lambda_{n}'\theta_{n,i}\\
        &=\mathbb{E}[X_{n,i}X_{n,i}']\theta_{n,i},
    \end{align*}
    where the first equality holds by \cref{asm:linear}, the second and the last equalities follow by $\mathbb{E}[X_{n,i}]=0$, which is implied by \cref{lem:as_rep}, and the third equality follows by the definition of $X_{n,i}$. Similarly, the following transformations hold almost surely for large enough $n$:
    \begin{align*}
        \mathbb{E}[\tilde{X}_{n,i}Y_{n,i}]&=\mathbb{E}[\tilde{X}_{n,i}T_{n,i}'|\bR_{n}]\theta_{n,i}+\mathbb{E}[\tilde{X}_{n,i}|\bR_{n}]\nu_{n,i}\\
        &=\mathbb{E}[\tilde{X}_{n,i}X_{n,i}|\bR_{n}]\theta_{n,i}+\mathbb{E}[\tilde{X}_{n,i}(T_{n,i}-X_{n,i})'|\bR_{n}]\theta_{n,i}\\
        &=\mathbb{E}[\tilde{X}_{n,i}X_{n,i}'|\bR_{n}]\theta_{n,i}+\mathbb{E}[\tilde{X}_{n,i}|\bR_{n}]Z_{n,i}'\Lambda_{n}'\theta_{n,i}\\
        &=\mathbb{E}[\tilde{X}_{n,i}X_{n,i}'|\bR_{n}]\theta_{n,i},
    \end{align*}
    where we used $\mathbb{E}[\tilde{X}_{n,i}|\bR_{n}]=0$. This completes the proof.
\end{proof}

\subsection{Proof of \cref{cor:causal_element}}
\begin{proof}
    By the population version of the Frisch-Waugh-Lovell theorem,
    \begin{align*}
        \theta_{n,(k)}^{\causal}=\frac{\sum_{i=1}^n\mathbb{E}[U_{n,i,(k)}Y_{n,i}]}{\sum_{i=1}^n\mathbb{E}[U_{n,i,(k)}^{2}]}
    \end{align*}
    By the linearity of the model (\cref{asm:linear}), the numerator can be transformed as
    \begin{align*}
        \sum_{i=1}^n\mathbb{E}[U_{n,i,(k)}Y_{n,i}]
        &=\sum_{i=1}^n\mathbb{E}[U_{n,i,(k)}T_{n,i}']\theta_{n,i}+\sum_{i=1}^n\mathbb{E}[U_{n,i,(k)}]\nu_{n,i}\\
        &=\sum_{i=1}^n\mathbb{E}[U_{n,i,(k)}(X_{n,i}+\mathbb{E}[T_{n,i}\mid\bR_n])']\theta_{n,i}\\
        &=\sum_{i=1}^n\mathbb{E}[U_{n,i,(k)}X_{n,i,(k)}]\theta_{n,i,(k)} + \sum_{i=1}^n\mathbb{E}[U_{n,i,(k)}X_{n,i,(-k)}']\theta_{n,i,(-k)},
    \end{align*}
    where the second equality holds as $\mathbb{E}[U_{n,i,(k)}]=0$, which is implied by $\mathbb{E}[X_{n,i}]=0$, a consequence of \cref{lem:as_rep} and the last equality follows from the law of iterated expectations and $\mathbb{E}[X_{n,i}\mid\bR_n]=0$.
    
    Similarly,
    \begin{align*}
        \theta_{n,(k)}^{\causample}=\frac{\sum_{i=1}^n\mathbb{E}[\tilde{U}_{n,i,(k)}Y_{n,i}|\bR_{n}]}{\sum_{i=1}^n\mathbb{E}[\tilde{U}_{n,i,(k)}^{2}|\bR_{n}]}.
    \end{align*}
    The numerator is given by
    \begin{align*}
        \sum_{i=1}^nR_{n,i}\mathbb{E}[\tilde{U}_{n,i,(k)}Y_{n,i}|\bR_{n}]
        &=\sum_{i=1}^nR_{n,i}\mathbb{E}[\tilde{U}_{n,i,(k)}T_{n,i}'|\bR_{n}]\theta_{n,i}\\
        &=\sum_{i=1}^nR_{n,i}\sum_{l=1}^{d_{T}}\mathbb{E}[\tilde{U}_{n,i,(k)}X_{n,i,(l)}|\bR_{n}]\theta_{n,i,(l)}.
    \end{align*}
    Under $d_T=d_{\tilde{T}}$, the last equation can be simplified further to
    \begin{align*}
        &\sum_{i=1}^nR_{n,i}\mathbb{E}[\tilde{U}_{n,i,(k)}X_{n,i,(k)}|\bR_{n}]\theta_{n,i,(k)}+\sum_{i=1}^n\sum_{l\neq k}R_{n,i}\mathbb{E}[\tilde{U}_{n,i,(k)}X_{n,i,(l)}|\bR_{n}]\theta_{n,i,(l)}
    \end{align*}
    as above. This completes the proof.
\end{proof}

\subsection{Proof of \cref{cor:no_contamination}}
\begin{proof}
    By \cref{lem:as_rep},
    \begin{align*}
        &\mathbb{E}[\tilde{X}_{n,i,(k)}X_{n,i,(l)}|\bR_{n}]\\
        =&\mathbb{E}[(\tilde{T}_{n,i,(k)}-\mathbb{E}[\tilde{T}_{n,i,(k)}|\bR_{n}])(T_{n,i,(l)}-\mathbb{E}[T_{n,i,(l)}])|\bR_{n}]\\
        =&\mathbb{E}[(\tilde{T}_{n,i,(k)}-\mathbb{E}[\tilde{T}_{n,i,(k)}|\bR_{n}])(T_{n,i,(l)}-\mathbb{E}[T_{n,i,(l)}|\bR_{n}]+\mathbb{E}[T_{n,i,(l)}|\bR_{n}]-\mathbb{E}[T_{n,i,(l)}])|\bR_{n}]\\
        =&\mathbb{E}[(\tilde{T}_{n,i,(k)}-\mathbb{E}[\tilde{T}_{n,i,(k)}|\bR_{n}])(T_{n,i,(l)}-\mathbb{E}[T_{n,i,(l)}|\bR_{n}])|\bR_{n}]\\
        =&\Cov(\tilde{T}_{n,i,(k)},T_{n,i,(l)}|\bR_{n}).
    \end{align*}
    Also, by the law of iterated expectations,
    \begin{align*}
        \mathbb{E}[X_{n,i,(k)}X_{n,i,(l)}]
        &=\mathbb{E}[\mathbb{E}[X_{n,i,(k)}X_{n,i,(l)}|\bR_{n}]]\\
        &=\mathbb{E}[\mathbb{E}[(T_{n,i,(k)}-\mathbb{E}[T_{n,i,(k)}|\bR_{n}])(T_{n,i,(l)}-\mathbb{E}[T_{n,i,(l)}|\bR_{n}])]|\bR_{n}]]\\
        &=\mathbb{E}[\Cov(T_{n,i,(k)},T_{n,i,(l)}|\bR_n)].
    \end{align*}
    By \cref{thm:weakly_causal} and the above equivalences, the no contamination result follows if the covariance condition is satisfied.
    
    Moreover, the numerator of $\theta_{n,(k)}^{\causample}$ is 
    \begin{align*}
        &\sum_{i=1}^nR_{n,i}\mathbb{E}[\tilde{X}_{n,i,(k)}X_{n,i,(k)}|\bR_{n}]\theta_{n,i,(k)}+\sum_{i=1}^n\sum_{l\in\{1,\cdots,d_T\}\setminus \{k\}}R_{n,i}\mathbb{E}[\tilde{X}_{n,i,(k)}X_{n,i,(l)}|\bR_{n}]\theta_{n,i,(l)}\\
        =&\sum_{i=1}^nR_{n,i}\mathbb{E}[\tilde{X}_{n,i,(k)}X_{n,i,(k)}|\bR_{n}]\theta_{n,i,(k)},
    \end{align*}
    and $R_{n,i}\mathbb{E}[\tilde{X}_{n,i,(k)}X_{n,i,(k)}|\bR_{n}]\geq0$ if we assume that $\Cov(\tilde{T}_{n,i,(k)},T_{n,i,(k)}|\bR_{n})\geq0$.
\end{proof}

\subsection{Proof of \cref{thm:consistency}}
\begin{proof}
    By \cref{lem:as_rep}, we have $\mathbb{E}[\tilde{X}_{n,i}\tilde{Z}_{n,i}|\bR_{n}]=0$ a.s. for large enough $n$. Thus, $\tilde{Q}^{ZX}_{n},\tilde{Q}_{n}^{XZ}\overset{a.s.}{\longrightarrow}0$. Since
    $$
    \left(
    \begin{array}{c}
    \hat{\theta}_{n}\\
    \hat{\gamma}_{n}
    \end{array}
    \right)
    =
    \left(
    \begin{array}{cc}
    \tilde{Q}_{n}^{XX}&\tilde{Q}_{n}^{XZ}\\
    \tilde{Q}_{n}^{ZX}&\tilde{Q}_{n}^{ZZ}
    \end{array}\right)^{-1}
    \left(
    \begin{array}{c}
    \tilde{Q}_{n}^{XY}\\
    \tilde{Q}_{n}^{ZY}
    \end{array}
    \right),
    $$
    \cref{lem:LLN} implies that
    $$
    \hat{\theta}_{n}-\theta_{n}^{\causample}=\hat{\theta}_{n}-(\tilde{\Omega}_{n}^{XX})^{-1}\tilde{\Omega}_{n}^{XY}+o_{p^R}(1)\overset{p^R}{\longrightarrow}0,
    $$
    which further implies
    $$
    \hat{\theta}_{n}-\theta_{n}^{\causample}\overset{p}{\to}0.
    $$
\end{proof}

\subsection{Proof of \cref{thm:causal_consistency}}
\begin{proof}
    Since we have already shown \cref{thm:consistency}, it suffices to prove $\theta_{n}^{\causample}-\theta_{n}^{\causal}\overset{p}{\to}0$.
    By \cref{lem:as_rep}, we have $\mathbb{E}[X_{n,i}Z_{n,i}|\bR_{n}]=0$ a.s. and $\mathbb{E}[X_{n,i}Z_{n,i}]=0$ for large enough $n$.
    Thus, for large enough $n$,
    $\theta_{n}^{\causample}=(\tilde{\Omega}_{n}^{XX})^{-1}\tilde{\Omega}_{n}^{XY}$ a.s. and
    $\theta_{n}^{\causal}=(\Omega_{n}^{XX})^{-1}\Omega_{n}^{XY}$.
    Without loss of generality, assume that the first element of $T_{n,i}$ depends on $R_{n,i}D^*_{n,i}$.
    By \cref{asm:exposure} (i) and (iii), we can treat the first element of $\theta_{n,(1)}^{\causample}$ and  $\theta_{n,(1)}^{\causample}$ the other elements separately as $\theta_{n,(1)}^{\causal}=(\Omega_{n,(1,1)}^{XX})^{-1}\Omega_{n,(1,1)}^{XY}$, 
    $\theta_{n,(-1)}^{\causal}=(\Omega_{n,(-1,-1)}^{XX})^{-1}\Omega_{n,(-1,-1)}^{XY}$, 
    $\theta_{n,(1)}^{\causample}=(\tilde{\Omega}_{n,(1,1)}^{XX})^{-1}\tilde{\Omega}_{n,(1,1)}^{XY}$, and 
    $\theta_{n,(-1)}^{\causample}=(\tilde{\Omega}_{n,(-1,-1)}^{XX})^{-1}\tilde{\Omega}_{n,(-1,-1)}^{XY}$, 
    where $\Omega_{n,(1,1)}$ is the $(1,1)$ element of $\Omega_{n}$ and $\Omega_{n,(-1,-1)}$ is the submatrix of $\Omega_{n}$ except for its first row and first column. $\tilde{\Omega}_{n,(1,1)}$ and $\tilde{\Omega}_{n,(-1,-1)}$ are defined analogously.
    By  \cref{lem:LLN_W}, 
    \begin{align*}
        \theta_{n,(-1)}^{\causample}-\theta_{n,(-1)}^{\causal}=(\tilde{\Omega}_{n,(-1,-1)}^{XX})^{-1}\tilde{\Omega}_{n,(-1,-1)}^{XY}-(\Omega_{n,(-1,-1)}^{XX})^{-1}\Omega_{n,(-1,-1)}^{XY}\overset{p}{\to}0.
    \end{align*}
    By \cref{lem:LLN_RW}, 
    \begin{align*}
        \theta_{n,(1)}^{\causample}-\theta_{n,(1)}^{\causal}=(\tilde{\Omega}_{n,(1,1)}^{XX})^{-1}\tilde{\Omega}_{n,(1,1)}^{XY}-((1/\rho_n)\Omega_{n,(1,1)}^{XX})^{-1}(1/\rho_n)\Omega_{n,(1,1)}^{XY}\overset{p}{\to}0.
    \end{align*}
    We can conclude by stacking them.
\end{proof}

\subsection{Proof of \cref{thm:normality}}
\begin{proof}
    We have $\tilde{\Omega}_n^{XZ}\overset{a.s.}{\longrightarrow}0$ and $(n\rho_n)/N\overset{a.s.}{\longrightarrow}1$ under the invertibility and the moment conditions.
    Thus,
    \begin{align*}
        &\sqrt{n\rho_n}\left(
        \begin{array}{c}
            \hat{\theta}_{n}-\theta_n^{\causample}\\
            \hat{\gamma}_{n}-\gamma_n^{\causample}
        \end{array}
        \right)\\
        =&
        \left(
        \begin{array}{cc}
            \tilde{Q}_{n}^{XX}&\tilde{Q}_{n}^{XZ}\\
            \tilde{Q}_{n}^{ZX}&\tilde{Q}_{n}^{ZZ}
        \end{array}\right)^{-1}
        \left(
        \begin{array}{c}
            \frac{\sqrt{n\rho_n}}{N}\sum_{i=1}^{n}R_{n,i}\tilde{X}_{n,i}\tilde{\varepsilon}_{n,i}\\
            \frac{\sqrt{n\rho_n}}{N}\sum_{i=1}^{n}R_{n,i}\tilde{Z}_{n,i}\tilde{\varepsilon}_{n,i}
        \end{array}
        \right)\\
        =&\left[\left(
    \begin{array}{cc}
        \tilde{Q}_n^{XX} & O \\
        O & \tilde{Q}_n^{ZZ}
    \end{array}\right)^{-1}+o_{p^R}(1)\right]\left(
    \begin{array}{c}
            (1+o_{p^R}(1))\frac{1}{\sqrt{n\rho_n}}\sum_{i=1}^{n}R_{n,i}\tilde{X}_{n,i}\tilde{\varepsilon}_{n,i}\\
            (1+o_{p^R}(1))\frac{1}{\sqrt{n\rho_n}}\sum_{i=1}^{n}R_{n,i}\tilde{Z}_{n,i}\tilde{\varepsilon}_{n,i}
        \end{array}
        \right),
    \end{align*}
    and it suffices to show\footnote{
            A random variable $X_n$ is $O_{p^R}(1)$ if for any $\varepsilon>0$, there exist some constant $M_\varepsilon<\infty$ such that
            \begin{align*}
                \mathbb{P}\left(\|X_n\|>M_\varepsilon\mid\bR_n\right)<\varepsilon\quad\text{a.s.}
            \end{align*}
            for large enough $n$.
        }
    \begin{align}
        \frac{1}{\sqrt{n\rho_n}}\sum_{i=1}^n R_{n,i}\tilde{X}_{n,i}\tilde{\varepsilon}_{n,i}&=O_{p^R}(1),\label{eq:tilde_X_eps}\\
        \frac{1}{\sqrt{n\rho_n}}\sum_{i=1}^{n}R_{n,i}\tilde{Z}_{n,i}\tilde{\varepsilon}_{n,i}&=O_{p^R}(1),\label{eq:tilde_Z_eps}\\
        \frac{1}{\sqrt{n\rho_n}}\tilde{\Sigma}_{n}^{-1/2}&=O_{\text{a.s.}}(1)\label{eq:tilde_Sigma_bound}
    \end{align}
    since these conditions imply that 
    \begin{align*}
        &\tilde{\Sigma}_{n}^{-1/2}\tilde{Q}_n^{XX}\left(
            \hat{\theta}_{n}-\theta_n^{\causample}
        \right)\\
        =&\frac{1}{\sqrt{n\rho_n}}\tilde{\Sigma}_{n}^{-1/2}\tilde{Q}_n^{XX}\left(\tilde{Q}_n^{XX}\right)^{-1}\frac{1}{\sqrt{n\rho_n}}\sum_{i=1}^{n}R_{n,i}\tilde{X}_{n,i}\tilde{\varepsilon}_{n,i}+o_{p^R}(1),
    \end{align*}
    and we can conclude the convergence in conditional distribution with \cref{lem:CLT}. The dominated convergence theorem and the law of iterated expectations imply the unconditional result.

    We show \cref{eq:tilde_X_eps}-\cref{eq:tilde_Sigma_bound}.
    By Chebyshev's inequality, it suffices to show that its conditional variance is almost surely bounded.
    \begin{align}
        &\Var\left(\frac{1}{\sqrt{n}}\sum_{i=1}^{n}\frac{R_{n,i}}{\sqrt{\rho_n}}\tilde{X}_{n,i}\tilde{\varepsilon}_{n,i}\mid\bR_n\right)
        \nonumber\\
        =&\frac{1}{n}\sum_{i=1}^{n}\Var\left(\frac{R_{n,i}}{\sqrt{\rho_n}}\tilde{X}_{n,i}\tilde{\varepsilon}_{n,i}\mid\bR_n\right)
        +\frac{1}{n}\sum_{i=1}^{n}\sum_{j\in\mathcal{N}_{n}(i,2K)\setminus\{i\}}\Cov\left(\frac{R_{n,i}}{\sqrt{\rho_n}}\tilde{X}_{n,i}\tilde{\varepsilon}_{n,i},\frac{R_{n,j}}{\sqrt{\rho_n}}\tilde{X}_{n,j}\tilde{\varepsilon}_{n,j}\mid\bR_n\right)
        \nonumber\\
        \leq&\frac{1}{n}\sum_{i=1}^{n}\frac{R_{n,i}}{\rho_n}\mathbb{E}\left[\tilde{X}_{n,i}\tilde{X}_{n,i}^{\prime}\tilde{\varepsilon}_{n,i}^2\mid\bR_n\right]
        \label{eq:clt_var}\\
        &+\frac{1}{n}\sum_{i=1}^{n}\sum_{j\in\mathcal{N}_{n}(i,2K)\setminus\{i\}}\frac{R_{n,i}R_{n,j}}{\rho_n}\left(\mathbb{E}\left[\tilde{X}_{n,i}\tilde{X}_{n,j}^{\prime}\tilde{\varepsilon}_{n,i}\tilde{\varepsilon}_{n,j}\mid\bR_n\right]-\mathbb{E}\left[\tilde{X}_{n,i}\tilde{\varepsilon}_{n,i}\mid\bR_n\right]\mathbb{E}\left[\tilde{X}_{n,j}\tilde{\varepsilon}_{n,j}\mid\bR_n\right]'\right).
        \label{eq:clt_cov}
    \end{align}
    Each element of the first term \cref{eq:clt_var} is almost surely bounded by
    \begin{align*}
        \left(\frac{N}{n\rho_n}\right)\cdot\max_{i}\|\tilde{X}_{n,i}\|^2\cdot\max_{i}|\tilde{
        \varepsilon
        }_{n,i}|^2.
    \end{align*}
    Thus, the first term \cref{eq:clt_var} is $O_{\text{a.s.}}(1)$ by \cref{asm:moment,lem:eps_bound,lem:nNconv}.
    The second term \cref{eq:clt_cov} is also $O_{\text{a.s.}}(1)$ by a similar argument as the first term and \cref{asm:sparsity}. Hence, \cref{eq:tilde_X_eps} is $O_{p^R}(1)$.
    Similarly, we can show that \cref{eq:tilde_Z_eps} is $O_{p^R}(1)$. \cref{eq:tilde_Sigma_bound} is also $O_{\text{a.s.}}(1)$ by the invertibility assumption (\cref{asm:moment}). 
\end{proof}

\subsection{Proof of \cref{thm:normality_causal}}
\begin{proof}
    Since $\tilde{X}_{n,i}=X_{n,i}$ and $\tilde{Z}_{n,i}=Z_{n,i}$,
    \begin{align*}
        &\sqrt{n\rho_n}\left(
        \begin{array}{c}
            \hat{\theta}_{n}-\theta_n^{\causal}\\
            \hat{\gamma}_{n}-\gamma_n^{\causal}
        \end{array}
        \right)\\
        =&
        \left(
        \begin{array}{cc}
            \tilde{Q}_{n}^{XX}&\tilde{Q}_{n}^{XZ}\\
            \tilde{Q}_{n}^{ZX}&\tilde{Q}_{n}^{ZZ}
        \end{array}\right)^{-1}
        \left(
        \begin{array}{c}
            \frac{\sqrt{n\rho_n}}{N}\sum_{i=1}^{n}R_{n,i}X_{n,i}(Y_{n,i}-X_{n,i}'\theta_{n}^{\causal}-Z_{n,i}'\gamma_{n}^{\causal})\\
            \frac{\sqrt{n\rho_n}}{N}\sum_{i=1}^{n}R_{n,i}Z_{n,i}(Y_{n,i}-X_{n,i}'\theta_{n}^{\causal}-Z_{n,i}'\gamma_{n}^{\causal})\\
        \end{array}
        \right)\\
        =&\left[\left(
    \begin{array}{cc}
        \tilde{Q}_n^{XX} & O \\
        O & \tilde{Q}_n^{ZZ}
    \end{array}\right)^{-1}+o_{p}(1)\right]\\
    &\qquad\times\left(
    \begin{array}{c}
            (1+o_p(1))\frac{1}{\sqrt{n\rho_n}}\sum_{i=1}^{n}R_{n,i}X_{n,i}\varepsilon_{n,i}\\
            (1+o_p(1))\frac{1}{\sqrt{n\rho_n}}\sum_{i=1}^{n}R_{n,i}Z_{n,i}\varepsilon_{n,i}
        \end{array}
        \right).
    \end{align*}
    By an argument similar to the proof of \cref{thm:normality}, we can show that 
    \begin{align}
        \frac{1}{\sqrt{n\rho_n}}\sum_{i=1}^{n}R_{n,i}X_{n,i}\varepsilon_{n,i}&=O_p(1),
        \label{eq:X_eps}\\
        \frac{1}{\sqrt{n\rho_n}}\sum_{i=1}^{n}R_{n,i}Z_{n,i}\varepsilon_{n,i}&=O_p(1),
        \label{eq:Z_eps}\\
        \frac{1}{\sqrt{n\rho_n}}\Sigma_{n}^{-1/2}&=O_p(1).
        \label{eq:Sigma_bound}
    \end{align}
    Thus, \cref{eq:X_eps,eq:Z_eps,eq:Sigma_bound} imply that 
    \begin{align*}
        \Sigma_{n}^{-1/2}\tilde{Q}_n^{XX}\left(
            \hat{\theta}_{n}-\theta_n^{\causal}
        \right)
        =\frac{1}{\sqrt{n\rho_n}}\Sigma_{n}^{-1/2}\tilde{Q}_n^{XX}\left(\tilde{Q}_n^{XX}\right)^{-1}\frac{1}{\sqrt{n\rho_n}}\sum_{i=1}^{n}R_{n,i}X_{n,i}\varepsilon_{n,i}+o_p(1),
    \end{align*}
    and we can conclude with \cref{lem:CLT_causal}.
\end{proof}

\subsection{Proof of \cref{thm:var_consistency}}
\begin{proof}
    \textbf{[Proof for $\frac{1}{n\rho_n}\tilde{\Sigma}_{n}$]}

    Let 
    \begin{align*}
        \frac{1}{n\rho_n}\tilde{\Sigma}_{n}^{\dagger}
        =\frac{1}{n\rho_n}\sum_{i=1}^{n}\sum_{j\in\tilde{\mathcal{N}}_{n}(i,2K)}R_{n,i}R_{n,j}\left(\tilde{\Psi}_{n,i}-\mathbb{E}\left[\tilde{\Psi}_{n,i}\mid\bR_n\right]\right)\left(\tilde{\Psi}_{n,j}-\mathbb{E}\left[\tilde{\Psi}_{n,j}\mid\bR_n\right]\right)^{\prime}.
    \end{align*}
    Under the boundedness (\cref{asm:moment}) and the local dependence (\cref{asm:local}), \cref{cond:var_consistency} is automatically satisfied.
    Then, \cref{lem:var_consistency} implies that 
    \begin{equation*}
        \frac{1}{n\rho_n}\tilde{\Sigma}_{n}^{\dagger}=\frac{1}{n\rho_n}\tilde{\Sigma}_{n}+o_{p^R}(1).
    \end{equation*}
    Hence, it suffices to show that 
    \begin{align*}
        \frac{1}{N}\hat{\Sigma}_n
        =&\frac{1}{n\rho_n}\tilde{\Sigma}_{n}^{\dagger}+\tilde{B}_{n}+o_{p^R}(1).
    \end{align*}

    Here, $\max_{i}|\hat{\varepsilon}_{n,i}-\tilde{\varepsilon}_{n,i}|=o_p^{R}(1)$ by \cref{asm:moment,thm:consistency,lem:gamma_consistency}. Also,
    $$\frac{1}{n\rho_n}\sum_{i=1}^{n}\sum_{j\in\tilde{\mathcal{N}}_{n}(i,2K)}R_{n,i}R_{n,j}\tilde{X}_{n,i}\tilde{X}_{n,j}^{\prime}\hat{\varepsilon}_{n,i}\hat{\varepsilon}_{n,j}=O_{\text{a.s.}}(1)$$
    by \cref{asm:moment,asm:sparsity}, $\rho\in(0,1]$, and \cref{lem:eps_bound}.
    Thus, we can show that
    \begin{align}
        \frac{1}{N}\hat{\Sigma}_n
        =&\frac{1}{N}\sum_{i=1}^{n}\sum_{j\in\tilde{\mathcal{N}}_{n}(i,2K)}R_{n,i}R_{n,j}\hat{\Psi}_{n,i}\hat{\Psi}_{n,j}^{\prime}\nonumber\\
        =&\frac{1}{N}\sum_{i=1}^{n}\sum_{j\in\tilde{\mathcal{N}}_{n}(i,2K)}R_{n,i}R_{n,j}\tilde{X}_{n,i}\tilde{X}_{n,j}^{\prime}\hat{\varepsilon}_{n,i}\hat{\varepsilon}_{n,j}\nonumber\\
        =&\frac{1}{n\rho_n}\sum_{i=1}^{n}\sum_{j\in\tilde{\mathcal{N}}_{n}(i,2K)}R_{n,i}R_{n,j}\tilde{X}_{n,i}\tilde{X}_{n,j}^{\prime}\tilde{\varepsilon}_{n,i}\tilde{\varepsilon}_{n,j}+o_{p^R}(1),\label{eq:hateps_to_tildeeps}
    \end{align}
    where the last equality holds by \cref{lem:nNconv}.

    Then,
    \begin{align}
        \cref{eq:hateps_to_tildeeps}
        =&\frac{1}{n\rho_n}\sum_{i=1}^{n}\sum_{j\in\tilde{\mathcal{N}}_{n}(i,2K)}R_{n,i}R_{n,j}\tilde{\Psi}_{n,i}\tilde{\Psi}_{n,j}^{\prime}+o_{p^R}(1)\nonumber\\
        =&\frac{1}{n\rho_n}\tilde{\Sigma}_{n}^{\dagger}+\tilde{B}_{n}+o_{p^R}(1)\\
        &+\frac{2}{n\rho_n}\sum_{i=1}^{n}\sum_{j=1}^{n}R_{n,i}R_{n,j}
        \left(\tilde{\Psi}_{n,i}-\mathbb{E}\left[\tilde{\Psi}_{n,i}\mid\bR_n\right]\right)\mathbb{E}\left[\tilde{\Psi}_{n,j}\mid\bR_n\right]^{\prime}\mathds{1}\{\tilde{d}_n(i,j)\leq2K\},\label{eq:tilde_psi_rem}
    \end{align}
    thus, it suffices to show that the remainder term $\cref{eq:tilde_psi_rem}=o_{p^R}(1)$.

    We will show it element-wise. Take the $(k,k^{\prime})$-element of \cref{eq:tilde_psi_rem}. Let 
    \begin{align*}
        \tilde{\varphi}_i=\sum_{j=1}^{n}R_{n,j}\mathbb{E}\left[\tilde{\Psi}_{n,j,(k^{\prime})}\mid\bR_n\right]\mathds{1}\{\tilde{d}_n(i,j)\leq2K\}.
    \end{align*}
    Then,
    \begin{align*}
        &\mathbb{E}\left[\left|(k,k^{\prime})\text{-element of }\cref{eq:tilde_psi_rem}\right|\mid\bR_n\right]\\
        =&
        \mathbb{E}\left[\left|\frac{2}{n\rho_n}\sum_{i=1}^{n}R_{n,i}\left(\tilde{\Psi}_{n,i}-\mathbb{E}\left[\tilde{\Psi}_{n,i,(k)}\mid\bR_n\right]\right)\tilde{\varphi}_i\right|\mid\bR_n\right]\\
        \leq&\mathbb{E}\left[\left(\frac{2}{n\rho_n}\sum_{i=1}^{n}R_{n,i}\left(\tilde{\Psi}_{n,i}-\mathbb{E}\left[\tilde{\Psi}_{n,i,(k)}\mid\bR_n\right]\right)\tilde{\varphi}_i\right)^2\mid\bR_n\right]^{1/2}\\
        \leq&\frac{2}{\rho_n}\left(\frac{1}{n^2}\sum_{i=1}^{n}\Var\left(\tilde{\Psi}_{n,i,(k)}\mid\bR_n\right)\tilde{\varphi}_i^2+\frac{1}{n^2}\sum_{i=1}^{n}\sum_{j\neq i}\left|\Cov\left(\tilde{\Psi}_{n,i,(k)},\tilde{\Psi}_{n,j,(k)}\mid\bR_n\right)\right|\times|\tilde{\varphi}_i\tilde{\varphi}_j|\right)^{1/2},
    \end{align*}
    where the first inequality follows from Jensen's inequality.
    
    By \cref{asm:moment} and \cref{lem:eps_bound}, $\tilde{\Psi}_{n,i,(k)}$ is uniformly bounded, thus $\max_{i}\Var\left(\tilde{\Psi}_{n,i,(k)}\mid\bR_n\right)=O_{\text{a.s.}}(1)$ and $\tilde{\varphi}_i^2\leq C\times(\sum_{j=1}^{n}\mathds{1}\{\tilde{d}_n(i,j)\leq2K\})^2\leq C\times|\mathcal{N}_{n}(i;2K)|^2$ for some constant $C>0$. Hence,
    $\frac{1}{n^2}\sum_{i=1}^{n}\Var\left(\tilde{\Psi}_{n,i,(k)}\mid\bR_n\right)\tilde{\varphi}_i^2\leq C'\delta_n(2K,2)/n$ for some constant $C'>0$.
    By \cref{asm:var_dependence} (i), $\delta_n(2K,2)/n\to0$ as $n\to\infty$.
    
    By \cref{lem:psi_dependency}, $\tilde{\Psi}_{n,i,(k)}$ is conditionally $\psi$-dependent with $\xi_{n,s}=\mathds{1}\{s\leq 2K\}$ given $\bR_n$, thus $\left|\Cov\left(\tilde{\Psi}_{n,i,(k)},\tilde{\Psi}_{n,j,(k)}\mid\bR_n\right)\right|\leq C''\sum_{s=1}^{\infty}\mathds{1}\{s\leq 2K\}\times\mathds{1}\{d_n(i,j)=s\}$ for some constant $C''>0$. Thus,
    \begin{align*}
        &\frac{1}{n^2}\sum_{i=1}^{n}\sum_{j\neq i}\left|\Cov\left(\tilde{\Psi}_{n,i,(k)},\tilde{\Psi}_{n,j,(k)}\mid\bR_n\right)\right|\times|\tilde{\varphi}_i\tilde{\varphi}_j|\\
        \leq&\frac{C''}{n^2}\sum_{s=1}^{2K}\sum_{i=1}^{n}\sum_{j\neq i}\mathds{1}\{d_n(i,j)=s\}\sum_{i'\in\mathcal{N}(i,2K)}\sum_{j'\in\mathcal{N}(j,2K)}1\\
        \leq&\frac{C'''}{n^2}\sum_{s=1}^{2K}|\mathcal{J}_n(s,2K)|
    \end{align*}
    for some constant $C'''>0$.
    By \cref{asm:var_dependence} (ii), $\sum_{s=1}^{2K}\mathcal{J}_n(s,2K)/n^2\to0$ as $n\to\infty$.
    
    Therefore, we have shown that $\mathbb{E}\left[\left|(k,k^{\prime})\text{-element of }\cref{eq:tilde_psi_rem}\right|\mid\bR_n\right]=o_{\text{a.s.}}(1)$.
    By Markov's inequality, we can conclude that the remainder term $\cref{eq:tilde_psi_rem}=o_{p^R}(1)$.
    
    \vspace{20pt}
    \textbf{[Proof for $\frac{1}{n\rho_n}\Sigma_{n}$]}
    Let 
    \begin{align*}
        \frac{1}{n\rho_n}\Sigma_{n}^{\dagger}=\frac{1}{n\rho_n}\sum_{i=1}^{n}\sum_{j\in\mathcal{N}(i,2K)}\left(R_{n,i}\Psi_{n,i}-\rho_n\mathbb{E}[\Psi_{n,i}]\right)\left(R_{n,j}\Psi_{n,j}-\rho_n\mathbb{E}[\Psi_{n,j}]\right)^{\prime}.
    \end{align*}
    We can show that
    \begin{align*}
        \frac{1}{N}\hat{\Sigma}_n
        =&\frac{1}{n\rho_n}\sum_{i=1}^{n}\sum_{j\in\mathcal{N}_{n}(i,2K)}R_{n,i}\Psi_{n,i}R_{n,j}\Psi_{n,j}^{\prime}+o_p(1)\\
        =&\frac{1}{n\rho_n}\Sigma_{n}^{\dagger}+\hat{B}_n+o_p(1)\\
        &+\frac{2}{n\rho_n}\sum_{i=1}^{n}\sum_{j=1}^{n}\left(R_{n,i}\Psi_{n,i}-\rho_n\mathbb{E}[\Psi_{n,i}]\right)\rho_n\mathbb{E}[\Psi_{n,j}]\mathds{1}\{d_n(i,j)\leq2K\}\\
        =&\frac{1}{n\rho_n}\Sigma_{n}^{\dagger}+\hat{B}_n+o_p(1),
    \end{align*}
    where the first equality follows by the similar arguments as we derive \cref{eq:hateps_to_tildeeps} and by \cref{lem:Lambda}, the second equality is a simple transformation, and the last equality holds by the similar arguments for the remainder term $\cref{eq:tilde_psi_rem}$.
    We can conclude by applying \cref{lem:var_consistency} to $(n\rho_n)^{-1}\Sigma_{n}^{\dagger}$.
\end{proof}

\subsection{Proof of \cref{thm:var_consistency_eig}}
\begin{proof}
    Let $\frac{1}{N}\hat{\Sigma}_n^{-}=\frac{1}{N}\sum_{i=1}^{n}\sum_{j=1}^{n}R_{n,i}R_{n,j}\hat{\Psi}_{n,i}\hat{\Psi}_{n,j}^{\prime}\tilde{K}_{n,i,j}^{-}$.
    Since $\tilde{K}_n^{+}=\tilde{K}_n+\tilde{K}_n^{-}$, we have
    \begin{align}
        \label{eq:Sigma_decomp}
        \frac{1}{N}\hat{\Sigma}_n^{+}
        &=\frac{1}{N}\hat{\Sigma}_n+\frac{1}{N}\hat{\Sigma}_n^{-}.
    \end{align}

    \textbf{[Proof for $\frac{1}{n\rho_n}\tilde{\Sigma}_{n}$]}

    \cref{thm:var_consistency} implies 
    \begin{align*}
        &\frac{1}{N}\hat{\Sigma}_n\\
        =&\frac{1}{n\rho_n}\tilde{\Sigma}_n+\tilde{B}_{n}+o_{p^{R}}(1)\\
        =&\frac{1}{n\rho_n}\tilde{\Sigma}_n+\frac{1}{n\rho_n}\sum_{i=1}^{n}\sum_{j=1}^{n}R_{n,i}R_{n,j}\mathbb{E}\left[\tilde{\Psi}_{n,i}\mid\bR_n\right]\mathbb{E}\left[\tilde{\Psi}_{n,j}\mid\bR_n\right]^{\prime}\left(\tilde{K}_{n,i,j}^{+}-\tilde{K}_{n,i,j}^{-}\right)+o_{p^{R}}(1).
    \end{align*}
    By the same logic as in the proof of \cref{thm:var_consistency} after replacing $\mathds{1}\{\tilde{d}_n(i,j)\leq2K\}$ by $\tilde{K}_{n,i,j}^{-}$ and \cref{asm:var_dependence} by \cref{asm:var_dependence_eig}, we can show that
    \begin{align*}
        \frac{1}{N}\hat{\Sigma}_n^{-}
        =&\frac{1}{n\rho_n}\sum_{i=1}^{n}\sum_{j=1}^{n}R_{n,i}R_{n,j}\mathbb{E}\left[\tilde{\Psi}_{n,i}\mid\bR_n\right]\mathbb{E}\left[\tilde{\Psi}_{n,j}\mid\bR_n\right]^{\prime}\tilde{K}_{n,i,j}^{-}\\
        &+\frac{1}{n\rho_n}\sum_{i=1}^{n}\sum_{j=1}^{n}R_{n,i}R_{n,j}\mathbb{E}\left[\left(\tilde{\Psi}_{n,i}-\mathbb{E}\left[\tilde{\Psi}_{n,i}\mid\bR_n\right]\right)\left(\tilde{\Psi}_{n,j}-\mathbb{E}\left[\tilde{\Psi}_{n,j}\mid\bR_n\right]\right)^{\prime}\mid\bR_n\right]\tilde{K}_{n,i,j}^{-}\\
        &+o_{p^{R}}(1).
    \end{align*}
    We get the conclusion by substituting these results into \cref{eq:Sigma_decomp}.

    \textbf{[Proof for $\frac{1}{n\rho_n}\Sigma_{n}$]}

    The proof is similar.
    By the same logic as in the proof of \cref{thm:var_consistency},
    \begin{align*}
        \frac{1}{N}\hat{\Sigma}_n^{-}
        =&\frac{1}{n}\sum_{i=1}^{n}\sum_{j=1}^{n}\rho_n\mathbb{E}\left[\Psi_{n,i}\right]\mathbb{E}\left[\Psi_{n,j}\right]^{\prime}K_{n,i,j}^{-}\\
        &+\frac{1}{n\rho_n}\sum_{i=1}^{n}\sum_{j=1}^{n}\mathbb{E}\left[\left(R_{n,i}\Psi_{n,i}-\rho_n\mathbb{E}\left[\Psi_{n,i}\right]\right)\left(R_{n,j}\Psi_{n,j}-\rho_n\mathbb{E}\left[\Psi_{n,j}\right]\right)^{\prime}\right]K_{n,i,j}^{-}\\
        &+o_p(1).
    \end{align*}
    We get the conclusion by combining it with the result of \cref{thm:var_consistency}.
\end{proof}

\section{Additional simulation results}
\label{app:sim_additional}
We consider the following exposure mapping:
\begin{align*}
    T_{n,i}=\left(R_{n,i}D^*_{n,i},\sum_{j\neq i}A_{n,i,j}R_{n,j}D_{n,j}^{*}\right)
    =:(D_{n,i},\text{net}_{n,i}).
\end{align*}
We set $\tilde{T}_{n,i}=T_{n,i}$. Note that, since $D_{n,i}\indep \text{net}_{n,i}$, no contamination bias would arise. Our focus here is to evaluate our inference procedure based on the asymptotic approximation in this correctly specified model. 

We follow the same implementation procedure as in the simulation exercise in \cref{sec:simulation}, except for the definition of $T_{n,i}$ and $\tilde{T}_{n,i}$, and $\theta_{n,i,(1)}\sim \text{Exponential}(1/3)$ and $\theta_{n,i,(2)}=\frac{\sum_{j\neq i}A_{n,i,j}}{\max_{k}\sum_{j\neq k}A_{n,k,k}}$. Here, the average direct effect is $1/3$ and the average spillover effect is about $2/9$.

In \cref{tab:simulationT}, Panels A and B, we report the results of this simulation when we vary $\rho_{n}$ from $0.1$ to $0.5$ and from $0.6$ to $1.0$, respectively. Since the population size (the number of nodes) is $1770$, the sample size varies from about $177$ to $1770$.
In each panel, the first three rows report the averages of the population and sample-level causal estimands and the OLS estimator. The fourth to sixth rows report the averages of the EHW standard errors and the averages of our proposed standard errors in Equation \cref{eq:se_mod}. The seventh and eighth rows report the average absolute deviations of the estimator from the causal estimands. The last four rows report the coverage probabilities of the $95\%$ confidence intervals constructed using the EHW standard errors and those based on \cref{eq:se_mod} for the two causal estimands.

The first three rows in \cref{tab:simulationT} show that the estimator closely approximates both estimands, as expected from our asymptotic theory (\cref{thm:consistency,thm:causal_consistency}). The difference between $\theta_{n}^{\causal}$ and $\theta_{n}^{\causample}$ is negligible because $T_{n,i}=\tilde{T}_{n,i}$. 
We also observe that while the direct effect estimands $\theta^{\causal}_{(1)}$ and $\theta^{\causample}_{(1)}$ are close to the average direct effect of $1/3$, the spillover effect estimands $\theta^{\causal}_{(2)}$ and $\theta^{\causample}_{(2)}$ are larger than the average spillover effect of $2/9$. 
This occurs because the spillover effect estimands place greater weight on nodes with more connections, who tend to have larger spillover effects, resulting in an upward bias. 
The seventh and eighth rows, showing the average absolute deviations of the estimator from the estimands, also confirm that the estimator closely approximates the estimands, especially as $\rho_{n}$ increases and the sample size becomes larger.

The fourth to sixth rows show that our proposed standard errors based on \cref{eq:se_mod} tend to be larger than the EHW standard errors, especially as $\rho_{n}$ increases. This is because (i) the EHW standard errors do not account for the network dependence structure, and the observed network becomes denser as $\rho_{n}$ increases, and (ii) our standard errors are designed to be conservative, as established in \cref{thm:var_consistency_eig}. When $\rho_{n}$ is small, the difference between the two types of standard errors is less pronounced because (i) the observed network is sparser and the dependence structure is less important, and (ii) the sample-to-population ratio approaches the infinite population case, where the standard model-based inference is valid.
Additionally, we observe that our proposed standard errors based on \cref{eq:se_mod} for $\theta_{n}^{\causal}$ tend to be slightly larger than those for $\theta_{n}^{\causample}$, reflecting the additional adjustment for sampling variation in the former.

The last two rows in \cref{tab:simulationT} show that the coverage rates based on our proposed method \cref{eq:se_mod} are reasonably close to the nominal $95\%$ target. We observe under-coverage for $\theta_{n,(2)}^{\causal}$ and $\theta_{n,(2)}^{\causample}$ when $\rho_{n}$ is small, likely due to the small sample size and limited variation in the $\text{net}$ variable in sparse networks. In contrast, the coverage rates for $\theta_{n,(2)}^{\causal}$ and $\theta_{n,(2)}^{\causample}$ based on the EHW standard errors are substantially below the nominal level as $\rho_{n}$ increases. This is because the EHW standard errors ignore the network dependence structure and finite population bias, which likely leads to over-rejection of the null hypothesis.

Overall, our simulation exercise shows that as long as the model is correctly specified and relevant network information is observed, reliable inference for the causal estimands is possible even when not everyone in the population is sampled. Since exhaustive network collection can be costly in practice, our results provide a rationale for collecting network data based on sampled units, which is less costly.

\begin{table}[htbp]
  \caption{Simulation Results: $T_{n,i}=\tilde{T}_{n,i}$ case}
  \label{tab:simulationT}
  \scalebox{0.9}{
  \begin{threeparttable}
    \centering
    \begingroup
  \begin{tabular}{lcccccccccc}
    \hline
    \hline
    \multicolumn{11}{c}{\textbf{Panel A: $\rho=0.1-0.5$}} \\
   \hline
     & \multicolumn{2}{c}{$0.1$} & \multicolumn{2}{c}{$0.2$} & \multicolumn{2}{c}{$0.3$} & \multicolumn{2}{c}{$0.4$} & \multicolumn{2}{c}{$0.5$} \\
     & D & net & D & net & D & net & D & net & D & net \\
    \hline
    $\theta^{\causal}$ & 0.348 & 0.312 & 0.348 & 0.312 & 0.348 & 0.312 & 0.348 & 0.312 & 0.348 & 0.312 \\
    $\theta^{\causample}$ & 0.346 & 0.311 & 0.349 & 0.311 & 0.349 & 0.312 & 0.349 & 0.311 & 0.349 & 0.312 \\
    $\hat{\theta}$ & 0.347 & 0.285 & 0.350 & 0.305 & 0.348 & 0.308 & 0.352 & 0.310 & 0.350 & 0.305 \\
    $\text{SE EHW}$ & 0.214 & 0.265 & 0.216 & 0.268 & 0.153 & 0.136 & 0.153 & 0.136 & 0.126 & 0.093 \\
    $\text{SE \eqref{eq:se_mod} }\theta^{\causal}$ & 0.214 & 0.263 & 0.215 & 0.265 & 0.156 & 0.146 & 0.156 & 0.145 & 0.132 & 0.109 \\
    $\text{SE \eqref{eq:se_mod} }\theta^{\causample}$ & 0.214 & 0.263 & 0.215 & 0.265 & 0.156 & 0.147 & 0.156 & 0.146 & 0.132 & 0.110 \\
    $|\hat{\theta}-\theta^{\causal}|$ & 0.172 & 0.233 & 0.174 & 0.223 & 0.119 & 0.122 & 0.128 & 0.118 & 0.097 & 0.093 \\
    $|\hat{\theta}-\theta^{\causample}|$ & 0.172 & 0.232 & 0.171 & 0.220 & 0.118 & 0.120 & 0.126 & 0.117 & 0.096 & 0.092 \\
    $\text{Coverage EHW }\theta^{\causal}$ & 0.945 & 0.920 & 0.958 & 0.937 & 0.953 & 0.907 & 0.942 & 0.919 & 0.953 & 0.879 \\
    $\text{Coverage EHW }\theta^{\causample}$ & 0.948 & 0.914 & 0.955 & 0.933 & 0.951 & 0.908 & 0.945 & 0.920 & 0.954 & 0.886 \\
    $\text{Coverage \eqref{eq:se_mod} } \theta^{\causal}$ & 0.941 & 0.904 & 0.953 & 0.926 & 0.953 & 0.924 & 0.944 & 0.937 & 0.963 & 0.931 \\
    $\text{Coverage \eqref{eq:se_mod} } \theta^{\causample}$ & 0.946 & 0.907 & 0.952 & 0.924 & 0.955 & 0.925 & 0.949 & 0.939 & 0.963 & 0.928 \\
    \hline
    \multicolumn{11}{c}{\textbf{Panel B: $\rho=0.6-1.0$}} \\
     \hline
     & \multicolumn{2}{c}{$0.6$} & \multicolumn{2}{c}{$0.7$} & \multicolumn{2}{c}{$0.8$} & \multicolumn{2}{c}{$0.9$} & \multicolumn{2}{c}{$1.0$} \\
     & D & net & D & net & D & net & D & net & D & net \\
    \hline
    $\theta^{\causal}$ & 0.348 & 0.312 & 0.348 & 0.312 & 0.348 & 0.312 & 0.348 & 0.312 & 0.348 & 0.312 \\
    $\theta^{\causample}$ & 0.348 & 0.312 & 0.348 & 0.312 & 0.348 & 0.312 & 0.348 & 0.312 & 0.348 & 0.312 \\
    $\hat{\theta}$ & 0.352 & 0.311 & 0.348 & 0.305 & 0.350 & 0.306 & 0.348 & 0.305 & 0.350 & 0.307 \\
    $\text{SE EHW}$ & 0.126 & 0.092 & 0.110 & 0.071 & 0.110 & 0.071 & 0.099 & 0.058 & 0.100 & 0.058 \\
    $\text{SE \eqref{eq:se_mod} }\theta^{\causal}$ & 0.132 & 0.109 & 0.119 & 0.091 & 0.119 & 0.091 & 0.110 & 0.082 & 0.110 & 0.083 \\
    $\text{SE \eqref{eq:se_mod} }\theta^{\causample}$ & 0.132 & 0.110 & 0.119 & 0.093 & 0.119 & 0.093 & 0.110 & 0.084 & 0.110 & 0.085 \\
    $|\hat{\theta}-\theta^{\causal}|$ & 0.105 & 0.089 & 0.086 & 0.076 & 0.092 & 0.076 & 0.078 & 0.067 & 0.084 & 0.067 \\
    $|\hat{\theta}-\theta^{\causample}|$ & 0.104 & 0.088 & 0.086 & 0.076 & 0.091 & 0.075 & 0.078 & 0.067 & 0.084 & 0.066 \\
    $\text{Coverage EHW }\theta^{\causal}$ & 0.942 & 0.899 & 0.955 & 0.845 & 0.943 & 0.859 & 0.952 & 0.832 & 0.936 & 0.816 \\
    $\text{Coverage EHW }\theta^{\causample}$ & 0.948 & 0.904 & 0.954 & 0.850 & 0.940 & 0.862 & 0.954 & 0.831 & 0.940 & 0.822 \\
    $\text{Coverage \eqref{eq:se_mod} } \theta^{\causal}$ & 0.946 & 0.940 & 0.969 & 0.933 & 0.953 & 0.935 & 0.969 & 0.951 & 0.962 & 0.956 \\
    $\text{Coverage \eqref{eq:se_mod} } \theta^{\causample}$ & 0.953 & 0.949 & 0.970 & 0.936 & 0.959 & 0.940 & 0.970 & 0.950 & 0.966 & 0.959 \\
    \hline
  \end{tabular}
        \begin{tablenotes}
            \footnotesize
            \item {\it Note:} Panel A reports the results for $\rho_{n}=0.1,\ldots,0.5$ and Panel B reports the results for $\rho_{n}=0.6,\ldots,1.0$. The first three rows report the averages of the population and sample-level causal estimands and the OLS estimator. The fourth and fifth rows report the averages of the EHW standard errors and our proposed standard errors based on \cref{eq:se_mod}. The sixth and seventh rows report the average absolute deviations of the estimator from the two causal estimands. The last four rows report the coverage probabilities of the $95\%$ confidence intervals constructed using the EHW standard errors and the standard errors based on our proposed method \cref{eq:se_mod} for the two causal estimands.
        \end{tablenotes}
    \endgroup
    \end{threeparttable}   
    }
\end{table}

\section{Survey of OLS usage in network experiment applications}
\label{app:survey}
In this section, we summarize our survey of the usage of OLS in network experiment applications in economics, as introduced in the second paragraph of the introduction. Our survey provides an overview of the prevalence of OLS in estimating spillover effects in network experiments.

We considered papers published from April 2010 through April 2025 in the following journals: American Economic Review, Econometrica, Quarterly Journal of Economics, Journal of Political Economy, Review of Economic Studies, American Economic Journal: Applied Economics, and Journal of Development Economics. We searched for articles that included both ``networks'' and either ``field experiments'' or ``randomized trial'' as keywords on the Web of Science platform. This search yielded 52 papers, as listed in \cref{tab:survey_summary}. We then reviewed each paper to determine whether it conducted a network experiment and estimated spillover effects using regression. Among these, 29 papers ran regressions to estimate spillover effects; all 29 used the OLS estimator, while only two papers (\citealp{Coutts2022} and \citealp{Fafchamps2013}) mentioned propensity scores or used related estimators.
\input{Table/survey_summary}

\phantomsection
\addcontentsline{toc}{section}{References}
\putbib[list_ref]
\end{bibunit}

\end{document}

%% file: Table/survey_summary.tex
{\small
\begin{longtable}{p{0.3\textwidth}p{0.2\textwidth}p{0.2\textwidth}p{0.2\textwidth}}
\caption{Survey of OLS usage in network experiment applications} \label{tab:survey_summary} \\
\hline
\hline
                     Citation &  Field/Lab Exp w/ Network? & Regression for Causal Effects? & Estimator(s) Used \\
\midrule
\endfirsthead

\toprule
                     Citation &  Field/Lab Exp w/ Network? & Regression? & Estimator(s) Used \\
\midrule
\endhead
\midrule
\multicolumn{4}{r}{{Continued on next page}} \\
\midrule
\endfoot

\bottomrule
\caption*{\footnotesize \textit{Notes}: The first column lists the citation of the paper. The second column indicates whether the paper uses a field or lab experiment with a network structure. The third column indicates whether the paper uses regression to estimate causal effects, and the fourth column lists the specific estimator(s) used in the regression analysis. Methodological papers are marked with ``No (method)'' in the second column and do not have the third and fourth columns filled in.} 
\endlastfoot
        \cite{Evsyukova2024} &                                Yes &                            Yes &                      OLS, Causal Forest \\
           \cite{Batista2025}  &                                 No &                            Yes &                                     OLS \\
            \cite{Karing2024}  &No &                            Yes &                              OLS, Logit \\
           \cite{Chegere2024} &                                Yes &                            Yes &                                     OLS \\
       \cite{Deutschmann2024} &         Yes &                            Yes &                 OLS \\
           \cite{Barsbai2024} &                                 No &                            Yes &                                     OLS \\
          \cite{Banerjee2024}  &  No  &                            Yes &                                 OLS, IV \\
        \cite{Colonnelli2024}  &                                 No &                            Yes &                                OLS, DiD \\
\cite{HernandezAgramonte2024} &                                 No &                            Yes &                          OLS, IPW, Logit \\
          \cite{Borusyak2023} &                                 No &                            Yes &                               OLS, 2SLS \\
          \cite{Banerjee2023} &                                Yes &                            Yes &            OLS \\
           \cite{Soldani2023}  &                                Yes &                            Yes &                                     OLS \\
           \cite{Bobonis2022}  &                                 No &                            Yes &                            OLS, IV \\
              \cite{Alan2022}  &                                 Yes &                            Yes &                         OLS \\
            \cite{Coutts2022}  &                                Yes &                            Yes &                               Propensity score matching,   OLS \\
             \cite{Leung2022} &                   No (method) &              - &                                    - \\
        \cite{Bjorkegren2022} &                                Yes &                            No, Structural &            OLS  \\
            \cite{Beaman2021}  &                                Yes &                            Yes &                                     OLS \\
              \cite{Hess2021}  &                                Yes &                            Yes &             OLS   \\
            \cite{Meghir2022}  &                                No &                            Yes &                   OLS  \\
            \cite{Carter2021}  &                                Yes &                            Yes &            OLS,  \\
             \cite{Hardy2021}  &                                Yes &                            Yes &                                     OLS \\
             \cite{Breza2020}  &                        No (method) &                            - &                         - \\
              \cite{Abel2020}  &                                 No &                            Yes &                                     OLS \\
            \cite{Afridi2020}  &                                Yes &                            Yes &                                     OLS \\
             \cite{Drago2020} &                                Yes &                            Yes &                                     OLS \\
         \cite{BenYishay2020} &                                Yes &                            Yes &                                     OLS \\
               \cite{Cai2020} &                                 No &                            Yes &          OLS, Propensity Score Matching \\
          \cite{Banerjee2019}  &                                Yes &                            Yes &                                     OLS \\
           \cite{Kandpal2019}  &                                No (natural experiment) &                            Yes &        OLS, IV \\
         \cite{Benyishay2019}  &                                Yes &                            Yes &            OLS \\
             \cite{Boltz2019}  &                                Yes &                            Yes &                                     OLS, Logit \\
             \cite{Breza2019}  &                                Yes &                            Yes &                                     OLS \\
             \cite{Flory2018} &                                Yes &                            Yes &                                     OLS \\
     \cite{Chandrasekhar2018}  &                                Yes &                            Yes &                                     OLS \\
               \cite{Cai2018}  &                                Yes &                            Yes &            OLS\\
           \cite{DiFalco2018} &                                Yes &                            Yes &                   OLS  \\
              \cite{Gine2018} &                                 No (cluster) &                            Yes &                                 OLS, IV \\
           \cite{Kessler2017} &                                 No &                            Yes &                      OLS,  \\
              \cite{Cruz2017}  &                                 No &                            Yes &                          OLS, IV,  \\
         \cite{Barnhardt2017} &                                Yes &                            Yes & OLS \\
           \cite{Belloni2017} &                                 No (method) &                            - & - \\
           \cite{Pallais2016}  &                                No &                            Yes &            OLS \\
            \cite{Alatas2016}  &                                Yes &                            Yes &            OLS  \\
        \cite{Nagavarapu2016} &                                No &                            Yes &            OLS  \\
            \cite{Levine2016} &                                 No &                            Yes &                                     OLS \\
           \cite{Jakiela2016} &                                Yes &                            Yes &                                     OLS \\
               \cite{Cai2015}  &                                Yes &                            Yes &                  OLS\\
            \cite{Callen2015}  &                                 No &                            Yes &                                     OLS \\
         \cite{Fafchamps2013}  &                                Yes &                            Yes &                                     OLS, Propensity score matching \\
          \cite{Robinson2012}  &                                 No &                            Yes &                                     OLS \\
         \cite{Godlonton2012} &                                Yes &                            Yes &        OLS \\
\end{longtable}
}